\documentclass[a4paper,11pt,abstracton,DIV=classic]{scrartcl} %
\usepackage[T1]{fontenc}
\usepackage[utf8]{inputenc}
\usepackage{csquotes}
\usepackage[inline]{enumitem}

\usepackage[tt=false,type1=true]{libertine} 	%
\usepackage{amsmath,amsthm,mathtools}	%
\usepackage[libertine]{newtxmath}				%
\usepackage{float}									%
\usepackage{booktabs}								%

\usepackage{bbm}
\usepackage{thmtools}
\usepackage{empheq}

\usepackage{extarrows}

\usepackage{bm}

\usepackage{xcolor}
\colorlet{b75}{black!75!white}
\colorlet{b50}{black!50!white}
\colorlet{b35}{black!35!white}
\colorlet{b20}{black!20!white}
\colorlet{b10}{black!10!white}
\colorlet{b5}{black!5!white}

\usepackage{subcaption}

\usepackage{pgffor}

\usepackage{tikz}
\usetikzlibrary{positioning}
\usetikzlibrary{arrows,arrows.meta}
\usetikzlibrary{patterns}
\usetikzlibrary{tikzmark}
\tikzset{>=stealth'}
\usetikzlibrary{decorations,calligraphy}
\usetikzlibrary{decorations.pathreplacing}%
\usetikzlibrary{decorations.pathmorphing}%

\usetikzlibrary{backgrounds}
\usetikzlibrary{fadings}
\usetikzlibrary{shapes}
\tikzset{inner/.style={draw=none,fill=black,circle,minimum width=4pt,
		inner sep=0pt},
		leaf/.style={draw,rectangle,fill=white,minimum width=3.5pt,minimum
	height=3.5pt,inner sep=0pt},
		hidden/.style={circle,minimum width=4pt,inner sep=0pt}}

\usepackage{hyperref}		%
\usepackage[noabbrev,capitalise]{cleveref}
\creflabelformat{enumi}{(#2#1#3)}
\makeatletter
\DeclareRobustCommand{\crefnosort}[1]{%
  \begingroup\@cref@sortfalse\cref{#1}\endgroup
}
\makeatother

\allowdisplaybreaks

\numberwithin{equation}{section}

\theoremstyle{plain}
\newtheorem{theorem}{Theorem}[section]
\newtheorem{lemma}[theorem]{Lemma}
\newtheorem{fact}[theorem]{Fact}
\newtheorem{proposition}[theorem]{Proposition}
\newtheorem{corollary}[theorem]{Corollary}
\theoremstyle{definition}
\newtheorem{definition}[theorem]{Definition}

\newtheorem*{construction}{Construction}
\theoremstyle{remark}
\newtheorem{remark}[theorem]{Remark}

\newtheorem{example}[theorem]{Example}

\foreach \environment in 
	{theorem,lemma,fact,proposition,corollary,definition,remark,example,%
	observation,assumption,construction} 
{%
	\AtBeginEnvironment{\environment}
	{%
		
		\pushQED{\qed}
	}
	\AtEndEnvironment{\environment}{\popQED}
}

\renewcommand{\thmcontinues}[1]{continued}

\DeclareSymbolFont{sfoperators}{OT1}{cmss}{m}{n}
\SetSymbolFont{sfoperators}{bold}{OT1}{cmss}{b}{n}
\makeatletter
\renewcommand{\operator@font}{\mathgroup\symsfoperators}
\makeatother

\let\phi\varphi
\let\epsilon\varepsilon

\let\cal\mathcal

\let\mathbb\varmathbb
\let\bb\mathbb
 
\let\fr\mathfrak
\let\sf\mathsf

\let\from\colon
\let\with\colon
\let\toto\rightrightarrows

\let\after\circ
\let\tup\vec
\newcommand*{\placeholder}{\bgroup\:\cdot\:\egroup}

\DeclareMathOperator{\dom}{dom}

\renewcommand*{\d}[1]{\mathop{d#1}}
\DeclareMathOperator{\graph}{graph}
\DeclareMathOperator{\diag}{diag}

\newcommand*{\sal}{\textnormal{sal}}

\newcommand*{\pp}{\bgroup\ensuremath{\mkern2mu\rotatebox{90}{$=$}}\egroup}
\DeclarePairedDelimiter{\set}{\{}{\}} %

\DeclarePairedDelimiterX{\dparens}[1]{\lparen}{\rparen}{%
	\nparen{\lparen}{#1}\delimsize\lparen\mathopen{}%
	#1%
	\mathclose{}\delimsize\rparen\nparen{\rparen}{#1}%
}
\DeclarePairedDelimiterX{\dbraces}[1]{\lbrace}{\rbrace}{%
	\nbrace{\lbrace}{#1}\delimsize\lbrace\mathopen{}%
	#1%
	\mathclose{}\delimsize\rbrace\nbrace{\rbrace}{#1}%
}
\newcommand{\dummydelim}[2]{$\left#1\vphantom{#2}\right.$}
\newcommand{\nparen}[2]{\sbox0{\dummydelim{#1}{#2}}\hspace{-0.5\wd0}}
\newcommand{\nbrace}[2]{\sbox0{\dummydelim{#1}{#2}}\hspace{\the\dimexpr -0.85\wd0 + 2pt\relax}}

\makeatletter
\def\bag{\@ifstar\@bag\@@bag}
\def\@bag#1{\dbraces{\smash{#1}}}
\def\@@bag#1{\dbraces*{#1}}
\newcommand{\Bags}[3][\@nil]{%
	\def\tmp{#1}%
	\ifx\tmp\@@nil%
		\dparens*{\smash{\begin{smallmatrix}#2\\#3\end{smallmatrix}}}%
	\else%
		\dparens*{\begin{smallmatrix}#2\\#3\end{smallmatrix}}%
	\fi%
}
\makeatother

\newcommand*{\Var}{\bm{V\mkern-1muar}}

\DeclareMathOperator{\ar}{ar}
\DeclareMathOperator{\pardim}{pardim}

\newcommand*{\IN}{\mathsf{in}}
\newcommand*{\SET}{\mathsf{set}}

\renewcommand*{\c}{\sf{c}}
\renewcommand*{\restriction}{\mathord{\upharpoonright}}

\makeatletter
\foreach \x in {A,...,Z}{%
	\expandafter\xdef \csname \x \endcsname			{ \noexpand\cal{\x} }%
	\expandafter\xdef \csname \x\x \endcsname			{ \noexpand\bb{\x} }%
	\expandafter\xdef \csname \x\x\x \endcsname		{ \noexpand\fr{\x} }%
	\expandafter\xdef \csname \x\x\x\x \endcsname	{ \noexpand\bm{\x} }%
}%
\makeatother

\newcommand*{\powerset}{\mathcal P}

\newcommand*{\Borel}{\mathop{\fr{B}\mkern-.5mu\fr{o}\mkern-.5mu\fr{r}}} %

\newcommand*{\Params}{\Pi}

\DeclarePairedDelimiter{\params}{\langle}{\rangle}
\DeclarePairedDelimiter{\size}{\lvert}{\rvert}

\newcommand*{\GDL}{\ensuremath{\mathrm{GDatalog}}} 
\newcommand*{\GDLOf}[1]{\ensuremath{\GDL[#1]}}
\newcommand*{\EDL}{\ensuremath{\mathrm{Datalog}^{\mathord{\exists}}}}

\newcommand*{\chasestep}[3]{#1 \xlongrightarrow{#2} #3}
\newcommand*{\pchasestep}[3]{#1 \xLongrightarrow{#2} #3}

\DeclareMathOperator{\paths}{paths}
\DeclareMathOperator{\liminst}{lim}
\newcommand*{\Error}{\bot}

\DeclareMathOperator{\FD}{FD}

\DeclareMathOperator{\App}{App}
\DeclareMathOperator{\app}{app}
\DeclareMathOperator{\Ext}{Ext}
\DeclareMathOperator{\ext}{ext}

\newcommand*{\Schema}[1][S]{\cal #1}

\newcommand*{\pdb}{\Delta}

\newcommand*{\whenever}{\mathrel{\mkern4mu\gets\mkern4mu}}

\newcommand*{\stepE}[1][]{\mathrel{\vdash_{#1}}}
\newcommand*{\stepP}[1][]{\mathop{\kappa_{\vdash_{#1}}}}
\newcommand*{\StepE}{\mathrel{\Vdash}}
\newcommand*{\StepP}{\mathop{\kappa_{\Vdash}}}

\newcommand*{\RelName}[1]{\mathbf{#1}}
\newcommand*{\DeclareRelation}[1]{\expandafter\def\csname #1\endcsname{\RelName{#1}}}

\DeclareRelation{Employee}
\DeclareRelation{BasedIn}
\DeclareRelation{PayScale}
\DeclareRelation{PartnerOf}
\DeclareRelation{AffilEmployee}
\DeclareRelation{Res}

\DeclareRelation{City}
\DeclareRelation{Earthquake}
\DeclareRelation{Unit}
\DeclareRelation{House}
\DeclareRelation{Business}
\DeclareRelation{Burglary}
\DeclareRelation{Alarm}
\DeclareRelation{Trig}

\DeclareRelation{Result}

\newcommand*{\Dist}[1][Dist]{\mathsf{#1}}
\newcommand*{\Flip}{\Dist[Flip]}

\newcommand*{\Poisson}{\Dist[Poisson]}
\newcommand*{\Binomial}{\Dist[Binomial]}

\newcommand*{\Gaussian}{\Dist[Gaussian]}

\title{Generative Datalog with Continuous Distributions}

\usepackage{authblk}

\author[1]{{Martin Grohe}}

\author[2]{{Benjamin Lucien Kaminski}}

\author[1]{{Joost-Pieter Katoen}}

\author[1]{{\mbox{Peter Lindner}}}

\affil[1]{{\normalsize \textsf{$\{$grohe,katoen,lindner$\}$@informatik.rwth-aachen.de}}}
\affil[2]{{\normalsize \textsf{kaminski@cs.uni-saarland.de}}}

\affil[1]{{\normalsize RWTH Aachen University, Aachen, Germany}}
\affil[2]{{\normalsize Saarland University, Saarland Informatics Campus, Saarbrücken, Germany, and \mbox{University College London}, London, United Kingdom}}

\newenvironment{acks}{\subsection*{Acknowledgments}}{} %

\begin{document}

\maketitle %

\begin{abstract}
Arguing for the need to combine declarative and probabilistic programming,
Bárány et al.\ (TODS 2017) recently introduced a probabilistic extension of
Datalog as a \enquote{purely declarative probabilistic programming language.} 
We revisit this language and propose a more principled approach towards 
defining its semantics based on stochastic kernels and Markov processes\,---\,standard notions from probability theory. 
This allows us to extend the 
semantics to continuous probability distributions, thereby settling an open 
problem posed by Bárány et al.

We show that our semantics is fairly robust, allowing both parallel execution
and arbitrary chase orders when evaluating a program. We cast our semantics in
the framework of infinite probabilistic databases (Grohe and Lindner, ICDT 
2020), and show that the semantics remains meaningful even when the input of
a probabilistic Datalog program is an arbitrary probabilistic database.
\end{abstract}

\section{Introduction}

Augmenting programming languages with stochastic behavior such as
probabilistic choices or random sampling has a long tradition in computer
science \cite{Saheb-Djahromi1978,Kozen1981}. In recent years, a lot of effort
went into the development of dedicated probabilistic programming languages
(such as, for example, Anglican~\cite{Tolpin+2015}, Church~\cite{Goodman+2008},
Figaro~\cite{Pfeffer2009}, Pyro~\cite{Bingham+2019}, R2~\cite{Nori+2014}, and
Stan~\cite{Carpenter+2017}) that allow the specification and
\enquote{execution}, via probabilistic inference, of sophisticated
probabilistic models.  Such languages are nowadays important tools in a large
variety of applications in different fields like artificial intelligence,
computer vision, and cryptography to name a few~\cite{Goodman2013, Gordon+2014,
vandeMeent+2018}.

From a database perspective, it is desirable to have a declarative
probabilistic programming language that operates on a standard relational data
model. Such a language, \emph{Probabilistic Programming Datalog } (PPDL\footnote{Not to be confused with Kozen's \emph{P}robabilistic \emph{P}ropositional \emph{D}ynamic \emph{L}ogic~\cite{DBLP:conf/stoc/Kozen83}.}) has
recently been introduced by Bárány, ten Cate, Kimelfeld, Olteanu, and Vagena
\cite{Barany+2017}. This language is relational and declarative while still
employing main features of probabilistic programming languages such as random
sampling and conditioning on observations. A PPDL program has a generative and
a constraint part. The generative part is a Datalog program augmented by 
special tuple-generating rules that involve random sampling. The constraint 
part conditions the resulting probability space on, say, typical database 
constraints.

In this work, we focus on the generative part of such a program, which is
referred to as \emph{Generative Datalog}. In a nutshell, a Generative Datalog
program is a standard Datalog program with the addition that it may contain
rules of the shape
\[ 
	R( \tup x ) \gets S_1( \tup x_1 ), \dots, S_m( \tup x_n )
\]
where, as usual, $R,S_1,\dots,S_m$ are relation symbols and the $S_i( \tup x_i
)$ are atoms. The difference is that we allow the tuple $\tup x$ to contain not
only variables and constants, but also references to \emph{parameterized
distributions} as in 
\[
	R( x, y, \Gaussian\params{ x, y } ) \gets S_1(x), S_2(y)\text.
\]
The standard point of view for generative programs is that they are run in an
iterative fashion where in each step an random fact is added to the current
database instance, until no rule is applicable anymore (which is made precise
later). Intuitively, the head $R(x,y,\Gaussian\params{x,y})$ of the above rule
can be understood as a sampling instruction: If the body $S_1(x), S_2(y)$ of
the rule is satisfied for valuations of $a=\mu$ and $b=\sigma^2$ of $x$ and
$y$, then we add a fact $R(a,b,X)$ where $X \sim \Gaussian\params{a,b}$.

The language of Bárány et al. comes with the restriction that it only
allows sampling from \emph{discrete} probability distributions. Continuous
probability distributions are explicitly mentioned as a relevant extension to
the language in \cite{Barany+2017}, starting with the definition of its
semantics. Yet, acknowledging that this requires even a new, and possibly more
challenging definition of the probability spaces in question, their
introduction was left open in \cite{Barany+2017}. We provide such an extension
in this work. The main technical questions concern
\begin{enumerate*} 
	\item a rigorous definition of its semantics and the well-definedness of
		outputs,
	\item the effect of the order of rule execution on the outcome, and 
	\item algorithmic properties.
\end{enumerate*} 
In this paper, we lay the foundations of the extended language by focusing on
(1) and (2).

Probabilistic databases (PDBs) are a formal model for describing uncertainty in
data \cite{Aggarwal+2009,Suciu+2011,vandenBroeckSuciu2017}. Traditionally, they
were limited to probability spaces that consist of a finite number of possible
alternative database instances (called \emph{possible worlds}). Continuous
probability distributions, however, arise naturally in many application
scenarios that involve uncertain data like noisy sensor measurements
\cite{ChengPrabhakar2003,Cheng+2003}. Moreover, for example, a lot of real
world statistical phenomena, especially those that concern aspects of human
behavior, follow normal or lognormal distributions \cite{Limpert+2001}. In
computational security, continuous distributions like the Laplace-distribution
are common \cite{Dwork+2006}.

Unfortunately, generalizing from discrete to continuous distributions usually
comes with substantial mathematical overhead. While several systems
\cite{Agrawal+2009,Jampani+2011,Singh+2008} handle continuous probability
distributions, only recently~\cite{GroheLindner2019,GroheLindner2020}, Grohe
and Lindner proposed a general framework for rigorously dealing with
probabilistic databases over continuous domains. Moreover, they establish basic
properties such as the measurability of relational calculus and Datalog
queries, which in turn allows for formally specifying the semantics of queries
over continuous probabilistic databases. This framework and some of the basic
results, specifically the measurability of relational calculus queries, is also
the foundation for this work. We emphasize that even for \emph{infinite}
probabilistic databases as in \cite{GroheLindner2019,GroheLindner2020}, the
definition of database instances remains unchanged: every database instance is
still a finite set (or, depending on the context, bag) of facts. Instead,
\emph{the probability space may be infinite, meaning that there may be
infinitely many possible worlds}. We note that even though some computations of
Generative Datalog programs may run forever, we are always only interested in
the finite outcomes, and treat infinite computations as errors where no output
is produced.

\subsection{Contributions}\label{s:contrib}

Our main contribution is the introduction of a formal semantics for the
probabilistic Datalog language of Bárány et al. \cite{Barany+2017} that allows
sampling from continuous probability distributions. We focus on the generative
component of this language, called \emph{Generative Datalog} ($\GDL$). We 
define how a $\GDL$ program, on input a single database instance produces an
output \mbox{(sub-)}probability distribution over database instances. With our 
approach (and further extending \cite{Barany+2017}), we also allow 
probabilistic inputs.

Bárány et al. \cite{Barany+2017} introduce the probabilistic semantics of
$\GDL$ as follows.
\begin{itemize}
	\item To every $\GDL$ program, they associate an \emph{existential Datalog
		program} \cite{Cabibbo1998,Cali+2013} from which they construct a chase
		tree. This can be seen as the introduction of a nondeterministic
		semantics in order to model the possible worlds generated by a $\GDL$
		program. 
	\item This nondeterminism is resolved by weighting the paths of the chase 
		tree according to the distributions that symbolically appear in the
		$\GDL$ program.
\end{itemize}
In a nutshell, Bárány et al. construct a discrete time, discrete space Markov
process. We adapt this approach to the continuous setting, and construct a
discrete time, but continuous space Markov process that is rooted in the
framework of standard PDBs \cite{GroheLindner2020}. In order to ensure that the
so-obtained probability space is well-defined, we show that the probabilistic
transitions described by a $\GDL$ program satisfy certain technical conditions
(that is, that they are \emph{stochastic kernels}). 

The main technical result is that the semantics is independent of the choice of
the chase tree and that it is equivalent to the semantics obtained from
parallel execution of all applicable rules at any execution step of a
GDatalog program.

\bigskip

In addition to the extensions sketched before, we propose slight changes to
the semantics introduced by Bárány et al. \cite{Barany+2017} in order to
avoid the following two peculiarities exhibited by the original semantics.
Note that with our changes to the semantics however, we lose the feature of
FO-rewritability. This is addressed in detail in \cref{ssec:origsim} (see
\cref{rem:fo-eq}). 

\begin{example}[Semantic Continuity]\label{exa:behavior1}
	Consider the two $\GDL$ programs shown in \cref{fig:rule_continuity}, and an
	input instance $D_{\IN} = \set{ R( 0 ) }$. A $\GDL$ program has access to a
	set $\Psi$ of parameterized distributions. In this case, $\Psi = \set{ \Flip
	}$, where $\Flip\params{ p }$ denotes the Bernoulli distribution with 
	parameter $p$, i.\,e. the flip of a coin that is biased according to $p$.

	\begin{figure}[H]
		\hfill%
		\begin{tabular}{ l l }
			\toprule
			$\G_0 \colon$
			& $S\big( \Flip\params[\big]{ \frac12 } \big) \whenever R(0)$ \\[1ex]
			& $S\big( \Flip\params[\big]{ \frac12 } \big) \whenever R(0)$ \\
			\bottomrule
		\end{tabular}
		\hfill
		\begin{tabular}{ l l }
			\toprule
			$\G_\epsilon \colon$
			& $S\big( \Flip\params[\big]{ \frac12 } \big) \whenever R(0)$ \\[1ex]
			& $S\big( \Flip\params[\big]{ \frac12 + \epsilon } \big) \whenever 
				R(0)$ \\
			\bottomrule
		\end{tabular}
		\hfill\mbox{}
		\caption{Two $\GDL$ programs $\protect\G_0$ and $\protect\G_\epsilon$ 
		(where $0 < \epsilon < \frac12 $).}
		\label{fig:rule_continuity}
	\end{figure}

	The fact $R( 0 )$ in $D_{ \IN }$ is just a dummy fact that makes all the 
	rules applicable. Under the original semantics\footnote{In the examples of
	this section, we ignore any auxiliary relations that are introduced in the
	semantics. We will later see that this can always be done without
	introducing any problems (see \cref{rem:endproj} in \cref{sec:seqchase}).},
	the program $\G_0$, on input $D_{ \IN }$, generates the following
	probabilistic database:
	\begin{table}[H]
		\centering
		\caption{The probabilistic database generated by $\protect\G_0$.}
		\label{tab:G0}
		\begin{tabular}{ l c c }
			\toprule
			possible world 	
			& $\set{ R( 0 ), S( 0 ) }$		
			& $\set{ R( 0 ), S( 1 ) }$\\
			\midrule
			probability			
			& $\frac12$				
			& $\frac12$\\
			\bottomrule
		\end{tabular}
	\end{table}

	Yet, for any $\epsilon \in \big(0, \frac12\big)$, the program $\G_{\epsilon}$
	generates three possible worlds under the original semantics:
	\begin{table}[H]
		\centering%
		\caption{The probabilistic database generated by $\protect\G_{\epsilon}$.}
		\label{tab:Geps}
		\begin{tabular}{ l c c c }\toprule
			possible world 	
				& $\set{ R( 0 ), S( 0 ) }$ 	
				& $\set{ R( 0 ), S( 1 ) }$	
				& $\set{ R( 0 ), S( 0 ), S( 1 ) }$
				\\\midrule
			probability			
				& $\frac14 - \frac\epsilon2$	
				& $\frac14 + \frac\epsilon2$			
				& $\frac12$
				\\\bottomrule
		\end{tabular}
	\end{table}

	The probabilities intuitively arise because the rules \enquote{fire} 
	independently. In \cite{Barany+2017}, nondeterminism is introduced by 
	allocating new relations to store the outcomes of the random samplings
	associated with $\psi$-terms for $\psi \in \Psi$. These relations are called
	$\smash{\Result_n^{\psi}}$ where $\psi$ is the involved distribution (and 
	$n$ is an additional technical parameter related to the arity of the new
	relation).  The relations $\smash{\Result^{\psi}_n}$ however do not
	distinguish which rule of the original program lead to their introduction.
	Therefore they are, for example, oblivious of multiple occurrences of the
	same probabilistic rule.  Thus, in the probabilistic semantics, both rules
	of $\G_0$ collapse and we arrive at the probabilities shown in \cref{tab:G0}.
	Yet, for every $\epsilon > 0$, this distinction happens in the semantics of
	$\G_{\epsilon}$, yielding the probabilities from \cref{tab:Geps}.

	Contrary to this, we argue that it is more reasonable to have the 
	probabilistic database produced by $\G_\epsilon$ converging to the one
	produced by $\G_0$ as $\epsilon \to 0$. Under the semantics we propose 
	later, both $\G_0$ and $\G_{\epsilon}$ generate the three possible worlds
	$\set{ R( 0 ), S( 0 ) }$, $\set{ R( 0 ), S( 1 ) }$, and $\set{ R( 0 ), 
	S( 0 ), S( 1 ) }$. The probabilities of these worlds obtained from $\G$ are
	$\frac14$, $\frac14$, and $\frac12$, respectively. This coincides
	(pointwise) with the limits of the probabilities obtained from
	$\G_{\epsilon}$ as $\epsilon \to 0$ (see \cref{tab:Geps}).
\end{example}

\begin{example}[Independence from Symbolic Names]\label{exa:behavior2}
	As another example, consider the following program $\G_0'$. The difference
	over $\G_0$ is that in the second rule, $\Flip$ has been replaced with
	$\Flip'$ which is mathematically the same distribution, but with a different
	symbolic name.
	\begin{figure}[H]
		\centering%
		\begin{tabular}{ l l }
			\toprule
			$\G_0' \colon$
			& $S\big( \Flip\params[\big]{ \frac12 } \big) \whenever R(0)$ \\[1ex]
			& $S\big( \Flip'\params[\big]{ \frac12 } \big) \whenever R(0)$ \\
			\bottomrule
		\end{tabular}
		\caption{A $\GDL$ program with two identical, yet differently named
			distributions.}
	\end{figure}
	Under the semantics of \cite{Barany+2017} (and with removing auxiliary 
	relations), $\G_0'$ generates three possible worlds, $\set{ R( 0 ), 
	S( 0 ) }$, $\set{ R( 0 ), S( 1 ) }$, and $\set{ R( 0 ), S( 0 ), S( 1 ) }$, 
	with probabilities $\frac14$, $\frac14$, and $\frac12$. Recall that this
	differs from the result of $\G_0$. The reason is that in \cite{Barany+2017},
	the syntactic name of the parameterized distribution is hard-coded into the
	relations $\smash{\Result_n^{\psi}}$. Under the semantics we propose later,
	the results of $\G_0$ and $\G_0'$ coincide (after removing auxiliary
	relations).
\end{example}

The properties of the original semantics sketched above are not a fundamental
problem for the definition of the semantics, and do not raise any questions
about the correctness of \cite{Barany+2017}. Instead, we argue that these are
two examples, where the original semantics feels unnatural. Consequently, we 
adapt our semantics to resolve such effects.

\subsection{Paper Outline}

After concluding the introduction with a
short survey of related work, we present the most central mathematical 
definitions and background results from measure theory and probabilistic
databases in \cref{sec:preliminaries}. In \cref{sec:generalgdl}, we introduce
the syntax of \GDL{} programs together with the backbone of its semantics---the 
translation into an existential Datalog program. In \cref{sec:seqchase} we 
present our version of a probabilistic chase, generalizing the ideas of 
\cite{Barany+2017}. We show that this notion defines a Markov process over 
database instances. \Cref{sec:parchase} is devoted to establishing similar 
results when a parallel chase procedure is used (which is novel over 
\cite{Barany+2017}). In \cref{sec:semantics}, we discuss various properties of
the semantics. First, we show that no matter which kind of chase procedure is 
used, the probabilistic database that is described by its semantics turns 
out to be the same. We argue that (regarding finite outcomes) our semantics can
simulate the original semantics of \cite{Barany+2017} and discuss the
termination behavior of \GDL{} programs. We briefly revisit the full
\emph{Probabilistic Programming Datalog} (PPDL) language in \cref{sec:ppdl}. We
conclude the work and indicate topics for future research in
\cref{sec:conclusion}.

To ease the accessibility, the longer technical sections of the paper,
\cref{sec:generalgdl,sec:seqchase,sec:parchase}, open with a separate in-depth
outline. For obtaining a better understanding before reading the full paper in
detail, we advise the reader to check out these introductions, along with the
discussions of \cref{ssec:origsim}.

\subsection{Related Work}
Both the probabilistic programming and the probabilistic database community
have developed a variety of models and systems that allow to specify continuous
probability distributions over data. We briefly mention some of these models
and compare them to the scope of this paper. As the original introduction of
PPDL already features a broad survey of related work \cite[Section
7]{Barany+2017} that we leave as is, we just comment on related work regarding
the added support of continuous variables.

\medskip

Basically all probabilistic programming languages support 
continuous distributions, for instance Church~\cite{Goodman+2008},
Anglican~\cite{Tolpin+2015}, and Figaro~\cite{Pfeffer2009}. Such languages are
contrasted by the database-centric nature of PPDL. Conceptually closer to PPDL
are languages studied in \emph{statistical relational artificial intelligence
(StarAI)} \cite{DeRaedt+2016}. A prominent example of such a language is
\emph{ProbLog} \cite{DeRaedt+2007,Fierens+2015}, a probabilistic variant of
Prolog. In ProbLog, standard Prolog rules can be annotated with a probability
value (\enquote{rule-based} uncertainty). The extension Hybrid ProbLog
\cite{Gutmann+2011} allows continuous attribute-level uncertainty in 
rule-heads.

\emph{Markov Logic Networks (MLNs)} \cite{RichardsonDomingos2006} describe
joint distributions of variables based on weighted (\enquote{soft}) first-order
constraints.  Hybrid MLNs \cite{WangDomingos2008} introduce numeric terms and
properties to MLNs, although it is not easy to tell the relationship between
the resulting system and \enquote{pure} continuous attribute-level uncertainty
(see the discussion in \cite{Gutmann+2011}). Infinite MLNs
\cite{SinglaDomingos2007} allow for countably infinitely many variables with
countable domains. Similarly, \emph{Probabilistic Soft Logic
(PSL)}~\cite{Kimmig+2012} is a formalism for specifying joint distributions
with weighted rules, but also \enquote{soft} truth values. PSL rules are
restricted to conjunctive bodies, as encountered in plain Datalog.  As MLNs
build upon Markov networks, \emph{Bayesian Logic
(BLOG)}~\cite{Milch+2005,Milch2006} builds upon Bayesian networks. BLOG is a
programming language supporting continuous distributions and a random (say
Poisson-distributed) number of objects. Its continuous semantics is formally
treated in~\cite{Wu+2018}. In a nutshell, MLNs, PSL, and BLOG provide
first-order templates for specifying graphical
models~\cite{KollerFriedman2009}.  
\medskip

While all formalisms mentioned so far (and, additionally, those discussed in 
\cite{Barany+2017}) share individual features with PPDL, conjoining Datalog
with classical probabilistic programming was novel to \cite{Barany+2017}. Also,
its introduction of attribute-level uncertainty in rule heads differs from
previous probabilistic versions of Datalog that adhere to the rule-based 
uncertainty and possibly prior uncertainty for given ground facts (like the 
probabilistic Datalog languages of \cite{Fuhr1995,Fuhr2000} and 
\cite{Deutch+2010}). A version of \emph{Datalog\textsuperscript{$\pm$}}
\cite{Gottlob+2013,Gottlob+2014} that is used for specifying ontologies
consists of an MLN, and a Datalog program with (among others) tuple-generating
dependencies that may be annotated with events in the probability space. On the
contrary, the event annotations of \emph{JudgeD} \cite{Wanders+2016} use
distinguished variables solely used to introduce dependencies among the rules.
The connection between probabilistic Datalog and database-valued Markov
processes is already noted in \cite{Deutch+2010,Barany+2017}. The \emph{Monte
Carlo Database System (MCDB)} \cite{Jampani+2011} allows for the specification
and querying of probabilistic databases in an SQL-like syntax. Its successor
\emph{SimSQL}~\cite{Cai+2013}, in particular, is a framework that allows the
definition of Markov processes over database instances.

\section{Preliminaries}\label{sec:preliminaries}

\subsection{Foundations from Measure Theory}\label{ssec:mt}

Here we cover the background from measure theory needed for this paper. More
details can be found in standard textbooks on measure theory and Borel sets
\cite{Kallenberg2002,Srivastava1998}. Some well-known key results we use can
moreover be found in \cref{app:meas}.

\subsubsection{Measure Spaces}
A family $\XXX$ of subsets of a set~$\XX$ is called a \emph{$\sigma$-algebra} 
on $\XX$ if it contains $\XX$ and is closed under complements and countable 
unions. A pair $(\XX, \XXX)$ with $\XXX$ being a $\sigma$-algebra on~$\XX$ is
called a \emph{measurable space}. The elements of $\XXX$ are called
\emph{measurable sets} or \emph{events}. 
A function \mbox{$\mu \from \XXX \to \RR_{\geq 0}\cup\set{\infty}$} is a 
\emph{measure} on a measurable space $(\XX, \XXX)$ if $\mu( \emptyset ) = 0$ 
and $\mu(\bigcup_{i\in\NN} \XXXX_i) = \sum_{i\in\NN} \mu(\XXXX_i)$ for any
sequence of pairwise disjoint events $\XXXX_i\in\XXX$ ($i\in\NN$). The value
$\mu(\XX)$ is called the \emph{mass} of~$\mu$. Measures with $\mu(\XX) = 1$ are
called \emph{probability measures}, measures with $\mu(\XX) \leq 1$ are called
\emph{sub-probability measures}.

A triple $(\XX, \XXX, \mu)$ is called a \emph{measure space} if $(\XX, \XXX)$
is a measurable space and $\mu$ is a measure on $(\XX, \XXX)$. If $\mu$ is a
(sub-)probability measure, then $(\XX,\XXX,\mu)$ is called a 
\emph{(sub-)probability space}. The measure space $(\XX, \XXX, \mu)$ (or simply
$\mu$) is called $\sigma$-finite, if there exists a partition of $\XX$ into
countably many measurable sets of finite measure. All measures that
appear in this paper are $\sigma$-finite.

\medskip

In this paper, we need four standard constructions of measurable spaces.
\begin{enumerate}
	\item \emph{Generated $\sigma$-algebra.}
		If $\XX$ is a non-empty set and $\GGG \subseteq \powerset( \XX )$, then
		$\sigma(\GGG)$ denotes the (unique) inclusionwise smallest
			$\sigma$-algebra on $\XX$ containing $\GGG$. We say $\sigma(\GGG)$ is
			\emph{generated} by $\GGG$.
		\item \emph{Trace $\sigma$-algebra.}
			If $(\XX, \XXX)$ is a measurable space and $\XXXX\subseteq
			\XX$, then $\XXX \restriction_{\XXXX} \coloneqq \set{ \XXXX' \cap \XXXX
			\colon \XXXX' \in \XXX }$ is a $\sigma$-algebra on $\XX \cap \XXXX$,
			called the \emph{trace $\sigma$-algebra} of $\XXXX$. If $\XXXX \in \XXX$,
			then $\XXX\restriction_{\XXXX} \subseteq \XXX$.
		\item \emph{Disjoint union $\sigma$-algebra.}
			Let $(\XX,\XXX)$ and $(\YY,\YYY)$ be measurable spaces with $\XX\cap\YY =
			\emptyset$. Then the family $\XXX\oplus\YYY \coloneqq \set{ \ZZZZ \subseteq \XX 
			\cup \YY \with \ZZZZ \cap \XX \in \XXX \text{ and } \ZZZZ \cap \YY \in 
			\YYY }$ is a $\sigma$-algebra on~$\XX\uplus\YY$. This easily generalizes
			to $\bigoplus_{i\in I} \XXX_i$ for any finite index set $I$.
		\item 
			Let $(\XX_i,\XXX_i)$, with $i\in I$ for some index set $I$, be a
			collection of measurable spaces and let $\XX \coloneqq \prod_{i\in I}
			\XX_i$. The \emph{product $\sigma$-algebra} $\bigotimes_{i\in I} \XXX_i$
			is the coarsest $\sigma$-algebra on $\XX$ that makes all canonical
			projections $\pi_i \from \XX \to \XX_i \with (x_i)_{i\in I}\mapsto x_i$
			measurable. If $I$ is countable, then $\bigotimes_{i\in I} \XXX_i$ is
			generated by the family of \emph{measurable rectangles} $\prod_{i\in I}
			\XXXX_i$ with $\XXXX_i \in \XXX_i$. If $I = \set{1,\dots,n}$ we write
			$\bigotimes_{i=1}^n \XXX_i$ or $\XXX_1\otimes \dots\otimes\XXX_n$ for the
			product $\sigma$-algebra. If all $(\XX_i,\XXX_i)$ are equal, we write
			$\XXX^{\otimes n}$. If $I = \NN$, and all $(\XX_i,\XXX_i)$ are equal,
			we write $\XXX^{\otimes\omega}$ for $\bigotimes_{i = 0}^{\infty}
			\XXX_i$.
	\end{enumerate}

	\medskip

	Measure theory is closely tied to notions from topology. A \emph{topological
	space} is a pair~$(\XX,\TTT)$ where $\XX$ is a set and $\TTT$ is a family of
	subsets of $\XX$, called the \emph{open sets}, such that $\TTT$ contains both
	$\XX$ and $\emptyset$ and is closed under \emph{finite} intersections and
	\emph{arbitrary} unions.  The $\sigma$-algebra on a topological space
	$(\XX,\TTT)$ that is generated by the open sets is called the \emph{Borel
	$\sigma$-algebra} on~$(\XX,\TTT)$ (resp. on $\XX$ if $\TTT$ is understood from
	context).  We denote the Borel $\sigma$-algebra on $\XX$ by~$\Borel(\XX)$.
	Typical examples are $\Borel(\RR)$ and $\Borel[0,1] \coloneqq
	\Borel([0,1])$.

	In probability theory, one often works with the Borel $\sigma$-algebras
	generated from \emph{Polish} topological spaces, i.e.\ from completely
	metrizable spaces containing a countable dense set. The resulting measurable
	spaces are called \emph{standard Borel spaces}. We do not delve into the
	details here, as all measurable spaces appearing in this paper are standard
	Borel. For further information, especially in the context of probabilistic
	databases, see~\cite{GroheLindner2020}.

	\subsubsection{Measurable Functions and Kernels}

	Let $(\XX,\XXX)$ and $(\YY,\YYY)$ be measurable spaces. A function $f \from \XX
	\to \YY$ is called \emph{$(\XXX, \YYY)$-mea\-sur\-able} (or simply
	\emph{measurable}, if clear from context) if for all $\YYYY \in \YYY$ it holds 
	that $f^{-1}(\YYYY) \coloneqq \set{x \in \XX \colon f(x) \in \YYYY} \in \XXX$.
	The function $f$ is called \emph{bimeasurable} if additionally $f( \XXXX ) \in
	\YYY$ for all $\XXXX \in \XXX$.

	If $f$ is $(\XXX, \YYY)$-measurable and $\mu$ is a measure on $(\XX,\XXX)$,
	then $\mu\after f^{-1}$ is the so-called \emph{push-forward measure} of $\mu$
	along $f$ on~$(\YY,\YYY)$. If $\mu$ is a (sub-)probability measure, so is
	$\mu\after f^{-1}$.

	A function $\kappa \from \XX \times \YYY \to [0,1]$ is called a
	\emph{(sub-)stochastic kernel} from $(\XX, \XXX)$ to $(\YY, \YYY)$ if
	\begin{itemize}
		\item for all $X \in \XX$, the function $\kappa(X, \placeholder)\colon \YYY\to[0,1]$ is a 
			(sub-)probability measure on $(\YY,\YYY)$, and
		\item for all $\YYYY \in \YYY$, $\kappa(\placeholder, \YYYY)\colon
			\XX\to[0,1]$ is $(\XXX, \Borel[0,1])$-measurable.
	\end{itemize}
	For every measurable space $(\XX,\XXX)$, the function $\iota \from \XX \times
	\XXX \to [0,1]$ with $\iota(x,\XXXX) = 1$ if $x \in \XXXX$ and $\iota(x,\XXXX)
	= 0$ if $x \notin \XXXX$ is a stochastic kernel from $(\XX,\XXX)$ to itself,
	called the \emph{identity kernel} on $(\XX,\XXX)$.

	\subsubsection{Graphs, Sections and Product Measures}
	First we note that any countable product of standard Borel spaces (with the
	product $\sigma$-algebra) is standard Borel again \cite[Proposition
	3.1.23]{Srivastava1998}. If $f \from \XX \to \YY$ is measurable with
	$(\XX,\XXX)$ and $(\YY,\YYY)$ standard Borel, then the \emph{graph of $f$},
	defined by
	\begin{equation*}
		\graph(f) \coloneqq \set{ (x,f(x)) \with x \in \XX }\text,
	\end{equation*}
	is measurable in $\XXX\otimes\YYY$ \cite[Proposition 3.1.21 and 2.1.9]
	{Srivastava1998}.

	If $\ZZZZ \subseteq \XX\times\YY$ and $x\in\XX$, then $\ZZZZ_x \coloneqq 
	\set{ y \in \YY \with (x,y) \in \ZZZZ }$ is called the
	\emph{$x$-section} of $\ZZZZ$. If $\ZZZZ \in \XXX\otimes\YYY$, then $\ZZZZ_x
	\in \YYY$. Symmetrically, for any $y$-section $\ZZZZ_y$ of $\ZZZZ$, we have
	$\ZZZZ_y\in\XXX$.

	If $(\XX,\XXX,\mu)$ and $(\YY,\YYY,\nu)$ are measure spaces with $\mu$ and 
	$\nu$ $\sigma$-finite, then there exists a unique \emph{product measure}
	$\mu\otimes\nu$ of $\mu$ and $\nu$ on $(\XX\times\YY, \XXX\otimes\YYY)$ with
	the property that $(\mu\otimes\nu)(\XXXX\times\YYYY) = \mu(\XXXX) \cdot 
	\nu(\YYYY)$ for all $\XXXX\in\XXX$ and $\YYYY\in\YYY$. This can be extended to
	any finite (nonempty) product of measures \cite[cf. Theorem 1.27 and p.~15]
	{Kallenberg2002}. We use the notation $\bigotimes_{i=1}^n \mu_i$ and 
	$\mu^{\otimes n}$ analogous to the one for product $\sigma$-algebras.

	Fubini's Theorem intuitively states that integration in a product space can be
	carried out in an arbitrary order.
	\begin{fact}[Fubini's Theorem, cf. {\cite[Theorem 1.27]{Kallenberg2002}}]\label{fac:fubini}
		Let $(\XX,\XXX)$ and $(\YY,\YYY)$ be measurable spaces and $\mu$ be a
		$\sigma$-finite measure on $(\XX,\XXX)$. Then for all measurable $f \from 
		\XX \times \YY \to \RR_{\geq 0}$, it holds that
		\begin{equation*}
			\int_{\XX\times\YY} f \:\d{(\mu\otimes\nu)} = 
			\int_{\XX} \Big(\int_{\YY} f \:\d\nu\Big) \d\mu=
			\int_{\YY} \Big(\int_{\XX} f \:\d\mu\Big) \d\nu\text.\qedhere
		\end{equation*}
	\end{fact}

	\subsubsection{Multifunctions and Selections}\label{sssec:mf}
	Let $(\XX, \XXX)$ and $(\YY,\YYY)$ be measurable spaces where $(\YY,\YYY)$ is
	standard Borel (with fixed Polish topology $\TTT_{\YY}$). A function $M \from 
	\XX \to \powerset( \YY ) \setminus \emptyset$ is called a \emph{multifunction},
	and is denoted $M \from \XX \toto \YY$. A mul{\-}ti{\-}func{\-}tion $M \from
	\XX \toto \YY$ is called
	\begin{itemize}
		\item \emph{closed-valued}, if for every $x\in\XX$, $M(x)\subseteq\YY$ is
			closed w.\,r.\,t. $\TTT_{\YY}$, and
		\item $\XXX$-\emph{measurable}, if $M^{-1}(\YYYY) \coloneqq \set{x \in \XX 
			\colon M(x)	\cap \YYYY \neq \emptyset} \in \XXX$ for every open set 
			$\YYYY\in\TTT_{\YY}$.
	\end{itemize}
	Similarly to the corresponding statement for measurable functions, if $M \from
	\XX\toto \YY$ is a closed-valued measurable multifunction, then
	\begin{equation*}
		\graph(M) \coloneqq \set{ (x,y) \with y\in M(x) }
	\end{equation*}
	is a measurable set in $\XXX\otimes\YYY$.

	A \emph{selection} of a multifunction $M$ is a function $s \from \XX \to \YY$ 
	with $s(x) \in M(x)$ for all $x \in \XX$. A well-known result from Kuratowski
	and Ryll-Nardzewski (\cref{fac:KRN}, see \cite{Kuratowski+1965} and 
	\cite[Theorem 5.2.1]{Srivastava1998}) states that for $(\YY,\YYY)$ standard
	Borel, every measurable, closed-valued multifunction $M \from \XX \toto \YY$ 
	has a $(\XXX, \YYY)$-measurable selection.

	\subsubsection{(Discrete-Time) Stochastic Processes}\label{sec:stochproc}
	A \emph{stochastic process in discrete time} is a sequence of random
	variables in some \emph{state space} $(\XX,\XXX)$. Intuitively, a
	(discrete-time) \emph{Markov process} is a stochastic process where
	the distribution in the $(i+1)$th step only depends on the
	distribution of the previous step $i$. By a theorem of Kolmogorov
	(\cref{fac:kolmogorov}), Markov processes in discrete time are
	guaranteed to exist for any initial distribution and any sequence of
	stochastic kernels $\kappa_i \from \XX \times \XXX \to [0,1]$,
	describing the probabilistic transition on the state space from the
	$i$th to the $(i+1)$th step. If $(\XX,\XXX)$ is the state
	space of the process, then $(\XX^{\omega},\XXX^{\otimes\omega})$ is
	its \emph{path space}.

	\subsection{Parameterized Distributions}\label{ssec:paramdist}

	Let $\Pi$ be a non-empty set (of parameters) and let $(\WW, \WWW, \mu)$ be a
	measure space.

	\begin{definition}\label{def:paramdist}
	A \emph{parameterized distribution} with parameter space $\Pi$ and underlying
	space $(\WW, \WWW, \mu)$ is a function $\psi \from \Pi \times \WW
	\to \RR_{ \geq 0 }$ such that for all $p \in \Pi$ it holds that $\psi(p,
	\placeholder)$ is $(\WWW, \Borel( \RR_{\geq 0} ))$-measurable and that
	\begin{equation}\label{eq:paramdist_prob}
		\int_{\WW} \psi(p,\placeholder) \:\d\mu = 1\text.
	\end{equation}
	\end{definition}
	If $\psi$ is a parameterized distribution, we use $\Pi_{\psi}$ to refer to its
	parameter space and $(\WW_\psi, \WWW_\psi, \mu_\psi)$ to refer to its 
	underlying measure space. Moreover, we usually make the parameter in the 
	argument of $\psi$ explicit by writing $\psi\params{ p }( w )$ instead of
	$\psi( p, w )$, and $\psi\params{ p }$ for the function $\psi( p, \placeholder 
	)$.

	The requirement from \labelcref{eq:paramdist_prob} demands that for a
	parameterized distribution $\psi$ and any fixed parameter $p$, the function
	\[
		P_{ \psi\params{ p } } 
		\from \WWW_{\psi} \to [0, 1] 
		\with \WWWW \mapsto \int_{ \WWWW } \psi\params{ p }\:\d{\mu_\psi}
	\]
	is a probability measure. We will always assume that $(\WW, \WWW, \mu)$ is
	either
	\begin{itemize}
		\item a discrete measure space, with $\WWW = \powerset( \WW )$ being the
			powerset $\sigma$-algebra and $\mu$ being the counting measure on
			$(\WW,\WWW)$; or
		\item the Euclidean space $\RR^n$ for some $n \in \NN_{ > 0 }$, equipped 
			with its Lebesgue-measurable sets $\WWW$ and the $n$-dimensional Lebesgue
			measure $\mu$.
	\end{itemize}
	Note that in the first case, the integral from \labelcref{eq:paramdist_prob}
	collapses to the sum $\sum_{ w \in \WW } \psi(p,w)$ and $\psi
	\params{p}$ plays the role of a \emph{probability mass function}.  Accordingly,
	in the second case, $\psi\params{p}$ is a \emph{probability density function}.
	If $\Pi_{\psi}$ is a space of $m$-tuples, then $m = \pardim( \psi )$ is called
	the \emph{parameter dimension} of $\psi$. Typically, if $\pardim( \psi ) = m > 
	1$, then $\Pi_\psi$ is the full Cartesian product of $m$ spaces.
	We refer to parameterized distributions by symbolic names such as $\Binomial$, 
	$\Poisson$ or $\Gaussian$ if they describe the corresponding well-known 
	distributions. Such examples are shown in \cref{tab:dists}.

	\begin{table}[t]
		\centering%
		\caption{Prominent examples of parameterized distributions.}
		\label{tab:dists}
		\begin{tabular}{ c c c c }\toprule
			\begin{tabular}{ @{}c@{} }
				\textbf{Parameterized}\\
				\textbf{Distribution} $\psi$
			\end{tabular}
			&
			\begin{tabular}{ @{}c@{} }
				\textbf{Parameter}\\
				\textbf{Space} $\Pi_\psi$
			\end{tabular}
			&
			\begin{tabular}{ @{}c@{} }
				\textbf{Underlying}\\
				\textbf{Space} $\WW_\psi$
			\end{tabular}
			&
			\begin{tabular}{ @{}c@{} }
				\textbf{pmf / pdf} $\psi\params{\tup p}( w )$
			\end{tabular}
			\\\midrule
			$\Flip$
			& $[0,1]$
			& $\set{ 0, 1 }$
			& $\Flip\params{ p }( w ) = 
			\begin{cases}
				p & \text{for } w = 1\text,\\
				1-p & \text{for } w = 0
			\end{cases}$
			\\\addlinespace[1ex]
			$\Binomial$
			& $\NN_{ > 0 } \times [0, 1]$
			& $\NN_{ \geq 0 }$
			& $\displaystyle\Binomial\params{ n, p }( k ) = 
			{\textstyle\binom{n}{k}} p^k (1-p)^{n-k}$
			\\\addlinespace[1.5ex]
			$\Poisson$
			& $\RR_{ > 0 }$
			& $\NN_{ \geq 0 }$
			& $\displaystyle\Poisson\params{ \lambda }( k ) = 
				\frac{ \lambda^k }{ k! } e^{ - \lambda }$
			\\\addlinespace[1.5ex]
			$\Gaussian$
			& $\RR \times \RR_{ > 0 }$
			& $\RR$
			& $\displaystyle\Gaussian\params{\mu,\sigma^2}( x ) = 
				\frac{1}{\sqrt{2 \pi \sigma^2 }} 
				e^{ -\frac{(x-\mu)^2}{\sigma^2}}$
			\\\bottomrule
		\end{tabular}
	\end{table}

	For our work, we need to discuss situations where the parameters themselves are
	random variables. Thus, the following result on measurability with respect to
	parametrizations is central for our work. It is a special case of~\cite[Theorem
	3.2]{GaudardHadwin1989}, tailored to our definition of parameterized
	distributions.  It states that, under suitable technical conditions, the
	probability of a fixed event under a parameterized distribution is a measurable
	function of the parameters.

	\begin{fact}[Gaudard \&{} Hadwin {\cite[Theorem~3.2]{GaudardHadwin1989}}]
		\label{fac:gaudardhadwin}%
		Let $\psi$ be a parameterized distribution such that $\Pi_\psi$ is a Borel
		subset of a Polish space and the following hold.
		\begin{enumerate}
			\item For all $w \in \WW_\psi$, the function $\Pi_\psi \to [0,1]
				\colon \tup p \mapsto \psi\params{ \tup p }( w )$ is continuous.
			\item Every $\tup p_0 \in \Pi_\psi$ has a neighborhood $N( \tup
				p_0 )$ with
				\[
					\int_{\WW_{\psi}}
					\bigg( 
						\sup_{ \tup p \in N( \tup p_0 ) } \psi\params{ \tup p } 
					\bigg)
					\d{\mu_\psi} < \infty\text.
				\]
			\item If $\tup p, \tup q \in \Pi_\psi$ with $\tup p \neq \tup q$, then 
				$P_{ \psi\params{ \tup p } }$ and $P_{ \psi\params{ \tup q } }$ are
				different probability measures.
		\end{enumerate}

		Then for every $\WWWW \in \WWW_\psi$ and $\BBBB \in \Borel[0,1]$, it holds
		that
		\[
			\set[\big]{
				\tup p \in \Pi_{\psi} \with 
				P_{\psi\params{\tup p}}( \WWWW ) \in \BBBB 
			} \in \Borel( \Pi_\psi )\text.
			\qedhere
		\]
	\end{fact}

	Let us first describe the three conditions in words. Condition 1 states that
	for every fixed argument $w$ the density $\psi\params{ \tup p }( w)$ in $w$
	is continuous with respect to $\tup p$. Condition 2 means that for every
	parameter $\tup p$, the supremum of densities parameterized from within a
	neighborhood of $\tup p$ is integrable. Finally, condition 3 states that
	distinct parameters produce different distributions through $\psi$.
	Together, the conditions enforce that the parameterized distribution is
	well-behaved with respect to certain topological properties in the parameter
	space.

	We emphasize the crucial role of \cref{fac:gaudardhadwin} for the results of
	this paper. \emph{We only allow such parameterized distributions in generative
	Datalog programs, that adhere to the technical preconditions of this theorem.}
	The reason for this is that we need its conclusion at a central point in our
	constructions.

	If the underlying space of a parameterized distribution is countable at most
	(and hence, $\mu_\psi$ is the counting measure), then condition 2 is trivially
	fulfilled. If additionally, the parameter space is discrete, then condition 1
	is always satisfied as well. For uncountable parameter spaces, this need not be
	the case. However, one may easily verify that, for example, the Binomial and
	the Poisson distribution meet condition 1 (and 2 and 3, for that matter). The
	Gaussian distribution is a continuous distribution for which
	\cref{fac:gaudardhadwin} is applicable \cite[p.~173]{GaudardHadwin1989}. Thus,
	all the distributions from \cref{tab:dists} are suitable for our later
	application. Moreover, according to \cite[p.~173]{GaudardHadwin1989}, the
	conditions generally apply
	\enquote{to the most common [parameterized] families}. Therefore, they should
	not be considered as too harsh a restriction for our purposes. We want to point
	out a particular caveat though. Notably, the theorem is not applicable for the
	Dirac distribution. This is, because the Dirac distribution is not even a
	parametrized distribution in the sense of \cref{def:paramdist} to begin with.  
	Related to this, we note that it has recently been pointed out
	\cite{AlvianoZamayla2021}, that the class of \enquote{allowed} parameterized
	distributions should also be treated with care in the setting of
	\cite{Barany+2017} (where the sample space is always discrete, but the
	parameter space may be uncountable).

\subsection{Relational Databases}\label{sec:reldb}

We fix a countably infinite set $\bm{Rel}$, and a non-empty set $\UU$. The
elements of $\bm{Rel}$ are called \emph{relation symbols} and $\UU$ is called
the \emph{universe}. We also fix a function $\ar \from \bm{Rel} \to \NN$ and
call $\ar( R )$ the \emph{arity} of $R$. The \emph{attributes} of $R$ are then
the numbers $1, \dots, \ar( R )$. Additionally, we fix a function $\dom$ that
maps every pair $(R,i)$ with $R \in \bm{Rel}$ and $1 \leq i \leq \ar(R)$ to a
non-empty subset of $\UU$, and we write $\dom_i(R)$ instead of $\dom(R,i)$.
Then $\dom_i(R)$ is called the \emph{domain} of the $i$th attribute in $R$. We
define the \emph{domain} of $R$ as
\[
	\dom(R) \coloneqq \prod_{ i = 1 }^{ \ar( R ) } \dom_i( R ) 
	\subseteq \UU^{ \ar(R) }\text.
\]
\emph{Throughout the rest of the paper, we assume that $\bm{Rel}$, $\UU$, $\ar$
and $\dom$ are fixed.}

A \emph{database schema} $\Schema$ is a non-empty, finite subset of $\bm{Rel}$.
A \emph{fact} (or, $\Schema$-fact) is an expression of the shape $R( u_1,
\dots, u_n )$ where $u_i \in \dom_i( R )$ for all $i = 1, \dots, \ar( R )$. The
set of facts with relation symbol $R$ (or, \emph{$R$-facts}) is denoted by
$\FF_R$. The set of all $\Schema$-facts is denoted by $\FF_{\Schema}$. A
\emph{database instance} over $\Schema$ and $\UU$ (or, $\Schema$-instance) is a
finite bag (multiset) of facts from $\FF_{ \Schema }$.

\bigskip

The following example is loosely based upon an example from \cite{Jampani+2011}
and serves as a running example throughout the paper.

\begin{example}[name=Corporate Data,label=exa:running]
	We consider a database that stores data of various companies, with database
	schema $\Schema = \set{ \PartnerOf, \Employee, \PayScale }$. The tuples of
	the relations capture the following information:
	\begin{itemize}
		\item $\PartnerOf(c_1,c_2)$ means that the companies $c_1$ and $c_2$ are 
			contract partners.
		\item $\Employee(s,c,d)$ means that the social security number (SSN) $s$ is associated with an
			employee at the department $d$ of company $c$.
		\item $\PayScale(c,d,\mu)$ means that employees of department $d$ at 
			company $c$ achieve an average annual income of $\mu$ dollars.
	\end{itemize}
	The database is shown in \cref{fig:running} below. 

	\begin{figure}[H]	\centering
		\hfill%
		\begin{tabular}{ c c c } 
			\multicolumn{3}{c}{\textbf{Employee}}					\\ \toprule
			SSN				& Company			& Department		\\ \midrule
			962-00-3472		& F-Corp				& HR					\\
			981-00-8876		& E-Corp				& IT					\\ \bottomrule
		\end{tabular}
		\hfill%
		\begin{tabular}{ c c }
			\multicolumn{2}{c}{\textbf{PartnerOf}}					\\ \toprule
			Company\textunderscore1	& Company\textunderscore2	\\ \midrule
			A-Corp						& F-Corp							\\
			A-Corp						& D-Corp							\\ \bottomrule
		\end{tabular}
		\hfill\mbox{}\\[4ex]%
		\hfill%
		\begin{tabular}{ c c c }
			\multicolumn{3}{c}{\textbf{PayScale}}					\\ \toprule
			Company			& Department		& Average\textunderscore Salary
																				\\ \midrule
			A-Corp			& IT					& \${} 55\,000		\\
			E-Corp			& IT					& \${} 63\,000		\\
			F-Corp			& HR					& \${} 56\,000		\\ \bottomrule
		\end{tabular}
		\hfill\mbox{}%
		\caption{A database with three relations, containing corporate data.}
		\label{fig:running}
	\end{figure}

	For example, we might assume that $\dom_1( \PayScale )$ is the set of
	strings over some alphabet, whereas typically $\dom_3( \PayScale ) = \NN$. 
	For technical reasons, it is convenient to formally let $\dom_3( \PayScale )
	= \RR$ though, in order for it to coincide with the respective parameter
	domain of the parameterized Gaussian distribution.
\end{example}

\subsection{Probabilistic Databases}\label{sec:pdb}

In a nutshell, a probabilistic database (PDB) is a collection of database
instances (in the sense of \cref{sec:reldb}) that is equipped with a
probability measure. Throughout this paper, we use the framework of
\emph{standard probabilistic databases} \cite{GroheLindner2020} (to be
consulted for details). Recall that for a database schema $\Schema$, we let
$\FF_{\Schema}$ denote the set of facts that can be built from $\cal S$.  A
basic assumption for standard PDBs is that all attribute domains are standard
Borel. Then $\FF_{\Schema}$ is standard Borel as well and we denote its (Borel)
$\sigma$-algebra by $\FFF_{\Schema}$. The sample space $\DD$ of a standard PDB
is the set of all database instances over ${\Schema}$, that is, finite bags of
facts from $\FF_{\Schema}$. We drop the subscript $\Schema$, if the schema is
clear. By a generic construction, $\DD$ is equipped with a
$\sigma$-algebra $\DDD$, turning it into a measurable space. The
$\sigma$-algebra~$\DDD$ is generated by the family of \emph{counting events}
$\CCCC(\FFFF,n)$ consisting of those instances that contain exactly $n$ facts from
$\FFFF$, where $\FFFF$ is a measurable set of facts. 

\begin{definition}
	A \emph{standard probabilistic database} is a probability space $\pdb = 
	(\DD, \DDD, P)$ where $(\DD,\DDD)$ is the measure space from the
	construction above.
\end{definition}

As we only work with standard PDBs, we omit the term \enquote{standard}
henceforth.

\begin{fact}[Measurability of Queries \cite{GroheLindner2020}]
  \label{fac:measurableviews}
	Relational algebra and aggregate queries are measurable functions on PDBs.
\end{fact}

The construction of PDBs sketched before inherently uses bag semantics, meaning that the sample space contains instances with duplicates. For the purpose of this paper, we only want to consider set semantics though. This can either be achieved on the side of measures, i.\,e. PDBs with almost surely set-valued instances; or by restricting the sample space to the set $\DD^{\SET}$ of duplicate-free instances from $\DD$. Note that $\DD^{\SET}$ is a measurable subset of $\DD$ and, consequentially, $\DDD^{\SET} \coloneqq \DDD\restriction_{\DD^{\SET}}$ a sub-$\sigma$-algebra of $\DDD$, that is, $\DDD^{\SET} \subseteq \DDD$. Moreover,
$\DDD$ is generated by the family of all set-valued counting events
$\CCCC^{\SET}(\FFFF,n) \coloneqq \CCCC(\FFFF,n)\cap \DD^{\SET}$ (cf.  \cite[p.
83]{Srivastava1998}).%

\begin{proposition}\label{pro:pdbstandardborel}
    For every standard PDB $(\DD,\DDD,P)$, the measurable space
    $(\DD,\DDD)$ is standard Borel, as is its restriction to set instances.
\end{proposition}

\begin{proof}
	This is an instantiation of a known result from point process theory and the
	theory of random measures. We use the notation from \cite{Daley+2008}. For 
	any standard Borel space $(\XX,\XXX)$, the set $\cal N_{\XX}^{\#}$ of
	$\NN\cup\set{\infty}$-valued measures $\mu$ on $(\XX, \XXX)$ with the
	property that $\mu(\XXXX) < \infty$ for all bounded $\XXXX \in \XXX$ is a
	Polish space and its Borel $\sigma$-algebra is generated by the evaluation
	maps
	\begin{equation*}
		\mathsf{eval}_{\XXXX} \from \cal N_{\XX}^{\#} \to \RR \with \mu \mapsto \mu(\XXXX)
	\end{equation*}
	where $\XXXX\in\XXX$ \cite[Proposition 9.1.IV]{Daley+2008}. The subspace
	$\cal N_{\XX} = \mathsf{eval}_{\XX}^{-1}(\NN)$ of measures of $\cal N_{\XX}^
	{\#}$ of finite total mass is a measurable subset of $\cal N_{\XX}^{\#}$
	and thus, a standard Borel space when equipped with the corresponding trace
	$\sigma$-algebra \cite[424G]{Fremlin2013}. It is easy to see that there is
	a Borel isomorphism between our space $(\DD,\DDD)$ and the space $(\cal 
	N_{\XX}, \Borel(\cal N_{\XX}))$. Since $\DD^{\SET}$ is a measurable subset of 
	$\DD$, $(\DD^{\SET},\DDD\restriction_{\DD^{\SET}})$ is standard Borel as
	well.
\end{proof}

With the single exception of the proposition above, all facts we use about
standard PDBs in this paper are shown in \cite{GroheLindner2020}.

\emph{Throughout this paper we will exclusively use set instances and set 
	semantics. To simplify notation, we write $(\DD,\DDD)$ instead of 
	$(\DD^{\SET}, \DDD^{\SET})$ for the measurable space of set instances.}

\begin{definition}[Sub-Probabilistic Databases]
	If $(\DD, \DDD)$ is the measurable space of a standard PDB, and $P$ is a
	sub-probability measure on $(\DD, \DDD)$, then $\D = ( \DD, \DDD, P )$ is
	called a \emph{sub-probabilistic} database.
\end{definition}

Let $\DD_{ \Error } \coloneqq \DD \uplus \set{ \Error }$ and equip this space
with the $\sigma$-algebra $\DDD_{ \Error } \coloneqq \DDD \cup 
\set{ \DDDD \cup \set{ \Error } \with \DDDD \in \DDD }$. There is a natural
one-to-one correspondence between probability measures on $( \DD_{ \Error },
\DDD_{ \Error } )$ and sub-probabilistic databases on the instance space
$(\DD, \DDD)$. A natural interpretation of the \enquote{missing} probability
mass of a sub-probabilistic database is that it describes the probability of an
error event (or the outcome of a draw from the PDB to be undefined). The space
$\DD_\Error$ makes this error event (\enquote{$\Error$})
explicit. Note that with this transformation the results of
\cite{GroheLindner2020} concerning query measurability directly also apply to
sub-probabilistic databases.

\subsection{Logic and Datalog}

We briefly introduce the background from logic that we need throughout the
paper. For details, we refer, for example, to \cite{Abiteboul+1995,Libkin2012}.
Let $\Schema$ be a relational schema and $\UU$ the universe as before. Let
$\Var \neq \emptyset$ be a fixed, countably infinite set of variables. As is
common in database theory, we do not distinguish between constant symbols and
constants from $\UU$.  An \emph{atom} is an expression of the shape $R( \tup u
)$ where $R$ is a relation symbol from $\Schema$ and $\tup u \in ( \Var \cup
\UU)^k$ where $k$ is the arity of $R$. First-order formulas are built from
atoms using $\neg$, $\wedge$, $\vee$, $\forall$ and $\exists$. The \emph{free
variables} of $\phi$ are the variables appearing among its atoms that are not
bound by a quantifier. A formula without free variables is called a
\emph{sentence}, or \emph{Boolean}.  We write $\phi( \tup x )$ to indicate that
$\phi$ has free variables exactly $\tup x$. 

Suppose $\tup x = (x_1, \dots, x_n )$ is a tuple of variables. A
\emph{valuation} of $\tup x$ is a function $\alpha$ mapping every variable
$x_i$ in $\tup x$ to a constant $\alpha( x_i ) = a_i \in \UU$. If $\tup u$ is a
tuple of variables and constants, and $\alpha$ a valuation of the variables in 
$\tup u$, then $\alpha( \tup u )$ denotes the tuple obtained by replacing every
variable $x$ in $\tup u$ with $\alpha( x )$. It is often convenient to identify
a valuation $\alpha$ of variables $x_1, \dots, x_n$ with the tuple $\tup a = (
a_1, \dots, a_n )$. If the free variables of $\phi$ are all contained in the
tuple $\tup x$, we write $\phi(\tup a)$ for the formula that emerges from
$\phi$ by replacing every occurrence of $x_i$ with $a_i$. 

\begin{remark}
	Formally, we consider sorted first-order languages. The set of valuations of
	a single variable $x$ is then given as the intersection of all the attribute
	domains of the positions where $x$ occurs. For simplicity, we assume that
	all positions where a variable $x$ occurs are typed equally.
\end{remark}

The semantics $\models$ of first-order logic are defined in the standard way.

For the introduction of Datalog, we follow \cite[Chapter~12]{Abiteboul+1995},
to which we refer the reader for further details. A \emph{Datalog rule} is a
logical expression of the shape
\begin{equation}\label{eq:dlrule}
	R( \tup x ) \whenever S_1( \tup x_1 ), \dots, S_m( \tup x_m )
\end{equation}
where $R$ and $S_i$ are relation symbols and $\tup x$, $\tup x_i$ are tuples of
variables or constants of the appropriate lengths such that every variable in
the tuple $\tup x$ appears among the variables of some tuple $\tup x_i$. The
head of the rule \labelcref{eq:dlrule} is $R( \tup x )$, and the body is $S_1(
\tup x_1 ), \dots, S_m( \tup x_m )$. A \emph{Datalog program} is a finite set
of Datalog rules.

The relation symbols that only occur in the rule bodies of a program are called
\emph{extensional}. The remaining ones (those appearing at least once in a rule
head) are called \emph{intensional}. The extensional (or intensional)
relation symbols form the \emph{extensional} (\emph{intensional}, resp.)
\emph{schema} of the program. The combined schema consists of both the
extensional and intensional relation symbols.

Under the model-theoretic view, a Datalog program $\mathcal P$ is a conjunction
$\phi$ of first-order sentences, where all variables in every rule are
universally quantified. A \emph{model} of $\mathcal P$ is a database instance
over the combined schema satisfying $\phi$. The input to a Datalog program is a
database instance $D$ over the extensional schema. The \emph{outcome} of
$\mathcal P$ on $D$ is the minimal model of $\mathcal P$ that contains $D$.
Such a minimal model always exists, and it contains no constants beyond those
present in $D$ and $\cal P$. It is a superset of $D$ that only contains
additional facts from the intensional schema.

Foreshadowing a \enquote{generative} point of view, we highlight the equivalent
approach to Datalog semantics through fixpoints. Let again $\mathcal P$ be a
Datalog program and $D$ be a database instance over the combined schema. A fact
$f$ is an \emph{immediate consequence} of $D$ subject to $\mathcal P$, if there
exists a valuation of the variables of a rule such that its body is satisfied
in $D$, and such that $f$ is the valuation of the head of the rule. The map
that sends any database instance to the set of its immediate consequences is a
monotone operator on database instances (with respect to $\subseteq$). Then the
outcome of $\mathcal P$ on $D$ is the minimal fixpoint of the operator. This is
unique, and equivalent to the model-theoretic notion of outcome described
above. This second definition can be interpreted algorithmically: Given an
input instance $D$, at every step in time, we add all immediate consequences to
the current database instance, until no further changes occur. This produces
exactly the outcome of $\mathcal P$ on $D$.

\section{General Generative Datalog}\label{sec:generalgdl}

In this section, we introduce the Generative Datalog language in a version that
allows the use of continuous distributions. After specifying the syntax, the 
main goal of this section is to provide the groundwork for a well-defined
semantics of Generative Datalog programs. 

\subsubsection*{Structure of this section}
First of all, in \cref{ss:syntax}, we establish the syntax of Generative
Datalog. For this, we build upon the original syntax of \cite{Barany+2017} but
implement some slight alterations that on the one hand are tailored to the
technical developments later on, and on the other hand allow us to overcome the
issues discussed in \cref{s:contrib}. In \cref{ss:informal}, we set the scene
for the later introduction of the semantics, by describing the desired workings
of a Generative Datalog program in an abstract, but informal way. The rest of
the section prepares key ingredients for the semantics: the translation of
Generative Datalog programs into \emph{existential} Datalog programs
(\cref{sec:assocedl}), the notion of rule applicability
(\cref{ss:rule-applicability}), the definition and properties of the possible
new instances emerging after letting a rule fire (\cref{ss:followup}), and a
brief discussion of functional dependencies that are induced by the rules of
the translated existential Datalog program (\cref{ss:fd}). We remark that all
of these developments mirror equivalent developments in the original
introduction of Generative Datalog in \cite{Barany+2017}, but come with the
additional need to discuss various measurability properties due to considering
continuous distributions.  Based on this, the concrete semantics will be
introduced and treated in the sections thereafter.

\subsection{Syntax of Generative Datalog}\label{ss:syntax}

We start by introducing the syntax of Generative Datalog programs. We note 
that, already here, there are slight differences to the version of Bárány et 
al.\ that stem from the updated semantics we are going to introduce. At a later
point (\cref{ssec:origsim}), we will come back to the differences to
\cite{Barany+2017}.

We fix two disjoint database schemas $\Schema[E]$, and $\Schema[I]$. The schema
$\Schema[E]$ is called the \emph{extensional} database schema, and $\Schema[I]$
is called the \emph{intensional} database schema. 

Additionally, we fix a set $\Psi$ of (symbolic names of) parameterized
distributions. In order to not worry about which combinations of parameters are
\enquote{legal}, we require that if $\psi$ is a parameterized distribution,
then the parameter space $\Params_\psi$ is a full Cartesian product of $m > 0$ 
spaces where $m$ is the number of parameters in $\psi$. All parameterized
distributions we would typically want to support (including, in particular,
those of \cref{tab:dists}) satisfy this requirement anyway.

\medskip

The \emph{terms} of the \GDL{}-language (over $\Schema[E]$,
$\Schema[I]$ and $\Psi$) are defined as follows:
\begin{enumerate}
\item All variables and all constants from $\Schema[E] \cup \Schema[I]$ are 
	terms. Such terms are called \emph{deterministic terms}.
\item Let $\psi \in \Psi$ be a parameterized distribution with $\pardim(\psi) =
	m$, say with parameter space $\Pi_\psi = \Pi_{\psi,1} \times \dots \times
	\Pi_{\psi,m}$. Then 
	\begin{equation}\label{eq:psiterm}
		\psi\params{ p_1, \dots, p_m }
	\end{equation}
	is a term where $p_i$ is either a variable, or a constant in $\Pi_{\psi,i}$.
	Terms of the shape \labelcref{eq:psiterm} are called \emph{probabilistic
	terms} or \emph{$\Psi$-terms}.
\end{enumerate}

An \emph{atom} is an expression of the form $R( t_1, \dots, t_n )$ where $n$
is the arity of $R \in \Schema[E] \cup \Schema[I]$ and $t_1, \dots, t_n$ are
terms subject to the following restrictions for all $i = 1, \dots, n$:
\begin{itemize}
	\item If $t_i = c$ is a constant, then $c \in \dom_i(R)$.
	\item If $t_i = \psi\params{ \tup p }$, then $R \in \Schema[I]$ and 
		$\WW_\psi \subseteq \dom_i( R )$.
\end{itemize}
If some term $t_i$ is probabilistic, $R( t_1, \dots, t_n )$ is called a
\emph{probabilistic atom}, and a \emph{deterministic atom} otherwise. If $R \in
\Schema[E]$, then $R( t_1, \dots, t_n )$ is called an \emph{$\Schema[E]$-atom}
and otherwise, if $R \in \Schema[I]$, then $R$ is called an
\emph{$\Schema[I]$-atom}. We emphasize that probabilistic terms are only
allowed to occur in $\Schema[I]$-atoms.

\begin{definition}\label{def:gdlrule}
	A \emph{\GDL{}$[\Schema[E],\Schema[I],\Psi]$ rule} $\phi$ is an expression 
	\begin{equation}\label{eq:gdlrule}
		R( t_1, \dots, t_n ) \whenever S_1( t_{11}, \dots, t_{1n_1} ), \dots, S_k( t_{k1}, \dots, t_{kn_k} ) 
	\end{equation}
	such that
	\begin{itemize}
		\item $R$ is intensional with $n = \ar( R )$ and $R( t_1, \dots, t_n )$ 
			is an $\Schema[I]$-atom, possibly with $\Psi$-terms;
		\item $S_1, \dots, S_k$ are relation symbols (extensional or intensional) 
			with $n_i = \ar( S_i )$ and for all $i = 1, \dots, k$, $S_i( t_{i1},
			\dots, t_{ in_i })$ is a deterministic atom; and,
		\item all variables appearing among $t_1, \dots, t_n$ appear in 
			$\set{ t_{ij_i} \with 1 \leq i \leq k \text{ and } 1 \leq j_i \leq n_i }$.
	\end{itemize}
	Moreover, we require that if $t_i = \psi\params{ p_1, \dots, p_m }$, and
	$p_j$ is a variable for $j = 1, \dots, m$, then all attribute positions 
	where the variable $p_j$ reappears on the right-hand side of 
	\labelcref{eq:gdlrule} have the same attribute domain, coinciding with 
	$\Pi_{\psi,j}$.
\end{definition} 

The last requirement in \cref{def:gdlrule} ensures that parameterized
distributions cannot be used with malformed parameters. We denote the rule
$\phi$ from \labelcref{eq:gdlrule} as
\[
	\phi_h( \tup x_h ) \whenever \phi_b( \tup x )
\]
so that $\phi_h$ is the formula (atom) $R( t_1, \dots, t_n )$ and $\tup x_h$ is
the tuple of variables appearing therein (i.\,e. $\phi_h$'s \emph{free 
variables}), and, similarly, $\phi_b$ is the conjunction of atoms on the 
right-hand side of \labelcref{eq:gdlrule}, with free variables $\tup x$. By
definition, all variables of $\tup x_h$ reappear in $\tup x$. As in standard
Datalog terminology, the formula $\phi_h$ is called the \emph{head}, and 
$\phi_b$ is called the \emph{body} of the rule. A rule $\phi$ is called
\emph{deterministic} if its head is a deterministic atom, and 
\emph{probabilistic} otherwise.

\begin{definition}[{$\GDL$ Programs}]
	A \emph{$\GDL[\Schema[E],\Schema[I],\Psi]$-program} is a finite bag
	$\bag{\phi_1, \dots, \phi_k}$ of $\GDL[\Schema[E],\Schema[I],\Psi]$ rules.
\end{definition}

The significance of letting a program be a \emph{bag} rather than a set of rules
is that it is our mechanism of sampling multiple times for the same parameters.
Every copy of a rule is interpreted as a separate sampling instruction. We
expand on this in \cref{ss:informal,ssec:origsim}.

\begin{example}[continues=exa:running]
	We use the database instance shown in \cref{fig:running} (see
	\cref{exa:running}) as an input instance for a \GDL{} program.  Apart from
	the given extensional schema $\Schema[E] \coloneqq \Schema = \set{
	\PartnerOf, \Employee, \PayScale }$, we let $\Schema[I] \coloneqq \set{
	\AffilEmployee, \Res }$. The intended purpose of $\Res$ is to store the
	query result we are interested in, whereas $\AffilEmployee$ is just an
	auxiliary relation. Let $\Psi \coloneqq \set{ \Gaussian }$. Then, for
	example,
	\begin{equation}\label{eq:running_rule}
		\Res( s, c, \Gaussian\params{ \mu, 10\,000 } )
		\whenever \Employee( s, c, d ), \PayScale( c, d, \mu )
	\end{equation}
	is a $\GDLOf{\Schema[E],\Schema[I],\Psi}$ rule. Intuitively,
	\labelcref{eq:running_rule} is an instruction that defines circumstances
	under which we should generate new $\Schema[I]$-facts. In this case, we want
	to sample an income amount for employees, based on the average income at 
	their workplace. The $\Psi$-term $\Gaussian\params{\mu, 10\,000}$ indicates
	that the income amount is Normal distributed, parameterized with the average
	value $\mu$ from the $\PayScale$ table and a constant variance $10\,000$.

	\begin{figure}[H]
		\centering%
		\makebox[0pt]{\begin{tabular}{ l l }
			\toprule
			$\G_{\sal}\colon$ 
			& $\AffilEmployee(s,c_0,d) \whenever \Employee(s,c_0,d)$\\
			& $\AffilEmployee(s,c,d) \whenever \Employee(s,c,d), 
					\AffilEmployee(s',c',d'), \PartnerOf(c,c')$\\
			& $\AffilEmployee(s,c,d) \whenever \Employee(s,c,d), 
					\AffilEmployee(s',c',d'), \PartnerOf(c',c)$\\
			& $\Res(s,c,\Gaussian\params{\mu, 10\,000}) \whenever 
				\AffilEmployee(s,c,d), \PayScale(c,d,\mu)$\\
			\bottomrule
		\end{tabular}}
		\caption{A $\GDL$ program for our running example.}
		\label{fig:running_program}
	\end{figure}

	\Cref{fig:running_program} shows the $\GDLOf{\Schema[E],\Schema[I],
	\Psi}$ program $\G_{\sal}$ with three deterministic rules and one
	probabilistic rule. This program computes tuples $(s,c,i)$, 
	such that $s$ is an employee at company $c$ with an annual income of $i$
	dollars, where $i$ is sampled from a Gaussian distribution like described 
	before. 
\end{example}

\begin{example}
	\Cref{fig:barany_program} shows the running
	example\footnote{The example itself is based on a famous example from
	\cite{Pearl1988}.} of \cite{Barany+2017} in our syntax (cf. 
	\cite[Figure~3, p.~22:8]{Barany+2017}).
	\begin{figure}[H]
		\centering%
		\begin{tabular}{ l l }
		\toprule
		$\G_{\text{burglary}}\colon$ 
		& $\Earthquake(c,\Flip\params{0.1}) \whenever \City(c,r)$\\
		& $\Unit(h,c) \whenever \House(h,c)$\\
		& $\Unit(b,c) \whenever \Business(b,c)$\\
		& $\Burglary(x,c,\Flip\params{r}) \whenever \Unit(x,c), \City(c,r)$\\
		& $\Trig(x,\Flip\params{0.6}) \whenever \Unit(x,c), \Earthquake(c,1)$\\
		& $\Trig(x,\Flip\params{0.9}) \whenever \Burglary(x,c,1)$\\
		& $\Alarm(x) \whenever \Trig(x,1)$\\
		\bottomrule
		\end{tabular}
		\caption{The $\GDL$ program from the running example of
		\cite{Barany+2017} in our syntax.}
		\label{fig:barany_program}
	\end{figure}
	There, $\Schema[E] = \set{ \City, \House, \Business }$ and $\Schema[I] =
	\set{ \Earthquake, \Unit, \Burglary, \Trig, \Alarm }$ with $\Psi = \set{
	 \Flip }$. Initially, we have a database instance containing assignments of
	 cities to regions (in $\City$), and of houses and businesses to cities 
	 (in $\House$ and $\Business$). The program does not further distinguish
	 between houses and businesses, and collects them together as units. With
	 the first rule, we flip a coin, whether city $c$ is struck by an
	 earthquake. Similarly, with the fourth rule, we flip a coin determining 
	 whether unit $x$ in city $c$ is burglarized. We assume that every unit is
	 equipped with an alarm system that triggers when someone is trespassing
	 (but may fail to do so). This is captured by rule number six. Yet, an
	 earthquake may also trigger the security system, but with a lower
	 probability, as modeled by rule number five. Finally, when the system is
	 triggered, it sounds the alarm.
\end{example}

From now on, to simplify notation, we assume that our $\GDL{}$ programs contain 
at most one parameterized distribution per rule. Furthermore, we assume that
the parameterized distribution (if existing) is invoked in the last attribute of 
the relation in the rule head. \emph{That is, we assume rule heads of
probabilistic rules to be of the form 	$R( \tup u , \psi\params{ \tup
  p })$. }

Our proofs generalize to the unrestricted setting. The measurability
discussions that follow are not affected by permutations of the attribute
positions, and, moreover, for the simultaneous usage of two or more
parameterized distributions, the resulting tuples are distributed with their
respective product distribution (cf.~\cref{fac:fubini}).

\subsection{An Informal Semantics}\label{ss:informal}

Before delving into the intricacies of a formal semantics for GDatalog
programs, let us explain an informal operational semantics for GDatalog rules
and programs. We have already given the intuition of applying a rule in
\cref{exa:running}, considering the rule \labelcref{eq:running_rule}:
\begin{equation}
		\Res( s, c, \Gaussian\params{ \mu, 10\,000 } )
		\whenever \Employee( s, c, d ), \PayScale( c, d, \mu )
	\tag{\ref{eq:running_rule}}
\end{equation}
Let $\alpha$ be a valuation of $s, c, d, \mu$ such that $\Employee(s,c,d)$ and
$\PayScale(c,d,\mu)$ exist in the input database. This makes the rule 
applicable for the valuation $\alpha$. Just as in plain Datalog, we 
generate a fact to add to the current database instance. For this, we sample
a random variable $X \sim \Gaussian\params{ \alpha( \mu ), 10\,000 }$, and
generate $\Res\big( \alpha(s), \alpha(c), X \big)$.

We run a GDatalog program ${\mathcal G}$ on a database instance over the
extensional schema similarly to a normal Datalog program. All
intensional relations are initialized to be empty. By repeatedly
applying the rules as described above, the program generates (random)
facts. All rule applications are stochastically independent. We
stipulate that each rule of the program (or more precisely, each
occurrence of each rule---remember a program is a bag of rules where a
rule may occur several times) can only be applied once for every
instantiation of the variables appearing in the head of the rule. 
The computation terminates if no rule is applicable anymore, and the output
consists of the original input, extended by the set of facts generated in the
execution, that is, a database instance over the combined (extensional plus
intensional) schema.\footnote{Defining the output this way enables us to treat
all input, output, and intermediate instances in the same space in our proofs.
Alternatively, one could consider only the generated facts as the output but
this is in the end just a matter of taste.}
Because of the sampling of values in the rule applications, the output is
probabilistic. We interpret it as a probabilistic database. Thus, given a
database instance over the extensional schema, a GDatalog program generates a
probabilistic database over the combined schema. 

However, our informal description of the semantics raises several crucial
questions:
\begin{enumerate}
\item Does the program always terminate?
\item How can we be sure that the output is indeed a well-defined
  probabilistic database?
\item In which order do we apply the rules, and does this make a
  difference? 
\end{enumerate}
The answer to Question (1) is simply \enquote*{no} (in general), as the
following easy example demonstrates.

\begin{example}\label{exa:infcomp}
	Suppose a $\GDL$ program $\G$ contains the rule
	\[
		R( \Gaussian\params{ \mu, 1 } ) \whenever R( \mu )\text,
	\]
	where $R \in \Schema[I]$. Intuitively the program stops, when a value is
	sampled that we have already seen. For this concrete example, this will
	happen with probability 0 though. That is, the program almost surely
	\emph{diverges}. A similar behavior can already be enforced with
	deterministic distributions alone, for example with the rule 
	\[
		R( \Dist[Incr]\params{i} ) \whenever R( i )\text,
	\]
	where $\Dist[Incr]\params{i}$, for parameter $i \in \NN$, is a probability
	mass function on $\NN$ with probability 1 on the outcome $i+1$. This program
	always diverges (provided that the rule fires at all).
\end{example}

The first program of \cref{exa:infcomp} points at another complicating
issue. It may well happen that a program terminates for certain random choices,
but does not for others. We resolve this issue by conditioning the output
probability distribution on termination, or by saying that a $\GDL$ program
only defines a \emph{sub-probabilistic database}, where the probability mass of
the whole space may be smaller than $1$. The \enquote{missing} probability mass
is then the probability of divergence.

It is an open research question to understand termination criteria for GDatalog
programs. The other two questions are main guiding questions of our paper. The
answer to (2) is given by a thorough investigation of the stochastic process
sketched above, and regarding (3), we will indeed see that we don't have to
worry too much about the order of rule applications in the end.

\subsection{Translation to Existential Datalog Programs} \label{sec:assocedl}

As in \cite{Barany+2017}, we first introduce a nondeterministic semantics for
our programs by translating a given $\GDL{}$ program into an \emph{existential
Datalog program} (\emph{$\EDL{}$ program}). This is basically the same
procedure as in \cite[Section~3.2]{Barany+2017} modulo slight changes that are
motivated by the discussions of the previous section.

Intuitively, the probabilistic rules of a $\GDL{}$ program introduce attribute
values that are the result of some random sampling. In contrast, the rules of a
\EDL{} program may introduce attribute values that are subject to
nondeterminism.

A \EDL{} program is a Datalog program that additionally allows rules of the
shape
\[
	\exists y \colon \phi_h( \tup x_h, y ) \whenever \phi_b( \tup x )\text,
\] 
where $\tup x$ again contains all variables of $\tup x_h$. Such rules are
called \emph{existential rules}. 

\begin{construction}[Associated $\EDL{}$ Program]
	Let $\G = \bag{ \phi_1,\dots,\phi_k }$ be a $\GDLOf{ \Schema[E],
	\Schema[I], \Psi }$ program. We construct the \EDL{} program
	$\G^{\exists}$ as follows. For all $i = 1, \dots, k$ do the following:
	\begin{enumerate}
		\item If $\phi_i$ is a deterministic rule of $\G$, then $\phi_i$ is a
			rule of $\G^{\exists}$.
		\item If $\phi_i$ is a probabilistic rule of $\G$, say
			\[
			\phi_i( \tup x ) = 
			\big(
				R( \tup u , \psi\params{ \tup p })
				\whenever \phi_{i,b}( \tup x )
			\big)
			\]
			with $\tup u = ( u_1, \dots, u_{ \ar( R ) - 1 } )$ and $\tup p = (p_1,
			\dots, p_{\pardim(\psi)})$ being tuples of variables or constants
			(such that the variables therein all appear in $\tup x$), then we add
			the following two rules to $\G^{ \exists }$, where $R_i$ is a new,
			distinguished relation symbol of arity $\ar( R ) + \pardim( \psi )$:
			\begin{equation}\label{eq:new_rules}
				\begin{aligned}
				& \exists z \with R_i( \tup u, \tup p, z) \whenever 
					\phi_{i,b}(\tup x)\\
				& R( \tup u , z ) \whenever R_i( \tup u, \tup p, z )
				\end{aligned}
          \end{equation}
	\end{enumerate}
	We call $\G^{\exists}$ the \emph{existential}, or \emph{$\EDL$ version} of
	$\G$. It inherits the extensional schema $\Schema[E]^{ \exists } \coloneqq
	\Schema[E]$ from $\G$. The intensional schema $\Schema[I]^{ \exists }$ of
	$\G^{\exists}$ is obtained from $\Schema[I]$ by adding the new relations
	$R_i$.
\end{construction}

Intuitively, probabilistic rules in the original program $\G$ introduce \emph{two} 
new rules (an existential, and a standard one) in $\G^{\exists}$ in order to 
\enquote{decouple} sampling values using parameterized distributions from
adding facts to the database. The first rule of \eqref{eq:new_rules} carries
the information which valuation, resp. parametrization, is used for the
sampling, and introduces a variable ($z$) storing the sample outcome. With
facts produced by the first rule, the second rule specifies how to assemble the
tuple with the given \enquote{sample outcome}. In particular, note the
parametrization $\tup p$ being projected away. This enables us to sample more
than once with different parametrizations $\vec p$, without altering the rule
applicability in the case where $\vec p$ is not contained in $\tup u$.

\begin{example}[continues=exa:running]
	Reconsider our running example, and the $\GDL{}$ program
	$\G_{\sal}$ from \cref{fig:running_program}. Its
	$\EDL$ version $\G_{\sal}^{\exists}$ is shown in 
	\cref{fig:running_edl} below. 

	\begin{figure}[H]
		\centering%
		\makebox[0pt]{\begin{tabular}{ l l }
			\toprule
			$\G_{\sal}^{\exists}\colon$ 
			& \color{gray} 
			$\AffilEmployee(s,c_0,d) \whenever \Employee(s,c_0,d)$\\
			& \color{gray}
			$\AffilEmployee(s,c,d) \whenever \Employee(s,c,d), 
				\AffilEmployee(s',c',d'), \PartnerOf(c,c')$\\
			& \color{gray}
			$\AffilEmployee(s,c,d) \whenever \Employee(s,c,d), 
				\AffilEmployee(s',c',d'), \PartnerOf(c',c)$\\
			& $\exists z \colon \Res'( s, c, \mu, 10\,000, z ) 
				\whenever \AffilEmployee(s,c,d), \PayScale(c,d,\mu)$\\
			& $\Res(s,c,z) \whenever \Res'(s,c,\mu,10\,000,z)$\\
			\bottomrule
		\end{tabular}}
		\caption{The $\EDL$ program associated with the $\GDL$ program
		$\protect\G_{\sal}$.}
		\label{fig:running_edl}
	\end{figure}

	The rules that are typeset in a lighter shade remain unchanged over the
	original $\GDL{}$ program $\G_{\sal}$. The fourth and fifth rules are the
	new rules introduced by our construction. Therein, $\Res'$
	is the new relation symbol introduced in the translation of the fourth rule
	of the original program. 
\end{example}

In the following, we prepare the introduction of a chase procedure
similar to \cite{Barany+2017}. With the presence of parameterized distributions
with uncountable measure spaces, deriving the probabilistic semantics from
$\G^{\exists}$ is more involved than in \cite{Barany+2017}. In fact, most of
their approach towards the construction of a probability space immediately
breaks down.  There, the authors construct the probability space based on
defining the probability of cylinder sets (to be thought of as an initial
sequence of generated facts) by multiplying their probabilities. This is only
possible, since \cite{Barany+2017} is restricted to discrete distributions.  In
particular, after every finite number of steps, their programs have only
countably many possible program configurations.  In our setting, this might be
a continuum, even after just a single step. Therefore, we already need to
proceed with additional care regarding the applicability of rules, and the
probability distribution induced by a single program step. In particular, we
need more advanced tools from measure theory in order to ensure
well-definedness of these concepts. Nevertheless, the whole approach can be
thought of as a generalization of ideas already present in \cite{Barany+2017}.

In the next section, we start with a rigorous treatment of the applicability of
rules, and the set of result instances after firing a rule.

\begin{remark}
	In the remainder of this paper, we predominantly need the existential 
	version $\G^{\exists}$ of our $\GDL$ program $\G$. Unless explicitly 
	mentioned otherwise, $\phi$ will denote a rule of $\G^{\exists}$ in the 
	following. Note that unlike the original $\GDL$ program which is a bag of
	rules, we can always assume that the constructed program $\G^{\exists}$ is a
	\emph{set} of rules $\set{ \phi_1, \dots, \phi_k }$. The existential rules
	we created are pairwise different anyway, and for every deterministic rule,
	we just retain a single copy (for the semantics of existential Datalog,
	multiple copies have no additional effect).

	For simplicity, we always let $\G^{\exists} = \set{ \phi_1, \dots, \phi_k }$
	from now on. In particular, $k$ is to be understood as the number of rules
	in $\G^{\exists}$. 
\end{remark}

\subsection{Rule Applicability}\label{ss:rule-applicability}

In this subsection, we formalize the notion of rules being enabled to fire in
the execution of the program. So far, we have associated an existential Datalog
program $\G^{\exists}$ with our original program $\G$. Existential Datalog is
already well-established, and has a well-defined semantics \cite{Cali+2013}. So
why do we need to worry about matters of rule applicability in the first place?
There are two issues we need to pay attention to:
\begin{enumerate}
	\item It is not immediately clear, how rule applications adhere to the
		measurable structure of our underlying spaces of database instances.
		Quite naturally, multiple rules might be applicable at once, leaving us
		with the burden to come up with a policy for rule execution.  In order to
		transform the probability measure in a well-defined way, this policy has
		to be measurable (in a sense that will become clear below).
	\item Usually, the choice of policy is not unique. Even if the outcomes of
		the existential program do not depend on the chosen policy, we still need
		to argue that different policies do not produce different probability 
		distributions.
\end{enumerate}
We start by formalizing the relevant notions to tackle the first problem, and
defer the discussion of the second problem to \cref{sec:semantics}.

\medskip

If $\phi = \phi( \tup x )$ is a formula with free variables $\tup x$, then we
let $\VV_{\phi} \subseteq \UU^{ \size{ \tup x } }$ denote the domain of the 
valuations of $\tup x$, which is the Cartesian product of the attribute 
domains. The space $\VV_{\phi}$ is naturally equipped with the corresponding
product $\sigma$-algebra $\VVV_{\phi}$ obtained from the attribute spaces.

\begin{definition}
	 Let $\phi = \phi( \tup x ) \in \G^{\exists}$ and let $\tup u \in 
	 \VV_{\phi}$. Let $D \in \DD$ be a database instance. Then $\phi$ is
	 \emph{applicable for valuation} $\tup u$ if $D \not\models \phi_h(\tup u)$ 
	 and $D \models \phi_b( \tup u )$. The rule $\phi$ is called 
	 \emph{applicable} if such $\tup u$ exists.
\end{definition}

Suppose $\phi$ is a rule of $\G^\exists$ with $\phi_h = \phi_h( \tup x_h )$ and 
$\phi_b = \phi_b( \tup x )$. Recall that $\VV_\phi$ is the space of valuations
of the free variables $\tup x$ in the rule. The \emph{ground space} of $\phi$
(or, space of \emph{head groundings} of $\phi$) is defined as follows.
\begin{itemize}
	\item If $\phi$ is existential, say $\phi_h( \tup x_h ) = \exists z\colon R(
		t_1, \dots, t_{m-1}, z)$ for an $m$-ary relation symbol $R$, then 
		$\VV_\phi^{(h)} \coloneqq \dom_1( R ) \times \dots \times \dom_{m-1}( R 
		)$.
	\item Otherwise, say if $\phi_h( \tup x_h ) = R( t_1, \dots, t_m )$, then
	$\VV_\phi^{(h)} \coloneqq \dom_1( R ) \times \dots \times \dom_{m}( R )$.
\end{itemize}
Again, $\VV_\phi^{(h)}$ is equipped with the corresponding product
$\sigma$-algebra $\VVV_\phi^{(h)}$, which makes $( \VV_\phi^{(h)},
\VVV_\phi^{(h)} )$ standard Borel. We consider the function $\pi_\phi
\from \VV_\phi \to \VV_\phi^{(h)}$ that maps valuations of $\tup x$ to valuations of the
full tuple in the head atom. For example, for the rule $\phi(x,y) = \big(
\exists z\colon R( x, x, 1, z ) \whenever S( x, y ), T( y, x )\big)$, we have 
\begin{align*}
	\VV_\phi &= \dom_1(S) \times \dom_2(S) = \dom_2(T) \times \dom_1(T)\text,\\
	\VV_\phi^{(h)} &= \set[\big]{ (u_1,u_2,u_3) \with u_i \in \dom_i(R)}
	\text{ and}\\
	\pi_\phi( x, y ) &= ( x, x, 1 )\text{ for all }(x,y) \in \VV_\phi\text.
\end{align*}

\begin{lemma}
	For all rules $\phi$ it holds that $\pi_\phi$ is $(\VVV_\phi,
	\VVV_\phi^{(h)})$-measurable.
\end{lemma}

\begin{proof}
	For all $\phi$, $\pi_\phi$ is a composition of the following kinds of 
	functions: projections $(x,y) \mapsto x$, transpositions $(x,y) \mapsto 
	(y,x)$, repetitions $x \mapsto (x,x)$, and appending constants $x \mapsto 
	(x,c)$. All of these are measurable with respect to the corresponding 
	standard Borel product $\sigma$-algebras, and so are their compositions.
\end{proof}

\begin{definition}
	The set of \emph{applicable pairs} for an instance $D \in \DD$ is given by
	\[
		\App( D )
		= \set[\big]{ (\phi,\tup a) \with \phi \text{ applicable for some }
		\tup u \text{ with } \pi_\phi( \tup u ) = \tup a }\text.\qedhere
	\]
\end{definition}

For any database instance $D$, the set $\App( D )$ tells us which rules are
applicable, and under which valuations of their free variables. Later, when we
talk about executions of a program, reaching an instance $D$ with $\App( D ) =
\emptyset$, intuitively means that the program terminates. We are going
to analyze the properties of $\App$ in order to show that we can 
select one particular pair from each $\App( D )$ in a 
measurable way.

\begin{example}[continues=exa:running]
	We consider two artificial examples of database instances $D$ of schema
	$(\Schema[E]^{\exists} \cup \Schema[I]^{\exists})$. Suppose that the rules
	of $\G_{\sal}^{\exists}$ (see \cref{fig:running_edl}) are enumerated as
	$\phi_1, \dots, \phi_5$, from top to bottom. Thus,
	\begin{align*}
		\phi_4 &= \phi_4(s,c,d,\mu) \\
				 &= \big( \exists z \colon \Res'( s, c, \mu, 10\,000, z ) 
				 \whenever \AffilEmployee(s,c,d), \PayScale(c,d,\mu) \big)
		\shortintertext{and}
		\phi_5 &= \phi_5(s,c,\mu,z)\\
				 &= \big( \Res(s,c,z) \whenever \Res'(s,c,\mu,10\,000,z) \big)
				 \text.
	\end{align*}
	\begin{enumerate}
		\item Suppose that $D$ contains exactly the facts\footnote{Units are only 
			displayed for illustration, and not part of the tuples.}
			\begin{align*}
				&\AffilEmployee( \text{981-00-8876}, \text{E-Corp}, \text{IT} )
				\text,\\
				&\AffilEmployee( \text{935-00-3912}, \text{E-Corp}, \text{IT} )
				\text,\\
				&\PayScale( \text{E-Corp}, \text{IT}, \text{\$~63\,000} )\text.
			\end{align*}
			Then $\App(D)$ contains exactly the two tuples
			\begin{align*}
				&( \phi_4, \text{981-00-8876}, \text{E-Corp}, \text{\$~63\,000}, 
				\text{$\$^{2}$~10\,000} ) \text{ and}\\
				&( \phi_4, \text{935-00-3912}, \text{E-Corp},
				\text{\$~63\,000}, \text{$\$^{2}$~10\,000} ) \text. 
			\end{align*}	
			Note that the department (for both tuples taking the value
			\enquote{IT}) is already projected away, and that the constant
			\enquote{[\textdollar\textsuperscript{2}] 10\,000} has been added).
		\item Now suppose that $D$ additionally contains the fact
			\[
				\Res'( \text{981-00-8876}, \text{E-Corp}, \text{\$~63\,000},
				\text{$\$^{2}$~10\,000}, \text{\$~62,271} )\text.
			\]
			With the presence of the new tuple, $\phi_4$ is no longer applicable
			for any valuation with
			\[
				(s,c,\mu) \mapsto 
				(\text{981-00-8876}, \text{E-Corp}, \text{\$~63\,000})\text.
			\]

			Yet, with the new tuple, $\phi_5$ is now applicable. In this situation
			$ \App( D ) $ contains exactly tuples
			\begin{align*}
				& ( \phi_4, \text{935-00-3912}, \text{E-Corp}, \text{\$~63\,000},
					\text{$\$^{2}$~10\,000})\text{ and}\\
				& ( \phi_5, \text{981-00-8876}, \text{E-Corp}, \text{\$~62,271} )
					\text.\qedhere
			\end{align*}
	\end{enumerate}
\end{example}

We let $( \DD, \DDD )$ be the measurable space of instances associated with the
schema $\Schema[E]^{\exists} \cup \Schema[I]^{\exists}$ (see \cref{sec:pdb}).
Note that the set $\DD_{\App} \coloneqq \set{ D \in \DD \with \App( D ) \neq 
\emptyset }$ is measurable in $( \DD, \DDD )$ using \cref{fac:measurableviews},
because the condition $\App( D ) \neq \emptyset$ is expressible as a Boolean
relational calculus query.
We let $\DDD_{ \App } \subseteq \DDD$ denote the trace $\sigma$-algebra of
$\DD_{ \App }$.  Formally, $\App$ is a multifunction $\App \from \DD_{ \App }
\toto \smash{\AA = \bigcup_{ \phi\in\G^{\exists} }\big( \set{ \phi }\times
\VV_{\phi}^{(h)}\big)}$. Therein, $\AA$ is naturally equipped with the
$\sigma$-algebra $\AAA \coloneqq \bigoplus_{ \phi \in \G^{ \exists } } \big(
\set{ \phi } \otimes \VVV_{\phi}^{(h)} \big)$.

Our goal is to show that $\App$ is a \emph{measurable} multifunction, and apply the
Theorem of Kuratowski and Ryll-Nardzewski (see \cref{fac:KRN}) to obtain a
measurable selection. Such a measurable selection is the kind of measurable
\enquote{policy} we sought to obtain.

\begin{lemma}\label{lem:App-meas}
	Let $\AAAA \in \AAA$. Then $\App^{-1}( \AAAA ) = \set{ D \in \DD_{\App} 
	\with \App(D) \cap \AAAA \neq \emptyset } \in \DDD_{\App} \subseteq \DDD$.
\end{lemma}

That is, for every measurable set $\AAAA$ of potentially applicable pairs, the
set of instances where a pair from $\AAAA$ really \emph{is} applicable, is
measurable.

\begin{proof}
	The function $D \mapsto \App(D)$ can be expressed as a relational algebra 
	view $V$ as follows. For every rule $\phi( \tup x )$ there exists a 
	relational algebra query $Q_\phi$ that on input $D$ returns all tuples $\tup
	u$ where $\tup u$ is a valuation of the free variables of $\phi$ making
	$\phi$ applicable. Then $\pi_\phi \after Q_\phi$ (with $\pi_\phi$ being 
	applied pointwise) is a measurable query as well. Our view $V$ finally is
	the deduplication of $\bigcup_{\phi \in \G^{\exists}} \set{ \phi } \times
	(\pi_\phi\after Q_\phi)(D)$. Then for all $\AAAA \in \AAA$ it holds that
	\[ \App^{-1}( \AAAA ) = V^{-1}\big( \CCCC( \AAAA, 0 )^{\c} \big) \]
	where $\CCCC(\AAAA,0)^{\c}$ is the set of instances in the
	associated instance measurable space that contain at least one fact from
	$\AAAA$. Since $V$ is measurable by \cref{fac:measurableviews}, the claim
	follows.
\end{proof}

It easily follows that $\App$ is a measurable multifunction on $\DD_{\App}$.

\begin{corollary}\label{cor:app}
	There exists a measurable function $\app \from \DD_{\App} \to \AA$ such that
	for all $D \in \DD_{\App}$ it holds that $\app(D) \in \App(D)$.
\end{corollary}

Subsequently, we use a measurable selection $\app$ to resolve the case of
multiple rules being applicable in a deterministic way. If multiple rules are
applicable (i.\,e. multiple tuples could be produced, possibly via sampling),
the function $\app$ selects the rule that is allowed to fire together with the
relevant part of the valuation.

\begin{remark}\label{rem:Appproj}
	With the introduction of $\pi_\phi$ into $\App$, we rectify an inaccuracy of
	the conference version \cite{Grohe+2020}. There, for the sake of simplicity,
	we did not distinguish valuations $\tup u$ that make a rule $\phi$
	applicable, and the resulting tuple $\pi_\phi(\tup u)$. With the above, we
	have made this distinction explicit. Technically, the simplification raises
	no problems due to the measurability of $\pi_\phi$.  Yet, possible
	projections have to be accounted for in the parallel chase procedure. We
	elaborate on this in the section on the parallel chase
	(\cref{sec:parchase}). 
\end{remark}

\subsection{Follow-Up Instances}\label{ss:followup}

For upgrading the semantics of $\EDL$ to a probabilistic one (according to the
original $\GDL$ program $\G$), we need a measurable correspondence between
\enquote{intermediate instances} that occur during the execution of the program
and all the \enquote{follow-up instances} or \enquote{extensions} that emerge
from such instances by a single rule application.  Intuitively, whenever a rule
is applicable (that is, its body is satisfied but its head is not), it may
fire.  If the rule is deterministic, then the ground fact from the head of the
rule gets added to the current database instance. If the rule is probabilistic,
then the ground fact from the head of the rule gets added with some valuation
of the existentially quantified variable and we get a distribution over the
follow-up instances according to the parameterized distribution from the
original rule. The present section is devoted to formalizing this set-up.  

\medskip

Let $\phi$ be a rule of $\G^{\exists}$ with ground space $( \VV_{\phi}^{(h)},
\VVV_{\phi}^{(h)} )$. If $\phi$ is existential, then its corresponding original
rule in $\G$ contains a $\Psi$-term, say using the parameterized distribution
$\psi$. We let $( \WW_{ \phi }, \WWW_{ \phi }, \mu_{ \phi } ) \coloneqq (
\WW_\psi, \WWW_\psi, \mu_\psi )$ be the underlying space of $\psi$. The
elements of $\WW_\phi$ are called the \emph{sample outcomes} of $\phi$ and
$(\WW_\phi, \WWW_\phi)$ the \emph{sample space} of $\phi$.

\begin{remark}\label{rem:determmeas}
	For every deterministic rule $\phi$, we introduce a dummy measure space $(
	\WW_\phi, \WWW_\phi, \mu_\phi )$ with $\WW_\phi = \set{ * }$ some fixed
	singleton set, $\WWW_\phi$ its powerset, and $\mu_\phi$ the function with
	$\set{ * } \mapsto 1$.
\end{remark}

We let $( \FF_{\phi}, \FFF_\phi ) = ( \FF_R, \FFF_R )$ if $R$ is the relation
symbol in the head of $\phi$. For all existential rules $\phi$, all $\tup a \in
\VV_\phi^{(h)}$, and all $b \in \WW_{\phi}$, we let $f_{ \phi }( \tup a, b )
\in \FF_\phi$ denote the fact that is obtained by substituting $\tup a$ and $b$
into the atom in the head of $\phi$. We define $f_{ \phi }( \tup a, b )$
similarly for deterministic rules, but in this case, the value $b = *$ is just
discarded. In particular, for all $\phi$, $f_{\phi}$ is a function $f_\phi
\from \VV_\phi^{(h)} \times \WW_\phi \to \FF_\phi$.

\begin{example}[continues=exa:running]
	Consider rule $\phi_4$ of $\G_{\sal}^{\exists}$ (see 
	\cref{fig:running_edl}). Let 
	\[
		\tup a = 
		( \text{981-00-8876}, \text{E-Corp}, \text{\$~63\,000}, 
		\text{$\$^{2}$~10\,000}) \in \VV_{\phi_4}^{(h)}
	\] 
	and let $b = \text{\$~62,271} \in \WW_{\phi_4}$ be a sample outcome. Then
	\[
		f_{ \phi_4 }( \tup a, b ) =
		\Res'( \text{981-00-8876}, \text{E-Corp}, \text{\$~63\,000},
		\text{\$~10\,000}, \text{\$~62,271} )\text.\qedhere
	\]
\end{example}

When $(\phi, \tup a)$ is an applicable pair in an instance $D$, then $\phi$ is
applicable for some $\tup u \in \pi_\phi^{-1}(\tup a)$. In that situation, 
$\phi$ may fire, which amounts to sampling a value $b \in \WW_{\phi}$, and
adding the fact $f_{ \phi }( \tup a, b )$ to $D$. We first describe this
process of adding facts formally, and regardless of rule applicability.
Note that clearly, $f_\phi$ is $(\VVV_\phi^{(h)}\otimes\WWW_\phi, 
\FFF_\phi)$-measurable, as it just prepends the right relation symbol to its
argument.

We consider two kinds of \emph{extension functions}, the \emph{sequential
extension function} and the \emph{parallel extension function}. The sequential
extension function captures the effect of firing a \emph{single} rule in a
database instance. In essence, if $(\phi, \tup a)$ is applicable in $D$ and $b$
a possible sample outcome for the parameterized distribution in  $\phi$, the
sequential extension function maps $D$ to the instance that is obtained by
adding $f_{\phi}( \tup a, b )$ to $D$. The parallel extension function does the
same thing, but for the case where \emph{multiple} rules may fire at once. We
will need that both functions obey certain measurability properties. The
remainder of this section formally introduces these functions, and establishes
various kinds of measurability results that are needed later.

The sequential extension function is defined as follows. Recall that
$\G^{\exists} = \set{\phi_1, \dots, \phi_k}$. First, let 
\[
	\GG 
	\coloneqq{}
		\DD \times 
		\bigcup_{ i = 1 }^{ k } \bigg(
			\set{ \phi_i } \times \VV_{ \phi }^{(h)} \times \WW_{ \phi } 
		\bigg)
	={}
	\set[\big]{
		( D, \phi, \tup a, b ) \with
		D \in \DD\text,~
		\phi \in \G^{ \exists }\text,~
		\tup a \in \VV_{ \phi }^{(h)}\text{, and }
		b \in \WW_{ \phi }
	}
      \]
and let $\GGG$ denote the $\sigma$-algebra on $\GG$ constructed in the
straight-forward way using the disjoint union and product constructions.
Then the function $\ext \from \GG \to \DD$ with
\[
	\ext( D, \phi, \tup a, b ) 
	= D \cup \set[\big]{ f_{ \phi} ( \tup a, b ) }
\]
is the \emph{sequential extension function} where $f_\phi$ is defined as
indicated above.

\smallskip

For every tuple $\tup \ell \coloneqq ( \ell_1, \dots, \ell_k ) \in \NN^{ k }$,
the parallel extension function with \emph{firing configuration} $\vec\ell$ is
defined as follows. Let
\begin{align*}
	\GG_{ \tup\ell } 
	\coloneqq{} &
	\DD \times \prod_{ i = 1 }^k \bigg( 
		\set{ \phi_i } \times 
		\big( \VV_{\phi_i}^{(h)} \times \WW_{\phi_i} \big)^{ \ell_i }
	\bigg)\\
	={} &
	\big\lbrace
		\big(D,(\phi_i,\tup a_{ij_i}, b_{ij})_{ij_i}\big) \with
		 D \in \DD\text,~ \phi_i \in \G^{\exists}\text,~ \tup a_{ij} \in
		 \VV_{\phi_i}^{(h)}\text{, and } b_{ij} \in \WW_{\phi_i})
		 \big\rbrace\text.
\end{align*}
As above, we equip $\GG_{\tup\ell}$ with its canonical $\sigma$-algebra 
$\GGG_{ \tup \ell }$. The function $\Ext_{ \tup \ell } \from \GG_{
\tup \ell } \to \DD$ with
\[
	\Ext_{\tup\ell}
	\big(D, (\phi_i,\tup a_{ij_i}, b_{ij_i})_{i,j_i} \big)
	= D \cup \bigcup_{ i = 1 }^{ k } \bigcup_{ j_i = 1 }^{ \ell_i }
\set[\big]{ f_{ \phi_i }( \tup a_{ij_i}, b_{ij_i} ) }
\]
is called a \emph{parallel extension function}.

\begin{definition}[Follow-Up Instances]
	Let $D \in \DD$.
	\begin{enumerate}
		\item If $(\phi,\tup a)$ is applicable in $D$, then every instance $\ext(
			D, \phi, \tup a, b )$ with $b \in \WW_{ \phi }$ is called a
			\emph{follow-up instance} of $D$ with respect to $(\phi, \tup a)$
			under sequential rule execution.
		\item If $\App(D) \neq \emptyset$, then every instance 
			$\Ext_{ \tup\ell }( D, (\phi_i, \tup a_{ij_i}, b_{ij_i}) )$ is called
			a \emph{follow-up instance} of $D$ with respect to $\App(D)$ under
			parallel rule execution, with firing configuration $\tup\ell = 
			( \ell_1, \dots, \ell_k )$, $\App(D) = \set{ (\phi_i,\tup a_{ij_i})
			\with 1 \leq i \leq k\text,~ 1\leq j_i \leq \ell_i }$, and  
			$b_{ij_i} \in \WW_{\phi_i}$.\qedhere
	\end{enumerate}
\end{definition}
Note that when discussing follow-up instances $\Ext_{\tup \ell}( D, ( \phi_i,
\tup a_{ij_i}, b_{ij_i} ))$, an entry $\ell_i = 0$ in the configuration $\vec
\ell$ corresponds to the rule $\phi_i$ not being able to fire in $D$.

For a given instance and fixed pairs $(\phi,\tup a)$, the various sets of
follow-up instances are measurable.

\begin{lemma}\label{lem:followupsmeas}
	Let $D \in \DD$.
	\begin{enumerate}
		\item For all $\phi\in\G^{\exists}$ and all $\tup a \in \VV_\phi^{(h)}$
			it holds that
			\[
				\bigcup_{ b \in \WW_{ \phi } } \ext( D, \phi, \tup a, b )
				\eqqcolon \ext( D, \phi, \tup a, \WW_{ \phi } )
				\in \DDD\text.
			\]
		\item For all $\vec \ell = ( \ell_1, \dots, \ell_k ) \in \NN^k$ and all
			$\phi_i \in \G^{\exists}$ and $\tup a_{ij_i} \in \VV_{\phi_i}^{(h)}$
			where $i = 1, \dots, k$ and $j_i = 1, \dots, \ell_i$ for all $i$, it
			holds that
			\[
				\bigcup_{ 
					\substack{ b_{ij_i} \in \WW_{\phi_i}\\
					\text{for all }i,j_i}
					} \Ext_{\tup\ell} \big(D, (\phi_i,\tup a_{ij_i}, b_{i{j_i}})_{i,j_i} \big)
				\eqqcolon 
				\Ext_{\tup\ell}\big( D,
					(\phi_i,\tup a_{ij_i},\WW_{\phi_i})_{i,j_i}\big)
					\in\DDD\text.\qedhere
			\]
	\end{enumerate}
\end{lemma}

\begin{proof}
	Consider the first part of the lemma and fix $D \in \DD$ and $(\phi,\tup a)
	\in \App(D)$. Then $f_{\phi}( \tup a, \WW_{\phi} ) = \set{ f_{\phi}(\tup a,
	b) \with b \in \WW_{\phi} } \in \FFF_\phi$, and it holds that
	\[
		\ext( D, \phi, \tup a, \WW_{\phi} ) = 
		\bigg(\bigcap_{ f \in D } \C( f, 1 ) \bigg) \cap \C( f_{\phi}( \tup a,
		\WW_\phi), 1 ) \cap \C\big( (D \cup f_{\phi}( \tup a,\WW_\phi ))^{\c}, 0
		\big)\text.
	\]
	This is a finite intersection of counting events, so the claim follows.

	The second part of the lemma can be shown analogously.
\end{proof}

In the remainder of this section, we show that the sequential and parallel
extension functions are measurable. Recall that $\FFF$ is the $\sigma$-algebra
on the space $\FF$ of all facts. 

\begin{lemma}\label{lem:containment}
	It holds that $\set[\big]{ (D,f) \in \DD\times\FF \with f \in D } \in \DDD
	\otimes \FFF$.
\end{lemma}

\begin{proof}
	Recall that $\FF$ is a Polish space. We fix a compatible Polish metric on
	$\FF$, as well as a countable dense set $\FF_0 \subseteq \FF$. Then for all
	$D \in \DD$ and $f \in \FF$ it holds that
	\[
		f \in D 
		\iff
		\forall\: \epsilon > 0\:
		\exists\: f_\epsilon \in \FF_0 \:
		\with
		(D,f) \in \C\big( B_{\epsilon}(f_\epsilon), \mathord>0 \big) \times 
		B_{\epsilon}( f_\epsilon )
		\text,
	\]
	where $B_{\epsilon}(f_\epsilon)$ denotes the open ball of radius $\epsilon$
	around $f_\epsilon$. The above equivalence easily translates to a countable
	combination of products of counting events and open balls. Thus, it follows
	that $\set{ (D,f) \with f \in D } \in \DDD\otimes\FFF$.
\end{proof}

It follows from \cref{lem:containment} and the measurability of $f_\phi$ that 
\begin{equation}\label{eq:Dphiab}
	\set{ (D, \phi, \tup a, b) \in \GG \with f_{\phi}( \tup a, b) \in D } \in 
	\GGG\text,
\end{equation}
as it is the intersection of the set from \cref{lem:containment} with
\[
	\bigcup_{ i = 1 }^k \big( \DD \times f_{\phi_i}^{-1}( \FF_{\phi_i} ) \big)
	\text.
\]

\begin{proposition}[Measurability of the Extension Functions]
	\label{pro:extExtmeas}\leavevmode
	\begin{enumerate}
		\item\label{itm:extmeas} The function $\ext \from \GG \to \DD$ is $(\GGG,
			\DDD)$-measurable.
		\item\label{itm:Extmeas} For all $\tup\ell \in \NN^k$, the function 
			$\Ext_{\vec\ell} \from \GG_{\tup\ell} \to \DD$ is $(\GGG_{\tup\ell},
			\DDD)$-measurable.\qedhere
	\end{enumerate}
\end{proposition}

\begin{proof}
	Note that \labelcref{itm:extmeas} is a consequence of 
	\labelcref{itm:Extmeas}, since
	\[
		\ext^{-1}(\D) = \Ext_{\tup e_1}^{-1}(\D) \cup \dots \cup 
		\Ext_{\tup e_k}^{-1}(\D)
	\] 
	where $\tup e_1,\dots,\vec e_k$ are the $k$ unit vectors. We will only show
	\labelcref{itm:Extmeas} for the special case of $k=1$ and $\tup\ell = (2)$
	and indicate in the end how the proof can be generalized to	arbitrary $k$
	and $\tup\ell$. 

	Let $\FFFF \in \FFF$ and $n \in \NN$. It holds that $( D, \phi, \tup a, b,
	\phi, \tup a', b' ) \in \Ext_{(2)}^{-1}\big( \C( \FFFF, n ) \big)$ if and
	only if one of the following holds (with $f \coloneqq f_\phi( \tup a, b )$
	and $f' \coloneqq f_\phi( \tup a', b' )$):
	\begin{enumerate}[label=(\roman*)]
		\item $D \in \C( \FFFF, n-2 )$ and $f,f' \in \FFFF$ with $f \neq f'$
			and $f,f' \notin D$.
		\item $D \in \C( \FFFF, n-1 )$ and 
			\begin{itemize}
				\item $f \in \FFFF$ but $f' \notin \FFFF$ and $f \notin D$, or
				\item $f = f' \in \FFFF$ and $f=f' \notin D$.
			\end{itemize}
		\item $D \in \C( \FFFF, n )$ and
			\begin{itemize}
				\item $f, f' \notin \FFFF$, or
				\item $f \in \FFFF$ but $f' \notin \FFFF$ and $f \in D$, or
				\item $f, f' \in \FFFF$ and $f, f' \in D$.
			\end{itemize}
	\end{enumerate}
	Thus, $\Ext_{(2)}^{ -1 }\big( \C( \FFFF, n ) \big)$ is the union of the sets
	described in the items above. The individual sets are measurable using the
	measurability of $f_\phi$, by the measurability of \eqref{eq:Dphiab}, and by
	the measurability of the diagonal $\set{ (f,f) \with f \in \FF_\phi }$ in
	$\FF_\phi \times \FF_\phi$.

	\medskip

	With the same ideas, the proof can be generalized to any firing 
	configuration $\tup \ell$ and any number of rules. All that is needed is a
	similar case distinction for facts $f_1, \dots, f_m$ (with $m = \ell_1 + 
	\dots + \ell_k$) over the number of facts of $\FFFF$ that are 
	contained in $D$, over the number of facts among $f_1, \dots, f_m$ that
	belong to $\FFFF$, and over whether some of the $f_1,\dots,f_m$ are equal.
\end{proof}

We let $\xi$ and $\Xi_{\tup\ell}$ denote the characteristic functions of the
graphs of $\ext$ and $\Ext_{\tup\ell}$, respectively. That is, $\xi \from \GG
\to \set{0,1}$ is defined by 
\begin{equation}\label{eq:xiXidef}
	\xi(D,\phi,\tup a,b,D') =
	\begin{dcases}
		1 & \text{if }\ext(D, \phi,\tup a,b) = D'\text{ and}\\
		0 & \text{otherwise,}
	\end{dcases}
\end{equation}
and $\Xi_{\tup \ell}$ is defined similarly.

\begin{corollary}\label{cor:xiXimeas}\leavevmode
	\begin{enumerate}
		\item The function $\xi$ is $(\GGG \otimes \DDD, \Borel(\RR))$-measurable.
		\item The function $\Xi_{\tup\ell}$ is $(\GGG \otimes \DDD,
			\Borel(\RR))$-measurable for all $\tup \ell\in\NN^{k}$.\qedhere
	\end{enumerate}
\end{corollary}

\begin{proof}
	By \cref{pro:extExtmeas}, the functions $\ext$ and $\Ext_{\tup \ell}$ are 
	measurable. Thus, their graphs are measurable sets in the corresponding 
	product space. Since characteristic functions of measurable sets are 
	measurable, the claim follows.
\end{proof}

\subsection{Induced Functional Dependencies}\label{ss:fd}

Following \cite{Barany+2017}, with every existential rule $\phi$ of 
$\G^{\exists}$, we associate a \emph{functional dependency} $\FD(\phi)$ in the
following way.  Recall that we assumed that all existential rules have the 
atom in their head in the format $R( \tup u, \tup p, z )$ where $z$ is the 
existentially quantified variable. Suppose $R$ is the relation symbol in the
head of $\phi$ with attributes $A_1, \dots, A_k$. Then $\FD(\phi)$ is the
functional dependency $R \with A_1, \dots, A_{k-1} \to A_k$. Then this
functional dependency intuitively expresses that there is at most one value of
the random (resp.\ existential) attribute when all other attribute values are
fixed, cf. \cite[p.~22:8]{Barany+2017}. 

Recall that $(\DD,\DDD)$ is the measurable space of database instances
of schema $\Schema[E]^{ \exists } \cup \Schema[I]^{\exists }$. Input instances
are restricted to the schema $\Schema[E]^{\exists}$. We denote by $(\DD_{\IN},
\DDD_{\IN})$ the measurable space of instances over $\Schema[E]^{\exists}$.
Note that $\DD_{\IN} \subseteq \DD$ and $\DDD_{\IN} \subseteq \DDD$.

The following is easy to check using the definitions of $\App$, $f_{\phi}$, 
$\ext$ and $\Ext$.

\begin{lemma}[cf. {\cite[Proposition 4.2]{Barany+2017}}]\label{lem:fd}
	Let $\phi$ be an existential rule of $\G^{\exists}$. Then the following
	holds:
	\begin{enumerate}
		\item Every database instance $D \in \DD_{ \IN }$ satisfies 
			$\FD(\phi)$.
		\item If $D \in \DD$ and $\App(D) \neq \emptyset$, then $D$ satisfies 
			$\FD(\phi)$. Moreover, for all $(\phi,\tup a) \in \App(D)$, and
			all $b \in \WW_\phi$, the follow-up instance $\ext(D,\phi,\tup a,b)$
			satisfies $\FD(\phi)$ as well. 
			Likewise, it holds that $\Ext_{\vec\ell}(D, (\phi_i,\tup a_{ij_i}, 
			b_{ij_i})_{i,j_i})$ satisfies $\FD(\phi)$ where $\App( D ) = \set{
			(\phi_i,\tup a_{ij_i}) \with 1\leq i\leq k \text{, } 1\leq j_i\leq
			\ell_i}$ and $b_{ij_i} \in \WW_{\phi_i}$ for all $i, j_i$.\qedhere
	\end{enumerate}
\end{lemma}

This result intuitively means that in the execution of our programs, we will
only ever have a single sample outcome for a given instantiation of head
variables. This becomes crucial at a later point, as it allows us to show that
the computation steps needed to obtain an intermediate instance $D$ from an
input instance $D_{\IN}$ to the program are unique.

\section{Sequential Probabilistic Chase}\label{sec:seqchase}

The chase of a \GDL{} program $\G$ corresponds
to chasing its \EDL{} version $\G^{\exists}$. We will construct a
\emph{chase tree} for the \EDL{} program $\G^{\exists}$, the nodes of
which are labeled with database instances, and the edges of which
capture applications of rules. Thus, existential rules lead to nodes
with multiple (in our case possibly uncountably many!) children. In
the countable case, one can label the edges to these children 
with probabilities according to the probabilistic rule that the
existential $\EDL$ rule was constructed from
(cf. \cref{sec:assocedl}). This is the approach \cite{Barany+2017} took.

We follow the general spirit of this approach. However, the edge
labeling outlined above is only sufficient for domains that are
countably infinite at most. Instead of labeling edges, we label nodes with
the probability distribution over their
children. Yet, making the distribution explicit in the chase tree is
not necessary, as it is implicit from the current instance $D$ and the
applicable pair we use. In this section, we formalize this procedure
and demonstrate how such \emph{chase trees induce a stochastic process on
database instances}.

\begin{remark}\label{rem:independenceass}
In the construction of said stochastic process, implicit independence 
assumptions are made. Intuitively, we want that if multiple samplings occur 
along a path, then they are stochastically independent, as long as there is no 
\emph{logical} dependence between them. This means that random samplings should
only depend on the current state of the database where the corresponding rule
and instantiation of variables get applicable, and ultimately comes down to
the stochastic process being Markov (cf. \cref{sec:stochproc}).
\end{remark}

Note that from a measure-theoretic point of view, there is no need to
associate the execution of a Datalog program to a tree as we
are going to do.  We believe though, that doing so is beneficial
for exposing the intuition behind the underlying stochastic process
and for emphasizing the connections to the original approach in
\cite{Barany+2017}.

\subsubsection*{Structure of this section}
In \cref{ssec:seqstep} we introduce the central notion of \emph{chase steps},
capturing the effect of a single rule application from the existential Datalog
program, and, in the case of existential rules providing it with a
probabilistic structure based upon the parameterized distribution that induced
the existential rule. This can be thought of as the continuous generalization
of the notion of chase steps from \cite{Barany+2017}. Chase steps naturally
compose into a \emph{chase tree}, that, in turn, captures the stepwise
execution of the whole program. Along the way, we already make some technical
observations concerning these notions, before focusing, in
\cref{ssec:seqpaths}, on paths in the chase tree. These paths correspond to
individual runs of the program, including fixed sampling outcomes for the
existential rules. Every chase path either leads to a finite output instance
(where no rules are applicable anymore), or is infinite. In \cref{ssec:limit}
we formalize this in terms of a function that maps terminating paths to their
output instance, and non-terminating paths to some error event. Our main
concern lies in proving that the probability distribution over chase paths can
be described in terms of \emph{stochastic kernels} comprising the individual
chase steps. We show this in \cref{ssec:seqmarkov}. This establishes that the
chase paths are the paths of a Markov process over the space of database
instances. Combining this with the tools from \cref{ssec:limit} allows us to
associate a well-defined output (sub-)probability distribution for our
programs.

\subsection{Chase Steps and Chase Trees}\label{ssec:seqstep}

A \emph{chase step} captures the semantics of applying a single (applicable)
rule to an input instance.

\begin{definition}[Chase Step]\label{def:seqstep}%
  A \emph{(sequential) chase step} for $\G$ is a tuple
  $(D, \phi, \tup a, \EEEE, \mu)$, where
  \begin{itemize}
  \item $D$ is a database instance in $\DD_{\App}$ (i.\,e. $\App(D) \neq 
    \emptyset$)
  \item $( \phi, \tup a ) \in \App( D )$ is an applicable pair
  \item $\EEEE = \ext( D, \phi, \tup a, \WW_{\phi } )$ is the set of
    follow-up instances of $D$ w.\,r.\,t. $(\phi, \tup a)$
    under sequential rule execution, and
  \item $\mu$ is the probability measure on $\DDD \restriction_{\EEEE}$ that 
    is defined as follows:
    \begin{itemize}
    \item If $\phi$ is an existential rule of $\G^{\exists}$, then
      for all measurable $\DDDD \subseteq \EEEE$,
      \begin{equation}\label{eq:seqstepmeas}
        \mu( \DDDD ) =
        \int_{ \WW_\phi } \xi( D, \phi, \tup a, \placeholder, \DDDD )
        \cdot \psi_{\phi}\params{ \tup a }( \placeholder ) 
        \d{ \mu_\phi }\text,
      \end{equation}
      where $\psi_\phi$ is the parameterized distribution in the rule
		of $\G$ that $\phi$ originated from.\footnote{Note that $\tup a$ contains
		the full tuple of parameters (cf. \labelcref{ss:rule-applicability}). To
		be precise, we should write $\psi_\phi\params{ \tup p }$ instead of
		$\psi_\phi\params{ \tup a }$, where $\tup p$ is the projection of $\tup
		a$ to the parameter part. Again, all we need is that this transformation 
		is measurable, but this is clearly the case, because $\tup a \mapsto \tup
		p$ is simply a projection between product measurable spaces.
		Alternatively, we could just let $\psi_\phi$ be a version of the
		parameterized distribution that ignores those components of $\tup a$ that
		do not belong to the parameter. Either way, we just write
		$\psi_\phi\params{\tup a}$ and treat $\tup a$ as if it were the parameter 
		tuple itself.}
    \item Otherwise,
      \begin{equation}\label{eq:detseqchasestep}
        \mu( \DDDD ) =
        \begin{cases}
          1 & \text{if }\DDDD = \set{ \ext( D, \phi, \tup a, * ) }
          \text{, and}\\
          0 & \text{if }\DDDD = \emptyset\text.
        \end{cases}\qedhere
      \end{equation}
    \end{itemize}
  \end{itemize}
\end{definition}

Recall that $\xi( D, \phi, \tup a, b, \DDDD )$ is an indicator telling us 
whether in $D$ an application of the rule $\phi$ with head grounding $\tup a$
and resulting sample $b$ leads to an instance in $\DDDD$. Thus, the integral in
\labelcref{eq:seqstepmeas} is the total probability of all samples $b$ that
lead from $D$ to an instance in $\DDDD$ when rule $\phi$ is fired with head
grounding $\tup a$.

We denote a chase step $(D, \phi, \tup a, \EEEE, \mu)$ as
\[
	\chasestep{ D }{ (\phi,\tup a) }{ (\EEEE,\mu) }
\]
and say that the chase step \emph{starts in $D$}, \emph{uses} $(\phi,\tup a)$,
and \emph{goes into $\EEEE$ with distribution $\mu$}. Note that in such a chase
step, $\EEEE$ and $\mu$ are determined by $D$, $\phi$ and $\tup a$.

\begin{figure}[t]
	\mbox{}\hfill%
	\begin{subfigure}{.45\textwidth}\centering%
			\begin{tikzpicture}[sibling distance=.3cm,level distance=2.25cm]
				\node[inner,label={above:$D$}] (D) {}
					child { node[inner] (A1) {} edge from parent[draw=none] }
					child { node[inner] (A2) {} edge from parent[draw=none] }
					child { node[inner] (A3) {} edge from parent[draw=none] }
					child { node[inner] (A4) {} edge from parent[draw=none] }
					child { node[inner] (A5) {} edge from parent[draw=none] }
					child { node[inner] (A6) {} edge from parent[draw=none] }
					child { node[inner] (A7) {} edge from parent[draw=none] }
					child { node[inner] (A8) {} edge from parent[draw=none] }
					child { node[inner] (A9) {} edge from parent[draw=none] }
					child { node[inner] (A10) {} edge from parent[draw=none] }
					child { node[inner] (A11) {} edge from parent[draw=none] }
					child { node[inner] (A12) {} edge from parent[draw=none] }
					child { node[inner] (A13) {} edge from parent[draw=none] }
					child { node[inner] (A14) {} edge from parent[draw=none] }
					child { node[inner] (A15) {} edge from parent[draw=none] }
					child { node[inner] (A16) {} edge from parent[draw=none] };
				\begin{scope}[on background layer]
					\path[fill=white!98!black,draw,thick] 
						(D.center) to (A1.center) to (A16.center) to (D.center);
				\end{scope}
				\draw[-stealth',thick,shorten >=8pt] 
					(D.center) to node[xshift=.2cm] {$b$} (A10);
				\node[left=.1cm of A1] {$\EEEE\colon$};
				\node[right=.1cm of A16] {$\phantom{\EEEE\colon}$};

				\node[text=b20] at (0.9,-1.55) {$\mu$};
				\node[text=b50] at (-1,-1.5) {$\mu(\DDDD)$};

				\tikzset{bar/.style={minimum width=6pt,rectangle,very thick,
					inner sep=0pt,draw=b20,fill=b10,anchor=south}};
				
				\begin{scope}[on background layer]
					\node[bar,minimum height=1pt] at (A1) {};
					\node[bar,minimum height=3pt] at (A2) {};
					\node[bar,minimum height=5pt,fill=b20,draw=b50] at (A3) {};
					\node[bar,minimum height=7pt,fill=b20,draw=b50]	at (A4) {};
					\node[bar,minimum height=9pt,fill=b20,draw=b50]	at (A5) {};
					\node[bar,minimum height=12pt,fill=b20,draw=b50] at (A6) {};
					\node[bar,minimum height=14pt,fill=b20,draw=b50] at (A7) {};
					\node[bar,minimum height=11pt,fill=b20,draw=b50] at (A8) {};
					\node[bar,minimum height=8pt] at (A9) {};
					\node[bar,minimum height=8pt,fill=b50,draw=black] at (A10) {};
					\node[bar,minimum height=10pt] at (A11) {};
					\node[bar,minimum height=9pt] at (A12) {};
					\node[bar,minimum height=7pt] at (A13) {};
					\node[bar,minimum height=5pt] at (A14) {};
					\node[bar,minimum height=3pt] at (A15) {};
					\node[bar,minimum height=1pt] at (A16) {};
				\end{scope}

				\path[fill=none,draw,thick]
					(D.center) to (A1.center) to (A16.center) to (D.center);
				\node[below=2ex of A10] (A10L) 
					{$\mathllap{\ext(D, \phi, \tup a, b) ={}} 
					D' \mathrlap{{}\in\EEEE}$};
				\draw[b50,<-,shorten >=2pt,inner sep=1pt] (A10L) to (A10);
			\end{tikzpicture}
			\caption{Discrete case.}
	\end{subfigure}\hfill%
	\begin{subfigure}{.45\textwidth}\centering%
		\begin{tikzpicture}[level distance=2.25cm]
			\node[inner,label={above:$D$}] (D) {}
				child { node (A1) {} edge from parent[draw=none] }
				child { node (A2) {} edge from parent[draw=none] }
				child { node[inner] (A3) {} edge from parent[draw=none] }
				child { node (A4) {} edge from parent[draw=none] };

			\begin{scope}[on background layer]
				\path[fill=white!98!black,draw,thick] 
					(D.center) to (A1.center) to (A4.center) to (D.center);
			\end{scope}
	
			\node[below=2ex of A3] (A3L) 
					{$\mathllap{\ext(D, \phi, \tup a, b) ={}} 
					D' \mathrlap{{}\in\EEEE}$};
				\draw[b50,<-,shorten >=2pt,inner sep=1pt] (A3L) to (A3);

			\draw (D.center) to (A1.center);
			\draw (D.center) to (A4.center);
			\draw[-stealth',thick] (D.center) to 
				node[xshift=.15cm,yshift=.2cm] {$b$} (A3);
			\draw (A1.center) to (A4.center);

			\node[left=0cm of A1] {$\EEEE\colon$};
			\node[right=0cm of A4] {$\phantom{\EEEE\colon}$};
			\begin{scope}[on background layer]
				\draw[fill=b10,draw=b20,very thick] plot[smooth] coordinates 
					{(A1.center) (-.6,-1.5) (.3,-1.7) (.9,-1.6) (A4.center)};
				\clip ([xshift=-1cm]A2.center) rectangle ([shift={(1,1)}]A2.center);
				\draw[fill=b20,draw=b50,very thick] 
					plot[smooth] coordinates
					{(A1.center) (-.6,-1.5) (.3,-1.7) (.9,-1.6) (A4.center)};
				\path[clip] plot[smooth] coordinates 
					{(A1.center) (-.6,-1.5) (.3,-1.7) (.9,-1.6) (A4.center)};
				
				\draw[b50,very thick] 
					([xshift=-.98cm]A2.center) to ([shift={(-.98,1)}]A2.center);
				\draw[b50,very thick] 
					([xshift=.98cm]A2.center) to ([shift={(.98,1)}]A2.center);
			\end{scope}

			\node[text=b20] at (0.9,-1.35) {$\mu$};
			\node[text=b50] at (-.6,-1.2) {$\mu(\DDDD)$};
			\draw[thick] (0,0) +(-165:.3) arc (-165:-15:.3);
		\end{tikzpicture}
		\caption{Continuous case.}
	\end{subfigure}
	\hfill%
	\caption{Illustration of a (sequential) chase step.}\label{fig:seqchasestep}
\end{figure}

\Cref{fig:seqchasestep} illustrates a sequential chase
step $\chasestep{ D }{ (\phi,\tup a) }{ (\EEEE,\mu) }$ as a directed
tree of depth 1 with root $D$. The children of $D$ are the follow-up
instances of $D$ using $(\phi,\tup a)$, and the edge from $D$ to a
follow-up instance $D'$ corresponds to a sample outcome
$b \in \WW_\phi$. The illustration also insinuates that the transition
from $D$ to $D'$ is probabilistic. In particular, note that $D$ has
uncountably many children if the random variable that is sampled has
an uncountable support.

\medskip

\begin{remark}\label{rem:determpdist}
Recall from \cref{rem:determmeas} that for deterministic rules $\phi$, we
defined $\WW_\phi = \set{ * }$ for some dummy singleton $\set{ * }$. Letting
$\psi_\phi\params{\tup a}( * ) = \mu_\phi(*) = 1$ in this case, we can regard
\labelcref{eq:detseqchasestep} as a special case of \labelcref{eq:seqstepmeas}.
This allows us without loss of generality to uniformly treat all the chase step
measures that appear later as if they were of shape \labelcref{eq:seqstepmeas}.
\end{remark}

In \cref{def:seqstep}, we covertly claimed that $\mu$ as in
\labelcref{eq:detseqchasestep,eq:seqstepmeas} is a well-defined
probability measure. We repay the debt by showing that this is indeed the case.

\begin{lemma}\label{lem:seqstepmeas}
  The function $\mu$ from \cref{def:seqstep} is well-defined, and a 
  probability measure on $\DDD\restriction_{\EEEE}$.
\end{lemma}

\begin{proof}
	Let $\chasestep{ D }{ (\phi,\tup a) }{ (\EEEE,\mu) }$ be a sequential chase
  step. Then by \cref{lem:followupsmeas}, it holds that $\EEEE = \ext( D, 
  \phi, \tup a, \WW_\phi ) \in \DDD$. For fixed, measurable $\DDDD \subseteq 
  \EEEE$, consider the function
  \begin{equation}\label{eq:integrandxi}
    \xi( D, \phi, \tup a, \placeholder, \DDDD ) \with
    \WW_\phi \to \set{0,1} \with
    b \mapsto \begin{cases}
      1 & \text{if } D \cup \set{ f_{\phi}(\tup a,b) } \in \DDDD\text, \\
      0 & \text{otherwise.}
    \end{cases}
  \end{equation}
  Then $\xi(D, \phi, \tup a, \placeholder, \DDDD )$ maps $b$ to
  $1$ if and only if $b$ is in the $(D, \phi, \tup a)$-section
  of $\ext^{-1}( \DDDD )$. In particular,
  $\xi( D, \phi, \tup a, \placeholder, \DDDD )$ is
  $\big( \WWW_{\phi}, \Borel( \RR_{\geq 0} ) \big)$-measurable,
  and it follows that
  \[
    \xi( D, \phi, \tup a, \placeholder, \DDDD ) \cdot 
    \psi_{\phi}\params{ \tup a }
  \]
  (as a product of real-valued, measurable functions) is
  $( \WWW_{\phi}, \Borel( \RR_{\geq 0} ) )$-measurable as
  well. Thus, the function $\mu$, as defined in
  \labelcref{eq:seqstepmeas} is a well-defined measure. Note
  that for $\DDDD = \EEEE$, the function from
  \labelcref{eq:integrandxi} is the constant $1$-function. Since
  $\psi_\phi$ is a parameterized distribution, it follows that
  $\mu$ is a probability measure.
\end{proof}

Given a database instance $D$, we can now argue about sequences of
follow-up instances using sequences of chase steps.
\begin{equation}
	\cdots
	\qquad
	\chasestep{ D }{ (\phi,\tup a) }{ (\tikzmark{mark1}\EEEE,\mu) }
	\qquad\vphantom{\rule{0pt}{6ex}}
	\tikzmark{mark2}\chasestep{ D' }{ (\phi',\tup a') }{ (\EEEE',\mu') }
	\qquad
	\cdots
\begin{tikzpicture}[overlay,remember picture] 
	\node[anchor=south west,inner sep=0pt] (node1) at (pic cs:mark1) 
		{\phantom{$\EEEE$}};
	\node[anchor=south west,inner sep=0pt] (node2) at (pic cs:mark2) 
		{\phantom{$D'$}};
	\draw[-{Straight Barb[angle'=90,scale=.8]},semithick] 
		([yshift=1ex]node1.60) 
		to[out=60,in=120] node[pos=.5,circle,fill=white,inner sep=2pt] {$\ni$} 
		([yshift=1ex]node2.120);
\end{tikzpicture}
\end{equation}

Here sequences can branch, when the rules that are applied are existential
rules of the \EDL{} version of $\cal G$ (cf.  \cref{fig:seqchasestep}). What we
just described is formalized in the notion of \emph{chase trees}.

Recall that $\DD_{\App}$ denotes the set of instances $D$ with $\App( D ) \neq
\emptyset$, and that $(\AA, \AAA)$ is the space of pairs $(\phi, \tup a)$ with
$\phi \in \G^{ \exists }$ and $\tup a \in \VV_{\phi}^{(h)}$. 

\begin{definition}
	A \emph{measurable chase policy} is a measurable function $\app \from \set{
	D \in \DD \with \App( D ) \neq \emptyset } \to \AA$ with the property that
	$\app(D) \in \App( D )$.
\end{definition}

By \cref{cor:app}, a measurable chase policy always exists (as long as any rule
is applicable in some instance at all). Intuitively, if multiple pairs $(\phi,
\tup a)$ are applicable in an instance, $\app(D)$ stipulates which of these is
used for the next chase step.

\begin{definition}[Chase Tree]\label{def:seqchasetree}
	Let $D_{ \IN } \in \DD_{ \IN }$ and let $\app \from \DD_{\App} \to \AA$ be a
	measurable chase policy. The \emph{(sequential) chase tree}
	$T_{\app,D_{\IN}}$ for input instance $D_{ \IN }$ with respect to the
	$\GDL$ program $\G$ and $\app$ is a labeled countable-depth tree 
	$T_{\app,D_{ \IN }} = ( V, E, \Lambda )$ with labeling function $\Lambda$, 
	root node $r \in V$, and the following properties.
	\begin{enumerate}
		\item The root node $r$ is labeled with $D_{ \IN }$.
		\item If $v \in V$ is a leaf node, then it is labeled with an instance
			$D_v$ such that $\App( D_v ) = \emptyset$.
		\item If $v \in V$ is an inner node, then it is labeled with an instance
			$D_v$ such that $\App( D_v ) \neq \emptyset$ and
			\begin{enumerate}
				\item $\chasestep{ D_v }{ \app(D_v) }{ (\EEEE_v, \mu_v) }$ is a
					chase step with $( \EEEE_v, \mu_v )$ as in 
					\cref{def:seqstep}.
				\item the function $v' \mapsto D_{ v' }$ is a bijection between
					the children of $v$ and $\EEEE_v$.
				\qedhere
			\end{enumerate}
	\end{enumerate}
\end{definition}

In essence, a chase tree captures all possible computations of a $\GDL$ 
program according to a (measurable) chase policy. See \cref{fig:seqchasetree}
for an illustration. Observe that for all $D_{\IN} \in \DD_{\IN}$, the chase
tree $T_{\app,D_{\IN}}$ is uniquely determined by $\app$. The tree may contain
paths of (countably) infinite length, and it may contain nodes with uncountably
many children.

\begin{figure}[t]\centering%
	\begin{tikzpicture}[sibling distance=.4cm,level distance=1.25cm]
		\node[inner, label={above:$D_{\IN}$}] (root) {}
		child { node[hidden] (a1) {} edge from parent[draw=none] }
		child { node[hidden] (a2) {} edge from parent[draw=none] }
		child { node[hidden] (a3) {} edge from parent[draw=none] }
		child { node[hidden] (a4) {} edge from parent[draw=none] 
			child { node[hidden] (b1) {} edge from parent[draw=none] }
			child { node[hidden] (b2) {} edge from parent[draw=none] }
			child { node[hidden] (b3) {} edge from parent[draw=none] }
			child { node[hidden] (b4) {} edge from parent[draw=none] }
			child { node[hidden] (b5) {} edge from parent[draw=none] }
			child { node[hidden] (b6) {} edge from parent[draw=none] }
			child { node[hidden] (b7) {} edge from parent[draw=none] }
			child { node[hidden] (b8) {} edge from parent[draw=none] }
		}
		child { node[hidden] (a5) {} edge from parent[draw=none] }
		child { node[hidden] (a6) {} edge from parent[draw=none] }
		child { node[hidden] (a7) {} edge from parent[draw=none] }
		child { node[hidden] (a8) {} edge from parent[draw=none] }
		child { node[hidden] (a9) {} edge from parent[draw=none] }
		child { node[hidden] (a10) {} edge from parent[draw=none] };

		\path[draw=b20,thick,fill=b10] 
			(a1.center)  .. controls ++(1.3,.7) ..  (a10.center);
		\draw[thick] (root.center) to (a1.center) to (a10.center) to (root.center);
		\node[inner] at (a4) {};
		\draw[-stealth',b75,thick] (root.center) to (a4);

		\path[draw=b20,thick,fill=b10] 
			(b1.center) .. controls ++(1.6,.7) .. (b8.center);
		\draw[thick] (a4.center) to (b1.center) to (b8.center) to (a4.center);
		\node[inner] at (b4) {};
		\draw[-stealth',thick,b75] (a4.center) to (b4);
		
		\node[font=\small, below right=.1cm and .1cm of b4,align=left] 
			{intermediate\\instance};

		\node[below left=4.5cm and 2cm of a1.center,anchor=center] (c1) {};
		\node[below right=4.5cm and 2cm of a10.center,anchor=center] (c3) {};
		\path (a1.center) to node[pos=.65,anchor=center] (d) {} (c1.center);
		\draw[dotted,thick] (a1.center) to (d.center);
		\draw[dotted,thick] (a10.center) to (c3.center);
		\node[right=1.3cm of d.center,anchor=center] (e) {};
		\node[right=2cm of c1.center,anchor=center] (f) {};
		\node[right=.3cm of d.center,anchor=center] (d1) {};
		\node[below=.5cm of d1.center,anchor=center] (d2) {};
		\node[right=1.7cm of d2.center,anchor=center] (d3) {};
		\node[below=.6cm of d3.center,anchor=center] (d4) {};
		\node[right=1.2cm of d4.center,anchor=center] (d5) {};
		\node[below=.475cm of d5.center,anchor=center] (d6) {};

		\draw[thick,b50,dotted] (d.center) to 
			node[black,pos=.6,leaf,solid]{}
			(d1.center) to (d2.center) to 
			node[black,pos=.2,leaf,solid]{} 
			node[black,pos=.5,leaf,solid]{} 
			node[black,pos=.7,name=myleaf,leaf,solid]{}
			(d3.center) to (d4.center) to
			node[black,pos=.3,leaf,solid]{}
			node[black,pos=.6,leaf,solid]{}
			(d5.center) to (d6.center);

		\node[font=\small,below=.1cm of myleaf,align=right,anchor=north east]
			{instances without\\applicable rule};

		\path (d6) to 
			node[shift={(-.2,.2)},rotate=-70,anchor=center] {$\cdots$}
			(c3);

		\path[draw=b75,thick,-stealth',decorate,
				decoration={snake,amplitude=.05cm,post length=.15cm}] 
				(b4) to (myleaf);
	\end{tikzpicture}
	\caption{Illustration of a sequential chase tree, cf. \cref{fig:seqchasestep}.}
	\label{fig:seqchasetree}
\end{figure}

\begin{lemma}\label{lem:seqinj}
	Let $D_{ \IN } \in \DD_{ \IN }$ be an input instance and let $\app$ be a
	measurable chase policy. Then $v \neq w$ implies $D_v \neq D_w$ for all 
	nodes $v \neq w$ in the corresponding chase tree $T_{ \app, D_{ \IN } }$.
\end{lemma}

This is shown using the functional dependencies introduced by the
$\EDL{}$ program. The short proof below is directly transferred from 
\cite{Barany+2017} to our setting.

\begin{proof}
	Let $D_{ \IN }$ and $\app$ be fixed and suppose there exist $v, w \in 
	V\big( T_{ \app, D_{\IN} } \big)$ with $v \neq w$ such that $D_v = D_w$. Let
	$u$ be the least common ancestor of $v$ and $w$ in $T_{\app,D_{\IN}}$. In
	particular, the node $u$ has multiple child nodes.  Thus, the rule $\phi_u$
	from $\app( D_u ) = ( \phi_u, \tup a_u )$ is existential. Let $v'$ and $w'$
	be the children of $u$ on the path to $v$ respectively $w$. Then $D_{v'} =
	D_u \cup \set{ f_v }$ and $D_{w'} = D_u \cup \set{ f_w }$ for some $f_v, f_w
	\in \FF_{ \phi_u }$ with $f_v \neq f_w$. By the setup of $T_{\app,D_{\IN}}$, 
	$D_u \subseteq D_v = D_w$, so $f_v, f_w \in D_v = D_w$. This, however,
	contradicts \cref{lem:fd}.
\end{proof}

Note that if the $\GDL$ program $\G$ only contains discrete distributions, then
any chase policy $\app$ is trivially measurable. In this case, and modulo our 
changes to the existential version $\G^{\exists}$, we obtain the same chase
trees as \cite{Barany+2017} (omitting any kind of probability labels).

\subsection{Chase Paths}\label{ssec:seqpaths}

For the remainder of the section, we fix an arbitrary input instance $D_{ \IN 
} \in \DD_{ \IN }$, and an arbitrary measurable chase sequence $\app$. Our goal
is to construct a Markov process on the space of database instances $(\DD, 
\DDD)$ by embedding the sequential chase tree $T_{ \app, D_{\IN} } \eqqcolon
T = (V, E, \Lambda)$ into the \emph{path space} $( \DD^{\omega},
\DDD^{\otimes\omega} )$. 

\bigskip

We define a binary relation $\mathord{\stepE[\app]} \subseteq \DD \times \DD$
(denoted in infix notation) on $\DD$ as follows. Let $D \in \DD$.
\begin{itemize}
	\item If $\App( D ) = \emptyset$, then $D \stepE[\app] D'$ if and only
		if $D' = D$.
	\item If $\App( D ) \neq \emptyset$, then there exists a unique chase step
		$\smash{\chasestep{ D }{ \app( D ) }{ (\EEEE, \mu) }}$. Then $D
		\stepE[\app] D'$ if and only if $D' \in \EEEE$.
\end{itemize}
In this section, we drop the subscript and just write $\stepE$ instead of
$\stepE[\app]$. 

\begin{lemma}\label{lem:stepE}
	The relation $\mathord{\stepE} \subseteq \DD \times \DD$ is measurable,
	i.\,e. $\mathord{\stepE} \in \DDD \otimes \DDD$.  
\end{lemma}

\begin{proof}
  We decompose $\stepE$ depending on the rule prescribed by the chase 
  policy. First observe 
  \[
    \set{ ( D, D' ) \in \DD^2 \with \App( D ) = \emptyset \text{ and } D \stepE D' }
    = 
    \diag( \DD^2 ) \cap \big( \DD \setminus ( \DD_{ \App } ) \times \DD \big)
    \in \DDD \otimes \DDD\text.
  \]
  Now fix a rule $\phi \in \G^{\exists}$. We conclude the proof by showing
  \begin{equation}\label{eq:vdashphi}
    \set{ (D, D' ) \in \DD^2 \with 
      \App( D ) \ni (\phi,\tup a) = \app( D )\text{ for some } \tup a \in 
		\VV_\phi^{(h)} \text{ and } D \stepE D' } 
    \in \DDD \otimes \DDD\text.
  \end{equation}
  Recall that $( \VV_{\phi}^{(h)}, \VVV_{\phi}^{(h)} )$, the space of head
  groundings of $\phi$, is standard Borel. Moreover, recall that the space of
  all facts over $\Schema[E] \cup \Schema[I]$, $( \FF, \FFF )$ is standard 
  Borel.

  For each of these spaces, we fix compatible Polish metrics, and countable
  dense sets $\VV_0 \in \VVV_\phi^{(h)}$, respectively $\FF_0 \in \FFF$. We let
  $B_\epsilon(x)$ denote the open ball of radius $< \epsilon$ around an element 
  $x$ in either space.

  Then a pair of instances $( D, D' )$ is contained in the set from 
  \cref{eq:vdashphi} if and only if $D$ has the shape $\set{ f_1, \dots, f_n 
  }$ and it holds that $\app( D ) = ( \phi, \tup a )$ for some $\tup a$ such 
  that $D' = \set{ f_1, \dots, f_n } \cup \set{ f_\phi( \tup a, b ) }$ for some
  $b \in \WW_\phi$. This is the case if and only if there exists some
  $\epsilon_0 > 0$ such that for all $\epsilon \in (0, \epsilon_0)$ there
  exists some $\tup a_{\epsilon} \in \VV_0$, and facts $f_{1,\epsilon} \dots,
  f_{n,\epsilon} \in \FF_0$ of distance at least $\epsilon_0/3$ to each other
  such that
  \begin{enumerate}
	   \item $\app( D ) = ( \phi, \tup a )$ for some $\tup a
			\in B_{\epsilon}( \tup a_{\epsilon} )$, i.\,e. $D \in \app^{-1}\big( 
			\phi, B_{\epsilon}( \tup a_{\epsilon} ) )$,
		\item for all $i = 1, \dots, n$, both $D$ and $D'$ contain 
			exactly one fact from $B_{\epsilon}( f_{i,\epsilon} )$,
		\item $D$ contains no fact outside of $\bigcup_{ i = 1 }^{
			n } B_{\epsilon}( f_{i,\epsilon} )$, and
		\item $D'$ contains no fact outside of $\bigcup_{ i = 1 }^{ 
			n } B_{\epsilon}( f_{i,\epsilon} ) \cup f_\phi\big( B_\epsilon( 
			\tup a_\epsilon ), \WW_\phi \big)$ and contains exactly one fact in
			$f_\phi\big( B_\epsilon( \tup a_{\epsilon} ), \WW_\phi \big)$.
	\end{enumerate}

	Note that in the condition above, we have that $\bigcap_{\epsilon \in
	(0,\epsilon_0)} f_\phi\big( B_{\epsilon}( \tup a_\epsilon), \WW_\phi \big) =
	f_\phi( \tup a, \WW_\phi )$, and that $f_\phi\big( B_\epsilon( \tup
	a_\epsilon ), \WW_\phi \big) \in \FFF_\phi \subseteq \FFF$. 

	Our condition then describes a measurable set in $\DDD \otimes \DDD$, as it
	can be written as a countable intersection over $\epsilon = 1 / k$, small 
	enough, of counting events and preimages of $\app$.
\end{proof}

The relation $\stepE$ is our vehicle for embedding the chase tree $T$ into $(
\DD^{ \omega }, \DDD^{ \otimes\omega } )$.  Note that every edge $(v,w) \in E$
corresponds to $D_v \stepE D_w$. The relation $\vdash$, however contains
\enquote{loops} $D_v \stepE D_v$ for every leaf $v$ of $T$. Moreover, as the
chase tree starts with the fixed root label $D_{\IN}$, there may be pairs of
instances $(D, D')$ with $D \stepE D'$ that do not appear as instance labels in
the tree $T$ at all. We only achieve a direct correspondence between $\stepE$
and the edge relation $E$ of $T$, once we restrict the first component of
$\stepE$ to the set of inner node labels, i.\,e. to $\set{ D_v \with v \in V
\with \App( D_v ) \neq \emptyset }$. Let
\begin{align*}
	\paths( \app ) &\coloneqq 
	\set[\big]{ ( D_0, D_1, \dots ) \in \DD^{ \omega } \with 
	D_i \stepE D_{ i+1 } \text{ for all } i \in \NN }\text{ and}\\
	\paths( \app, D_{ \IN } ) &\coloneqq
	\set[\big]{ ( D_0, D_1, \dots ) \in \paths( \app ) \with
		D_0 = D_{ \IN } }\text.
\end{align*}
The elements of $\paths( \app )$ are called the \emph{paths of $\app$} and
the elements of $\paths( \app, D_{ \IN } )$ are called the \emph{paths of\/ 
$\app$ starting in $D_{ \IN }$}. Note that $\paths( \app )$ is the set of paths
in $T$ where all finite maximal paths have been extended to infinite sequences
by repeating the label of the leaf node infinitely often.  

While the full path space $(\DD^{\omega}, \DDD^{\otimes\omega})$ contains paths
completely unrelated to $\app$ (and $D_{ \IN }$), the sets $\paths( \app )$ and
$\paths( \app, D_{ \IN } )$ only contain relevant paths for the given $\GDL$
program and chase policy.

\begin{remark}
We introduce two separate notions, because we consider two scenarios: the 
single input instance scenario, where we evaluate a $\GDL$ program for a given
input instance $D_{ \IN } \in \DD_{ \IN }$; and the PDB input scenario, where
the input is already a probability distribution over database instances. 
Technically, the former can be cast as a special case of the latter. The single
input instance scenario however, is the one originally described by Bárány et 
al. \cite{Barany+2017}, and we thus prefer to give it an explicit treatment.
\end{remark} 

Using a pairwise intersection of $\stepE$-pairs in $\DD^{ \omega }$, we
immediately obtain the following from \cref{lem:stepE}.

\begin{corollary}\label{cor:seqpaths}
	The sets\/ $\paths( \app )$ and\/ $\paths( \app, D_{ \IN } )$ are
	measurable in $( \DD^{ \omega }, \DDD^{ \otimes\omega } )$.
\end{corollary}

This concludes the embedding of chase trees into the path space 
$( \DD^{ \omega }, \DDD^{ \otimes \omega } )$. We are not interested in the
paths themselves though. Intuitively, finite paths in the chase trees need to
be mapped back to database instances.

\subsection{Limit Instances}\label{ssec:limit}

Recall \cref{fig:seqchasetree}. By now, we have described how a $\GDL{}$ 
program spans a \enquote{computation tree} of database instances. Yet (as 
discussed in \cref{exa:infcomp}, some of these computations (i.\,e. paths in
the tree) are of infinite length. 

The definitions of Bárány et al. \cite{Barany+2017} also cover infinite paths 
with their (per path) probability. We choose a different approach and ignore
infinite paths altogether, which we justify by the following two reasons.

Infinite paths correspond to infinite results of the computation and such
\enquote{infinite instances} are not captured by our framework of standard
PDBs.

For practical means, arguably any interest lies on finite results.  In fact,
Bárány et al. put a strong focus on a restriction to their programs that
guarantees all computation paths to be finite. Yet, we include in our
discussion programs that terminate with a positive probability $\leq 1$. In
probabilistic programming, it is common practice to consider programs with
infinite computation paths and analyze the termination behavior of such
\cite{BournezGarnier2005,Chatterjee+2020}. For example, a $\GDL$ program with
infinite paths is as good as a completely finite one if the probability of all
infinite computation is $0$. 

\bigskip

With the above, we motivate how we wrap up our semantics of $\GDL$ with 
continuous distributions. Further discussion on semantic properties and a 
glimpse into termination behavior is given in \cref{sec:semantics}.

\bigskip

A path $\vec D = ( D_0, D_1, D_2, \dots ) \in \DD^{ \omega }$ is
called \emph{terminating after $i$ steps} if $D_{i} = D_{j}$ for all
$j > i$, but $D_i \neq D_{i-1}$ (in the case $i > 0$). We call a path $\vec D$
\emph{terminating} if such an $i$ exists. The sets of such paths are easily
seen to be measurable in $( \DD^{ \omega }, \DDD^{ \otimes\omega } )$.
\Cref{fig:paths} shows our intended mapping from paths to instances. A special
error event $\Error$ is used as a sink for non-terminating paths.

\begin{figure}[t]\centering%
	\begin{tikzpicture}[sibling distance=.4cm,level distance=1.25cm]
		\node[inner, label={above:$D_{\IN}$}] (root) {}
		child { node[hidden] (a1) {} edge from parent[draw=none] }
		child { node[hidden] (a2) {} edge from parent[draw=none] }
		child { node[hidden] (a3) {} edge from parent[draw=none] }
		child { node[hidden] (a4) {} edge from parent[draw=none] 
			child { node[hidden] (b1) {} edge from parent[draw=none] }
			child { node[hidden] (b2) {} edge from parent[draw=none] }
			child { node[hidden] (b3) {} edge from parent[draw=none] }
			child { node[hidden] (b4) {} edge from parent[draw=none] }
			child { node[hidden] (b5) {} edge from parent[draw=none] }
			child { node[hidden] (b6) {} edge from parent[draw=none] }
			child { node[hidden] (b7) {} edge from parent[draw=none] }
			child { node[hidden] (b8) {} edge from parent[draw=none] }
		}
		child { node[hidden] (a5) {} edge from parent[draw=none] }
		child { node[hidden] (a6) {} edge from parent[draw=none] }
		child { node[hidden] (a7) {} edge from parent[draw=none] }
		child { node[hidden] (a8) {} edge from parent[draw=none] }
		child { node[hidden] (a9) {} edge from parent[draw=none] }
		child { node[hidden] (a10) {} edge from parent[draw=none] };

		\path[draw=b20,thick,fill=b10] 
			(a1.center)  .. controls ++(1.3,.7) ..  (a10.center);
		\draw[thick] (root.center) to (a1.center) to (a10.center) to (root.center);
		\node[inner] at (a4) {};
		\draw[-stealth',b75,thick] (root.center) to (a4);

		\path[draw=b20,thick,fill=b10] 
			(b1.center) .. controls ++(1.6,.7) .. (b8.center);
		\draw[thick] (a4.center) to (b1.center) to (b8.center) to (a4.center);
		\node[inner] at (b4) {};
		\draw[-stealth',thick,b75] (a4.center) to (b4);
		
		\node[font=\small, below right=.1cm and .1cm of b4,align=left] 
			{intermediate\\instance};

		\node[below left=4.5cm and 2cm of a1.center,anchor=center] (c1) {};
		\node[below right=4.5cm and 2cm of a10.center,anchor=center] (c3) {};
		\path (a1.center) to node[pos=.65,anchor=center] (d) {} (c1.center);
		\draw[dotted,thick] (a1.center) to (d.center);
		\draw[dotted,thick] (a10.center) to (c3.center);
		\node[right=1.3cm of d.center,anchor=center] (e) {};
		\node[right=2cm of c1.center,anchor=center] (f) {};
		\node[right=.3cm of d.center,anchor=center] (d1) {};
		\node[below=.5cm of d1.center,anchor=center] (d2) {};
		\node[right=1.7cm of d2.center,anchor=center] (d3) {};
		\node[below=.6cm of d3.center,anchor=center] (d4) {};
		\node[right=1.2cm of d4.center,anchor=center] (d5) {};
		\node[below=.475cm of d5.center,anchor=center] (d6) {};

		\draw[thick,b50,dotted] (d.center) to 
			(d1.center) to (d2.center) to 
			node[black,pos=.7,name=myleaf,leaf,solid]{}
			(d3.center) to (d4.center) to
			(d5.center) to (d6.center);

		\begin{scope}[on background layer]
			\fill[pattern=north east lines,pattern color=b20] 
				(root.center) to (d5.center) to (d6.center) to (c3.center) to 
				(a10.center);
		\end{scope}

		\node[font=\small,below=.1cm of myleaf,align=right,anchor=north east]
			{limit instance};

		\path (d6) to 
			node[shift={(-.2,.2)},rotate=-70,anchor=center] {$\cdots$}
			(c3);

		\path[draw=b75,thick,-stealth',decorate,
				decoration={snake,amplitude=.05cm,post length=.15cm}] 
				(b4) to (myleaf);

		\draw[decorate,decoration={calligraphic brace,amplitude=6pt},thick] 
			([yshift=-1ex]d6.center) 
			to node[below=2ex,font=\small,align=center] 
			{
				terminating paths\\
				$\mapsto$ limit instance in $\DD$
			}
			([yshift=-1ex]c1.center);

		\draw[decorate,decoration={calligraphic brace,amplitude=6pt},thick] 
			([yshift=-1ex]c3.center) 
			to node[below=2ex,font=\small,align=center] 
			{
				non-terminating paths\\
				$\mapsto$ error event $\Error$
			}
			([yshift=-1ex]d6.center);

		\begin{scope}[overlay]	
			\coordinate[right=5cm of root] (D1);	
			\coordinate[below=1.25cm of D1] (D2);
			\coordinate[below=1.25cm of D2] (D3);
			\coordinate[below=1.25cm of D3] (D4);
			\coordinate[below=2cm of D4] (Dinfty);
			\node[inner sep=4pt] (DD1) at (D1) {$(\DD,\DDD)$};
			\node[inner sep=4pt] (DD2) at (D2) {$(\DD,\DDD)$};
			\node[inner sep=4pt] (DD3) at (D3) {$(\DD,\DDD)$};
			\node[inner sep=4pt] (DD4) at (D4) {$\vdots$};

			\draw[->]
				(DD1.south) 
				to [bend left=15] node[right] {$\stepE$}
				(DD2.north);
			\draw[->]
				(DD2.south) 
				to [bend left=15] node[right] {$\stepE$} 
				(DD3.north);
			\draw[->]
				(DD3.south) 
				to [bend left=15] node[right] {$\stepE$}
				(DD4.north);

			\node[minimum width=3em] (DDinfty) at (Dinfty) {};
			\draw[decoration={calligraphic brace,amplitude=5pt}, decorate, thick]
				([yshift=-1ex]DDinfty.east) 
				to node [below=2ex,align=center,font=\small]
				{path space\\$(\DD^\omega, \DDD^{\otimes\omega})$} 
				([yshift=-1ex]DDinfty.west);
		\end{scope}

	\end{tikzpicture}
	\caption{Terminating and non-terminating paths in the sequential chase 
	tree.}
	\label{fig:paths}
\end{figure}

We let $\DD_{ \Error } \coloneqq \DD \cup \set{ \Error }$ denote the augmented 
instance space with the additional error event, and $\DDD_{ \Error } = \DDD
\oplus \set{ \emptyset, \set{\Error} }$ its $\sigma$-algebra. Our mapping 
between $\DD^{\omega}$ and $\DD_{\Error}$ is defined as follows
\begin{equation}\label{eq:liminst}
	\liminst_{\app}(\vec D) \coloneqq
	\begin{cases}
		D_i		& 
		\text{if } \vec D \in \paths( \app ) \text{ and } \vec D \text{ is
		terminating at position } i \text,\\
		\Error	& \text{otherwise.}
	\end{cases}
\end{equation}
If $\liminst_{\app}( \vec D ) \neq \Error$, then $\liminst_{\app}( \vec D )$ is
called the \emph{limit instance} of $\vec D$. Just like $\liminst_{\app}$, we
define $\liminst_{ \app, D_{\IN} }$ by using $\paths( \app, D_{\IN} )$ instead of
$\paths( \app )$ in \labelcref{eq:liminst}. If $\liminst_{ \app, D_{\IN} } (
\vec D ) \neq \Error$, then $\liminst_{ \app, D_{ \IN } } ( \vec D )$ is called
the \emph{limit instance of\/ $\vec D$ in\/ $T_{ \app, D_{ \IN } }$}. We
emphasize that the effect of $\lim_{\app}$ is indeed that terminating chase
paths are mapped to the instance they terminate in, whereas all infinite chase
paths are collectively mapped onto the error event $\bot$.

\begin{lemma}\label{lem:seqlimbim}
	Both\/ $\liminst_{ \app }$ and\/ $\liminst_{ \app, D_{ \IN } }$ are 
	bimeasurable.
\end{lemma}

\begin{proof}
  We only show the assertions for $\liminst_{\app}$. The corresponding
  results for $\liminst_{\app,D_{\IN}}$ are obtained the same way.
  
  \smallskip
  
  First note that for all $\vec D \in \DD^{\omega}$ it holds that
  $\liminst_{ \app }( \vec D ) = \Error$ if and only if either
  $\vec D \notin \paths( \app )$, or $\vec D$ is not
  terminating. Thus,
  $\liminst_{ \app }^{ -1 }( \set{ \Error } ) \in \DDD^{ \otimes
    \omega }$.
  
  For $\DDDD \in \DDD$, it holds that
  $\vec D \in \liminst_{ \app }^{ -1 }( \DDDD )$ if and only if there
  exists $i \in \NN$ such that $\vec D$ is a path of $\app$
  terminating at position $i$ with the property that
  $\vec D \in \pi_i^{-1}(  \DDDD )$ (i.\,e. $\pi_i( \vec D ) \in \DDDD$), where
  $\pi_i$ is the (measurable) projection to the $i$th coordinate. Thus,
  $\liminst_{ \app }^{ -1 } ( \DDDD ) \in \DDD^{ \otimes \omega }$,
  and together with the above, $\liminst_{\app}$ is
  $( \DDD^{ \otimes\omega }, \DDD_{ \Error } )$-measurable.
  
  \smallskip
  
 It remains to prove that $\liminst_{\app}$ maps measurable sets to
 measurable sets. \cref{lem:seqinj} implies that $\liminst_{\app}$ is
  injective on $\liminst_{\app}^{-1}( \DD )$. The lemma only applies
  to the chase tree for a fixed input instance $D_{\IN}$, but as the
  input instance remains part of any subsequent instance in the
  evaluation of a Datalog program, it directly extends to all input
  instances. As $\liminst_{\app}$ is measurable, so is its restriction to
  $\liminst_{\app}^{-1}( \DD ) \in \DDD^{\otimes\omega}$ (with respect to
  $\DDD^{ \otimes\omega }
  \restriction_{ \liminst_{ \app }^{-1}( \DD )}$). Since $(\DD,\DDD)$ and $(
  \DD^{\omega}, \DDD^{\otimes\omega} )$ are standard Borel
  (\cref{pro:pdbstandardborel}), the assertion follows from
  \cref{fac:measinjectionimages}. \qedhere
\end{proof}

\subsection{Chase Trees as Markov Processes}\label{ssec:seqmarkov}

In this subsection, we establish a correspondence between a chase tree
for a given \GDL{} program and a discrete-time Markov process whose
state space is the (in general not countable) space of database
instances. We have seen in the previous subsection how paths in a
chase tree naturally correspond to a set of \emph{paths} of such a
process (cf. \cref{sec:stochproc}) in the countably infinite product space
$(\DD^\omega, \DDD^{\otimes \omega})$.

To obtain the correspondence to a Markov process, we need to show that the 
probabilistic transitions that are encoded within the nodes of any level of the 
chase tree, or, to be more precise, by its measurable chase policy, describe 
a stochastic kernel from $(\DD, \DDD)$ to itself.

The interpretation of the \GDL{} semantics as a database-valued Markov process
(which, by itself, was already recognized in \cite[p. 22:14]{Barany+2017}) 
makes also apparent that the natural generalization of the \GDL{} language is 
to allow the input to be a (sub\nobreakdash-)prob\-a\-bilis\-tic database
rather than a single instance. A \GDL{} program then induces a mapping from a
(sub-)probabilistic database to a sub-prob\-a\-bilis\-tic database 
(\enquote{losing} the mass of \enquote{non-terminating} paths). We will come
back to this at the end of the subsection.

\bigskip

We extend $\stepE[\app]$ to a function $\stepP[\app] \from \DD \times \DDD \to
[0,1]$ where again, if the reference is clear, we just write $\stepP$. For $D
\in \DD$ and $\DDDD \in \DDD$ we distinguish two cases.
\begin{itemize}
	\item If $\App( D ) \neq \emptyset$ with $\app( D ) = ( \phi, \tup a )$ such 
		that $\smash{\chasestep{ D }{ (\phi,\tup a) }{ (\EEEE, \mu) }}$ is the 
		corresponding chase step, then $\stepP( D, \DDDD ) \coloneqq \mu( \DDDD
		\cap \EEEE ) = \int \xi( D, \phi, \tup a, \placeholder, \DDDD \cap \EEEE 
		) \cdot \psi_{\phi}\params{\tup a}\:\d{\mu_{\phi}}$.
	\item If $\App( D ) = \emptyset$, we let $\stepP( D, \DDDD ) = \iota( D, 
		\DDDD )$ where $\iota$ is the identity kernel.
\end{itemize}
Intuitively, $\stepP( D, \DDDD ) = \Pr( D \stepE \DDDD )$ with the latter
referring to the probability space for the chase step starting in $D$. The
following proposition resolves the main technical obstacle for turning
measurable chase policies and sequential chase trees into Markov processes.

\begin{proposition}\label{pro:seqstepkernel}
	$\stepP$ is a stochastic kernel.
\end{proposition}

\begin{proof}
	Clearly, $\stepP( D, \placeholder )$ is a probability measure for all $D \in
	\DD$. The complicated part of the proof is establishing that $\stepP( 
	\placeholder, \DDDD )$ is $( \DDD_{ \phi }, \Borel[0,1] )$-measurable for
	all $\DDDD \in \DDD$.

	Let $\DDDD \in \DDD$ be fixed. It suffices to show for all rules $\phi$ that 
	$\stepP( \placeholder, \DDDD )$ is $(\DDD_\phi, \Borel[0,1])$-measurable
	where $( \DD_\phi, \DDD_\phi )$ is the restriction of $( \DD, \DDD )$ to
	the instances $D$ with $\App( D ) \neq \emptyset$ and $\app( D ) = ( \phi,
	\tup a )$ for some $\tup a \in \VV_\phi^{(h)}$. Note that the corresponding 
	statement for the restriction to the set of instances $D$ with $\App(D) = 
	\emptyset$ clearly holds, as $\stepP$ is the identity kernel in this case.

	Thus, let $\phi$ be a fixed rule of $\G^{ \exists }$. Without restriction
	(see the discussion below \cref{def:seqstep}), we always treat $\phi$
	as an existential rule. By \cref{fac:gaudardhadwin}, the function $\kappa_1
	\from \VV_\phi^{(h)} \times \WWW_\phi \to [0, 1]$ with
	\[
		\kappa_1 ( \tup a, \BBBB ) \coloneqq 
		\int_{ \BBBB } \psi_{\phi}\params{ \tup a } \:\d{\mu_\phi}
	\]
	is a stochastic kernel from $\VV_\phi^{(h)}$ to $\WW_\phi$. Thus, the function
	$\kappa_2 \from ( \DD_\phi \times \VV_\phi^{(h)} ) \times \WWW_\phi \to [0
	,1 ]$ with
	\[
		\kappa_2( D, \tup a, \BBBB ) \coloneqq \kappa_1( \tup a, \BBBB )
	\]
	is a stochastic kernel from $\DD_\phi \times \VV_\phi^{(h)}$ to $\WW_\phi$.
	The function $\xi( \placeholder, \phi, \placeholder, \placeholder, \DDDD )
	\from \DD_\phi \times \VV_\phi^{(h)} \times \WW_\phi \to \set{ 0 ,1 }$ with
	\[
		\xi( \placeholder, \phi, \placeholder, \placeholder, \DDDD ) \with
		( D, \tup a, b ) \mapsto \xi( D, \phi, \tup a, b, \DDDD )
	\]
	is $(\DDD_\phi \otimes \VVV_\phi^{(h)} \otimes \WWW_\phi, 
	\Borel[0,1])$-measurable, as it is the characteristic function of the
	$\phi$-section of $\ext^{-1}( \DDDD )$, and $\ext^{-1}( \DDDD ) \in 
	\DDD_\phi \otimes \VVV_\phi^{(h)} \otimes \WWW_\phi$ by \cref{pro:extExtmeas}.
	Using \cref{fac:kernelandfunction} for $\xi( \placeholder, \phi,
	\placeholder, \placeholder, \DDDD )$ and $\kappa_2$, the function
	\[
		g \from \DD_\phi \times \VV_\phi^{(h)} \to [0, 1] \with
		 ( D, \tup a ) \mapsto
		 \int_{ \WW_\phi } \xi( D, \phi, \tup a, \placeholder, \DDDD ) \cdot
		 \psi_\phi\params{\tup a }\:\d{\mu_\phi}
	\]
	is $( \DDD_\phi \otimes \VVV_\phi^{(h)}, \Borel[0,1] )$-measurable. Observe that
	the function $h \from \DD_\phi \to \DD_\phi \times \VV_\phi^{(h)}$ with
	\[
		h( D ) \coloneqq ( D, \tup a )
	\]
	is $( \DDD_\phi, \DDD_\phi \otimes \VVV_\phi^{(h)})$-measurable by the 
	measurability of $\app$, and using \cref{fac:families}. Thus, for all $x \in
	[0, 1]$, it follows that
	\[
		\stepP( \placeholder, \DDDD )^{ -1 }[ 0, x )
		= \set{ D \in \DD \with \stepP( D , \DDDD ) < x }
		= \big( h^{-1} \after g^{-1} \big) [ 0, x ) \in \DDD\text,
	\]
	entailing the claim.\qedhere
\end{proof}

\Cref{pro:seqstepkernel} directly implies the following by Kolmogorov's
existence theorem (\cref{fac:kolmogorov}).

\begin{corollary}\label{cor:seqmarkov}
	Let $\app$ be a measurable chase policy and let $(\DD, \DDD, P)$ be a
	sub-probabilistic database. Then there exists a Markov process with state
	space $( \DD, \DDD )$, with initial distribution $P$ and with transition
	kernels $\stepP$.
\end{corollary}

Note that every sub-probability measure $\vec P$ on the path space
$( \DD^{ \omega }, \DDD^{ \otimes \omega } )$ defines a
push-forward sub-probability measure
$\smash{\vec P \after \liminst_{\app}^{ -1 }}$ (or
$\smash{\vec P \after \liminst_{\app,D_{\IN}}^{-1}}$, respectively) on
$( \DD_{ \Error }, \DDD_{ \Error } )$. Thus, every Markov process like
in \cref{cor:seqmarkov} defines an output sub-probabilistic
database. The semantics of our $\GDL$ program $\G$, finally, is said
output.

\begin{theorem}\label{thm:seqspdb}
	Let\/ $\app$ be a measurable chase policy.
	\begin{enumerate}
		\item For all $D_{ \IN } \in \DD_{ \IN }$, the program $\G$ on input 
			$D_{\IN}$ defines a sub-probabilistic database $\G_{\app} (D_{\IN})$
			with respect to $\app$.
		\item For all sub-probabilistic databases $\D = ( \DD_{\IN}, \DDD_{\IN}, 
			P )$, the program $\G$ on input $\D$ defines a sub-probabilistic
			database $\G_{\app}( \D )$ with respect to $\app$.\qedhere
	\end{enumerate}
\end{theorem}

For the first part of the above theorem, we let the initial distribution of
\cref{cor:seqmarkov} be the Dirac one on the instance $D_0$.%
\footnote{For the initial distributions, just a measure is needed. In
particular, the problems of the Dirac distribution with respect to
\cref{fac:gaudardhadwin} pointed out in \cref{ssec:paramdist} are irrelevant
here.}
For the second part, the initial distribution is the sub-probability
distribution of the input sub-probabilistic database. Note that even if the
input $\D$ is a probabilistic database (i.\,e. $P$ has total mass $1$), the
output $D_{\app}$ may be a proper sub-probabilistic database.

\begin{remark}\label{rem:endproj}
	In the end, we might want to get rid of the auxiliary relations that
	were created in the translation to the \EDL{} program. This can be done in a
	measurable way by a relational algebra view (cf.
	\cref{fac:measurableviews}), yielding again a sub-probabilistic database.
\end{remark}

\section{Parallel Probabilistic Chase}\label{sec:parchase}

We obtain another variant of the chase procedure if we allow all applicable
rules to fire simultaneously. This notion of parallel chase is of interest as
it is not depending on having a measurable chase policy at hand. 

\begin{remark}[continues=rem:independenceass]
Recall our discussion of independence assumptions and the Markov property from
the beginning of \cref{sec:seqchase}. As there, the parallel chase has
independent sampling in the absence of logical dependencies. This is, in
particular, the case for multiple rules firing in parallel in a single parallel
step.
\end{remark}

In this section, we again fix a $\GDL$ program $\G$, with its $\EDL$ version
$\G^{\exists}$, and assume that $\G^{\exists} = \set{ \phi_1, \dots, \phi_k }$.

\subsubsection*{Structure of this section}
In the following, we introduce the parallel chase for the \GDL{} language
(\cref{ssec:parstep}) and construct a Markov process (\cref{ssec:parmarkov}) as
in the sequential case. Most of the definitions and results are modest
extensions of their counterparts in \cref{sec:seqchase}, thus allowing a
briefer presentation.

\subsection{Parallel Chase Steps and the Parallel Chase Tree}
\label{ssec:parstep}

If $D$ is a database instance, the \emph{firing configuration} of $D$ is the
tuple $\vec\ell(D) = (\ell_1,\dots,\ell_k)\in\NN^k$ where $\ell_i = \lvert
\set{ \vec a \with (\hat\phi_i,\tup a)\in\App(D) } \rvert$ for all $1 \leq i
\leq k$. Note that the set $\DD_{\vec\ell}$ of database instances having a 
fixed firing configuration $\vec\ell$ is measurable in $(\DD,\DDD)$, since we
know that $\App$ corresponds to an Relational Calculus view and the 
cardinalities in question can be obtained by a counting aggregation. This 
yields a measurable mapping by \cref{fac:measurableviews}.

\begin{definition}[Parallel Chase Step]\label{def:parstep}
	A \emph{parallel chase step} for $\G$ is a tuple $(D, \App( D ), \EEEE, \mu 
	)$ where
	\begin{itemize}
		\item $D$ is a database instance in $\DD_{\App}$, say with firing 
			configuration $\vec \ell = (\ell_1, \dots, \ell_k )$ such that
			\[
				\App( D ) = 
				\set[\big]{ 
					( \phi_i, \vec a_{ij_i} ) \with
					1 \leq i \leq k \text{ and } 1 \leq j_i \leq \ell_i
				}\text.
			\]
		\item $\EEEE$ is the event
			\[
				\EEEE = 
				\Ext_{ \vec \ell }
				\big( 
					D, \XXXX_{11}, \dots, \XXXX_{ 1\ell_1 }, \dots,
					\XXXX_{k1}, \dots, \XXXX_{ k\ell_k }
				\big)
			\]
			with $\XXXX_{ij_i} = \set[\big]{ ( \phi_i, \vec a_{ ij_i }, b ) \with
			b \in \WW_{ \phi_i } }$, and
		\item $\mu$ is the probability measure on $\DDD \restriction_{ \EEEE }$
			that is defined by
			\begin{equation}\label{eq:parstepmeas}
				\mu( \DDDD )
				= \int_{ \WW } \Xi_{ \vec\ell }
				\big( 
					D, \phi_1, \vec a_{ 11 }, b_{ 11 }, \dots, 
					\phi_k, \vec a_{k\ell_k}, b_{k\ell_k}, \DDDD 
				\big) \cdot \prod_{ i = 1 }^{ k } \prod_{ j_i = 1 }^{ \ell_i }
				\psi_i\params{ \vec a_{ij_i} } (b_{ij_i})
				\:\d{\mu^{\otimes}}
			\end{equation}
			for all measurable $\DDDD \subseteq \EEEE$, where $\psi_i$ is the
			parameterized distribution in the rule of $\G$ that $\phi_i$ originated 
			from, $\WW = \prod_{ i = 1 }^{ k } \WW_{ \phi_i }^{ \ell_i }$ and
			$\mu^{\otimes}$ is the product measure $\mu^{\otimes} = \bigotimes_{ 
			i = 1 }^{ k }\mu_{ \psi_i }^{ \otimes\ell_i }$ on $\WW$ with 
			$\mu_{\psi_i}$ being the measure underlying the parameterized
			distribution $\psi_i$ (cf. \cref{ssec:paramdist}).\qedhere
	\end{itemize}
      \end{definition}

Note that as before (see \cref{ssec:seqstep}), we use a dummy integration 
for non-existential rules to enable the unified expression of
\labelcref{eq:parstepmeas}. Then, we again concentrate on the existential rules
and interpret the deterministic ones as special cases (cf.\
\cref{rem:determpdist}).

\begin{remark}
	Note that by definition of $\Ext_{\vec\ell}$, the case where multiple
	deterministic rules with the same left-hand sides are applicable, is nicely
	resolved, and the corresponding fact is only added once in the extension, 
	i.e., the results from applications of deterministic \mbox{rules may collapse}.

	Recall that it may well happen that $\App(D)$ contains multiple pairs with
	first component an existential rule $\phi$, for example $(\phi,\tup a)$ and
	$(\phi,\tup a')$. However, for any related sample outcomes $b$ and $b'$ it
	holds that $f_{\phi}(\tup a,b) \neq f_{\phi}(\tup a',b')$ because $\tup a
	\neq \tup a'$. That is, if $\phi_i$ is existential, any follow-up instance
	will contain $\ell_i$ new facts with the relation symbol from the head of 
	$\phi$. In particular, the results from applications of existential rules
	do not collapse.
\end{remark}

\begin{lemma}\label{lem:parstepmeas}
	The function $\mu$ from \cref{def:parstep} is well-defined, and a
	probability measure on $\DDD\restriction_{\EEEE}$.
\end{lemma}

\begin{proof}
	This can be shown analogously to \cref{lem:seqstepmeas}. That the integral 
	is well-defined follows from the measurability of $\Xi_{ \vec \ell }$ and
	the measurability of the $\psi_i$. For the integration, we use a product
	density of our parameterized distributions. Thus, in the case $\EEEE = 
	\DDDD$, \labelcref{eq:parstepmeas} collapses to $\int \prod_i \prod_{j_i}
	\psi_i\params{\vec a_{ij_j}} \:\d{\mu^{\otimes}} = 1$ and it follows that
	$\mu$ is a probability measure.
\end{proof}

The use of the product density in \labelcref{eq:parstepmeas} stipulates the 
independence assumption we discussed before---all probabilistic rules that fire
together in a parallel chase step do so independently. Note that by
\cref{fac:fubini} and the definition of $\Xi_{\vec\ell}$, the concrete order of
the tuples $(\phi_i, \vec a_{ij}, b_{ij} )$ has no impact on $\mu$ whatsoever.
Also note that for firing configurations $\vec \ell$ with $\ell_i = 0$ for all
but one $i_0$ and $\ell_{i_0} = 1$, \labelcref{eq:parstepmeas} coincides with
\labelcref{eq:seqstepmeas}. \Cref{fig:parstep} contains an illustration of a
parallel chase step.

\begin{figure}[H]\centering%
	\begin{tikzpicture}[level distance=3.5cm, sibling distance=2cm]
		\node[inner,label={above:$D$}] (D) {}
			child { node (A1) {} edge from parent[draw=none] }
			child { node (A2) {} edge from parent[draw=none] }
			child { node[inner] (A3) {} edge from parent[draw=none] }
			child { node (A4) {} edge from parent[draw=none] };

		\begin{scope}[on background layer]
			\path[fill=white!98!black,draw,thick] 
				(D.center) to (A1.center) to (A4.center) to (D.center);
		\end{scope}

		\node[below=2ex of A3] (A3L) 
			{$\mathllap{\Ext(D, \phi_1, \tup a_{11}, b_{11}, \dots, \phi_k, \tup a_{k\ell_k}, b_{k\ell_k}) ={}} 
				D' \mathrlap{{}\in\EEEE}$};
			\draw[b50,<-,shorten >=2pt,inner sep=1pt] (A3L) to (A3);

		\draw (D.center) to (A1.center);
		\draw (D.center) to (A4.center);
		\draw[-{Stealth[round,width=8pt,length=6pt]},thick,double] (D.center) to 
			node[xshift=-.85cm,yshift=0cm] {$b_{11}, \dots b_{k\ell_k}$} (A3);
		\draw (A1.center) to (A4.center);
		\node[inner] at (D) {};

		\node[left=0cm of A1] {$\EEEE\colon$};
		\node[right=0cm of A4] {$\phantom{\EEEE\colon}$};

		\begin{scope}[on background layer]
			\draw[fill=b10,draw=b20,very thick] plot[smooth] coordinates 
				{(A1.center) (-.6,-2.5) (.3,-2.7) (.9,-2.6) (A4.center)};
			\clip ([xshift=-1cm]A2.center) rectangle ([shift={(1,1.5)}]A2.center);
			\draw[fill=b20,draw=b50,very thick] 
				plot[smooth] coordinates
				{(A1.center) (-.6,-2.5) (.3,-2.7) (.9,-2.6) (A4.center)};
			\path[clip] plot[smooth] coordinates 
				{(A1.center) (-.6,-2.5) (.3,-2.7) (.9,-2.6) (A4.center)};
			
			\draw[b50,very thick] 
				([xshift=-.98cm]A2.center) to ([shift={(-.98,1)}]A2.center);
			\draw[b50,very thick] 
				([xshift=.98cm]A2.center) to ([shift={(.98,1)}]A2.center);
		\end{scope}

		\node[text=b20] at (1.3,-2.45) {$\mu$};
		\node[text=b50] at (-1.4,-2.4) {$\mu(\DDDD)$};
		\draw[thick] (0,0) +(-165:.3) arc (-165:-15:.3);
	\end{tikzpicture}
	\caption{Illustration of a parallel chase step, continuous case.}
	\label{fig:parstep}
\end{figure}

We denote a parallel chase step $(D, \App(D), \EEEE, \mu )$ as
\[
	\pchasestep{ D }{ \App(D) }{ ( \EEEE, \mu ) }\text.
\]
Note that a parallel chase step is already determined by $D$ (and $\App$)
alone. That is, if there exists a parallel chase step starting in $D$, it is 
unique. This is in contrast to the sequential chase step that additionally 
depends on a measurable chase policy $\app$ of $\App$.  In the following, we
adapt the definitions from \cref{ssec:seqstep,ssec:seqpaths,ssec:limit} to the
parallel setting. For the definition of a chase tree, all we have to do is take
the definition of the sequential chase tree, and replace the sequential chase
steps with parallel ones.

\begin{definition}[Parallel Chase Tree]
	Let $D_{ \IN } \in \DD_{ \IN }$. The \emph{parallel chase tree} $T_{ \App,
	D_{ \IN } }$ for input instance $D_{ \IN }$  with respect to the $\GDL$
	program $\G$ is a labeled countable-depth tree $T_{ \App, D_{ \IN } } = ( V,
	E, \Lambda )$ with labeling function $\Lambda$, root node $r \in V$, and the
	following properties.
	\begin{enumerate}
		\item The root node $r$ is labeled with $D_{ \IN }$.
		\item If $v \in V$ is a leaf node, then it is labeled with an instance
			$D_v$ such that $\App( D_v ) = \emptyset$.
		\item If $v \in V$ is an inner node, then it is labeled with an instance
			$D_v$ such that $\App( D_v ) \neq \emptyset$ and
			\begin{enumerate}
				\item $\pchasestep{ D_v }{ \App(D_v) }{ (\EEEE_v, \mu_v) }$ is a
					chase step with $( \EEEE_v, \mu_v )$ as in \cref{def:parstep},
					and
				\item the function $v' \to D_{v'}$ is a bijection between the
					children of $v$ and $\EEEE_v$.
			\end{enumerate}
	\end{enumerate}
\end{definition}

For all $D_{ \IN } \in \DD_{ \IN }$, the parallel chase tree $T_{ \App, D_{\IN}
}$ is unique. As for the sequential chase tree, $T_{ \App, D_{ \IN } }$ has 
paths of at most countable length but may contain nodes with uncountably many
children. Again, the tree is labeled injectively. This can be shown just like
\cref{lem:seqinj}.

\begin{lemma}\label{lem:parinj}
	Let $D_{ \IN } \in \DD_{ \IN }$. Then $v \neq w$ implies $D_v \neq D_w$ for
	all $v \neq w$ in $T_{ \App, D_{\IN} }$.
\end{lemma}

We proceed with the introduction of parallel versions of the various relations
and functions encountered in \cref{ssec:seqpaths,ssec:limit}. First, we define a
parallel version of the relation $\stepE$ from \cref{ssec:seqpaths}. Our new
relation $\mathord{\StepE} \subseteq \DD^2$, again denoted in infix notation,
is defined just like $\mathord{\stepE}$, except that the sequential chase step
is once more replaced by a parallel one:
\begin{itemize}
	\item If $\App( D ) = \emptyset$, then $D \StepE D'$ if and only if $D =
		D'$.
	\item If $\App( D ) \neq \emptyset$, there exists a unique parallel chase
		step $\smash{\pchasestep{D}{ \App(D) }{ ( \EEEE, \mu ) }}$ starting in 
		$D$. In this case, $D \StepE D'$ if and only if $D' \in \EEEE$.
\end{itemize}
We let
\begin{align*}
	\paths( \App ) 
	& \coloneqq \set[\big]{ 
		( D_0, D_1, \dots ) \in \DD^{ \omega } \with
		D_i \StepE D_{i+1} \text{ for all } i \in \NN 
	} \text{ and}\\
	\paths( \App, D_{\IN} )
	& \coloneqq \set[\big]{
		( D_0, D_1, \dots ) \in \paths( \App ) \with D_0 = \D_{ \IN }
	}\text.
\end{align*}

Again, $\paths( \App, D_{ \IN })$ corresponds to the paths in $T_{\App,D_0}$ 
where finite paths are continued by repeating the last instance label
infinitely often. We keep using the notions of \emph{terminating} paths and
paths \emph{terminating at position $i$} from \cref{ssec:seqpaths}.

\bigskip

To map paths to instances in the parallel setting, we define $\liminst_{\App}
\from \DD^{ \omega } \to \DD_{ \Error }$ by
\begin{equation}\label{eq:limApp}
	\vec D \mapsto
	\begin{cases}
		D_i
		& \text{if } \vec D \in \paths( \App ) \text{ and } \vec D \text{ is
			terminating at position } i\text,\\
		\Error
		& \text{otherwise.}
	\end{cases}
\end{equation}
Again, we define a version $\liminst_{\App,D_{\IN}}$ of $\liminst_{\App}$ where
using $\paths( \App, D_{ \IN } )$ instead of $\paths(\App)$ in 
\labelcref{eq:limApp}.

The relations, sets and functions we just defined enjoy properties analogous to
their sequential counterparts (cf. 
\cref{lem:stepE,cor:seqpaths,lem:seqlimbim}).

\begin{lemma}\label{lem:StepE}\leavevmode%
  \begin{enumerate}
  \item \label{itm:StepE1}
    The relation $\mathord{\StepE} \subseteq \DD^2$ is measurable in 
    $\DDD \times \DDD$.
  \item The sets\/ $\paths( \App )$ and\/ $\paths( \App, D_{\IN} )$ are 
    measurable in $( \DD^{ \omega }, \DDD^{ \otimes\omega } )$.
  \item Both $\liminst_{ \App }$ and $\liminst_{ \App, D_{\IN} }$ are
    bimeasurable.\qedhere
  \end{enumerate}
\end{lemma}

\begin{proof}
  The proofs are easy extensions of the proofs of 
  \cref{lem:stepE,cor:seqpaths,lem:seqlimbim}.
  \begin{enumerate}
  \item Let $\vec \ell$ be a fixed firing configuration. For the set of
    instances having firing configuration $\vec \ell$, we can decompose
    $\App$ into a sequence of $\sum_{i} \ell_i$ measurable selections as
    follows. We start with $\App_0 \coloneqq \App( D)$ and note that
    $\size{ \App_0( D ) } = \sum_i \ell_i$. As $\App$ is a measurable
    multifunction, it has a measurable selection $\app_1$. We set $\App_1(
    D ) \coloneqq \App( D ) \setminus \set{ \app_1( D ) }$. One can easily show
    that $\set{ D \in \DD \with \App_1( D) \neq \emptyset }$ is
    measurable, and that $\App_1$ is a measurable multifunction on this set
    (as long as it is non-empty). Repeating this process for $j = 1,
    \dots, \sum_i \ell_i$, we obtain functions $\app_1, \dots,
    \app_{\sum_i\ell_i}$. 
    If $(\phi^{(i)},\tup a^{(i)}) = \app_i( D )$, then for all $b^{(i)}
    \in \WW_{\phi^{(i)}}$ it holds that
    \begin{equation}\label{eq:extExt}
      \Ext_{\vec\ell}\big( D, (\phi^{(i)},\tup a^{(i)},b^{(i)}) \big)
      = \ext\big( \dots ( 
      \ext( D, \phi^{(1)}, \tup a^{(1)}, b^{(1)} ) \dots ),
      \phi^{(\sum_i\ell_i)}, \tup a^{(\sum_i\ell_i)} \big)
    \end{equation}
    To see this note that on any database instance $D$ with
    $\size{\App_0(D)} = \sum_i \ell_i$, our sequence of $\app$-functions
    yields a sequence of pairs $(\phi,\tup a)$. The only situation in
    which several of these pairs yield the same resulting tuple
    $f_{\phi}(\tup a,b) = f_{\phi'}(\tup a',b')$ is when $\phi$ and
	 $\phi'$ are two different deterministic rules (with the same head), with 
	 $\tup a = \tup a'$ and $\tup b = \tup b' = \ast$. Thus, there are two pairs 
	 $(\phi,\tup a)$ and $(\phi',\tup a)$ among the sequence $( \app_i(D) )_i$. 
	 Even though the one of these two that appears \enquote{later} is not an 
	 applicable pair anymore in the intermediate instance constructed in
	 \labelcref{eq:extExt}, this is no problem, for if $f_{\phi}(\tup a,*)$
    is already contained in an intermediate instance $D'$, then by definition 
	 $\ext(D',\phi',\tup a,*) = D'$.
	 
	 Now with our sequence of $\app$-functions (resp. of pairs $(\phi^{(i)},
	 \tup a^{(i)})$), we can proceed as in
	 \cref{lem:stepE}, where we showed the measurability of the set 
	 \labelcref{eq:vdashphi}. There, we described the membership of a pair
	 $(D,D')$ in the relation $\vdash$ by countably approximating $D$ and $D'
	 = D \cup \set{ f_\phi( \tup a, b ) }$ (for some $b$). The only thing that
	 changes is that now $D' = D \cup \set{ f_{\phi_1}( \tup a^{(1)},b^{(1)} ), \dots,
	 f_{\phi_k}( \tup a^{(\sum_i\ell_i)}, b^{(\sum_i \ell_i )}) }$ (for some $b^{(1)},
	 \dots, b^{(\sum_i\ell_i)}$).
  \item This follows immediately from \labelcref{itm:StepE1}.
  \item This can be shown just like \cref{lem:seqlimbim}, with the only
    change being the use of $\App$ instead of $\app$, and using
    \cref{lem:parinj} instead of \cref{lem:seqinj}.\qedhere
  \end{enumerate}
\end{proof}

\subsection{The Markov Process for Parallel Chasing}\label{ssec:parmarkov}

In analogy to \cref{ssec:seqmarkov}, we show in this section how the
parallel chase defines a Markov process of database instances.

We define $\StepP \from \DD \times \DDD \to [0, 1]$ as follows. Let 
$D \in \DD$ and $\DDDD \in \DDD$. 
\begin{itemize}
	\item If $\App( D ) \neq \emptyset$ and $\pchasestep{ D }{ \App(D) }
		{ ( \EEEE, \mu ) }$ is the corresponding parallel chase step, then
		$\mathop{\StepP}( D, \DDDD ) = \mu( \DDDD \cap \EEEE )$ with $\mu$ as in
		\cref{def:parstep}.
	\item If $\App( D ) = \emptyset$, then $\mathop{\StepP}( D, \DDDD ) = 
		\iota( D, \DDDD )$ where $\iota$ is the identity kernel.
\end{itemize}
Intuitively, $\StepP( D, \DDDD ) = \Pr( D \StepE \DDDD )$ with the latter
referring to the probability space defined for the chase step starting in $D$.

\begin{proposition}\label{pro:parstepkernel}
	$\StepP$ is a stochastic kernel. 
\end{proposition}

The proof is similar to that of \cref{pro:seqstepkernel}.

\begin{proof}
	By \cref{lem:parstepmeas} and the definition of $\StepP$, the function
	$\StepP(D,\placeholder)$ is a probability measure for all $D \in \DD$.
	So again, the harder part is to show that $\StepP(\placeholder,\DDDD)$ is
	$(\DDD,\Borel[0,1])$-measurable. For this, it suffices to show that the 
	restriction of $\StepP( \placeholder, \DDDD )^{-1}[0, \alpha)$ to any
	fixed firing configuration $\vec\ell$ is in $\DDD$ for all $\alpha \in [0,
	1]$. The claim then follows since there are only countably many $\vec\ell$,
	since the set $\DD_{\vec\ell} = \set{ D \in \DD \with \vec\ell( D ) =
	\vec\ell }$ is measurable, and since the sets $[0,\alpha)$ generate
	$\Borel[0,1]$.

	We use the trick of the proof of \cref{lem:StepE}\labelcref{itm:StepE1} and
	decompose $\App$ into multiple measurable selections. With these selections
	we construct a measurable function $D \mapsto (D, \phi_1, \tup a_{11}, 
	\dots, \phi_k, \tup a_{k\ell_k} )$. With a repeated application of Fubini's
	theorem to the definition of $\mu$ (see \cref{def:parstep}), we obtain
	that for all $\DDDD$ it holds that
	\begin{align*}
		\mu( \DDDD ) &=
		\begin{multlined}[t]
		\sum_{ \vec\ell } 
		\int_{ \WW } \Xi_{\vec\ell}\big(
			D, \phi_1, \tup a_{11}, b_{11}, \dots, 
			\phi_k, \tup a_{k\ell_k}, b_{k\ell_k},
			\DDDD \cap \DD_{\vec\ell} 
		\big)\\
		\cdot \prod_{ i = 1 }^{ k }\prod_{ j_i = 1 }^{ \ell_i } 
		\psi_{i}\params{\tup a_{ij_i}}( b_{ij_i} )
		\d{\Big( \mu_{\psi_1}^{\otimes\ell_1} \otimes \dots \otimes 
		\mu_{\psi_{k}}^{\otimes\ell_k}\Big)}
		\end{multlined}\\
						 &{}={}
						 \begin{multlined}[t]
		\sum_{ \vec\ell }
		\int_{ \WW_1 } \psi_1\params{\tup a_{11}}(b_{11}) \dotsm
		\int_{ \WW_k } \psi_k\params{\tup a_{k\ell_k}}(b_{k\ell_k})\\
		\cdot \Xi_{\vec\ell}\big( \phi_1, \tup a_{11}, b_{11}, \dots, \phi_k, 
		\tup a_{k\ell_k}, b_{k\ell_k}, \DDDD\cap \DD_{\vec\ell} \big)
		\d{\mu_{\psi_k}} \dots \d{\mu_{\psi_1}}
		\end{multlined}
		\end{align*}
	(Note that there are $\ell_1 + \dots + \ell_k$ integrals in every summand
	of the last expression above.) In this situation, proceeding exactly like
	in the proof of \cref{pro:seqstepkernel}, one can show that the innermost
	integral is 
	\begin{equation}\label{eq:longmeasurablefn}
		\Big( \DDD \otimes \UUU^{\otimes\ell_1}_1 \otimes \dots \otimes 
		\UUU_{k-1}^{\otimes{\ell_{k}-1}} \otimes \powerset\big( \set{ \phi_k }
			\big) \otimes \VVV_{\phi_k}^{(h)}, \Borel[0,1]\Big)\text{-measurable,}
	\end{equation}
	where $\UUU_i$ is the $\sigma$-algebra of $\set{\phi_i}\times \VV_{\phi_i}^{(h) }\times 
	\WW_{\phi_i}$. Using the sequence of measurable functions we introduced
	before, we can get rid of the trailing $\powerset\big( \set{\phi_k} \big)
	\otimes \VVV_{\phi_k}^{(h)}$ in \labelcref{eq:longmeasurablefn} just like in
	the proof of \cref{pro:seqstepkernel}. Propagating the above procedure
	outwards yields that (the mentioned restriction of) $\StepP$ is
	$(\DDD,\Borel[0,1])$-measurable.
\end{proof}

By the merit of \cref{pro:parstepkernel}, we can construct a Markov process
analogously to \cref{ssec:seqmarkov}.

\begin{corollary}\label{cor:parmarkov}
	Let $( \DD, \DDD, P )$ be a sub-probabilistic database. Then there exists
	a Markov process with state space $(\DD, \DDD)$, with initial distribution
	$P$ and with transition kernels $\StepP$.
\end{corollary}

Mapping the paths of the process back to instances using $\liminst_{\App}$, the
parallel chase also generates an output sub-probabilistic database.

\begin{theorem}\label{thm:parspdb}\leavevmode
	\begin{enumerate}
		\item For all $D_{ \IN } \in \DD_{ \IN }$, the program $\G$ on input
			$D_{\IN}$ defines a sub-probabilistic database $\G_{ \App }( D_{\IN} 
			)$.
		\item For all sub-probabilistic databases $\D = (\DD_{\IN}, \DDD_{\IN},
			P)$, the program $\G$ on input $\D$ defines a sub-probabilistic 
			database $\G_{ \App }( \D )$.\qedhere
	\end{enumerate}
\end{theorem}

\section{Semantic Properties of Generative Datalog}\label{sec:semantics}

\subsection{Chase Independence}\label{ssec:chaseindependence}

Let $\G$ be a $\GDL$ program. We know from \cref{thm:seqspdb,thm:parspdb} that 
for every input sub-probabilistic database $\D$, $\G$ defines outputs 
$\G_{\App}( \D )$ for parallel steps and $\G_{\app}( \D )$ for sequential steps
where $\app$ is a measurable chase policy. In this section we show, that the
output is independent of the chase procedure and, in particular, independent of
the choice of policy in the sequential chase.

\begin{theorem}\label{thm:samedist}
	For all input instances $D_{ \IN } \in \DD_{\IN}$ and all measurable chase
	policies $\app$ we have $\G_{\app}( D_{\IN} ) = \G_{\App}( D_{\IN} )$.
\end{theorem}

\begin{remark}
	In the light of the above result, it is natural to ask why we introduced the
	sequential version (\cref{sec:seqchase}) in the first place, when the
	semantics could just have been introduced using the parallel chase. Apart
	from the sequential version being the more typical approach to the
	semantics, having both the sequential, and the parallel version at hand,
	together with the fact that they yield the same result, renders the
	semantics of $\GDL$ programs quite robust. The original motivation for
	introducing the parallel chase was, in fact, just to establish that the
	sequential chase does not depend on the choice of chase policy. The paper
	\cite{Barany+2017} features a similar statement but does not introduce a
	parallel version of the chase.  Thus, the sequential chase also serves the
	purpose of connecting our work to \cite{Barany+2017}. 
\end{remark}

We outline the proof of \cref{thm:samedist} and ask the reader to already have
a peek at \cref{fig:chasecorrespondence}. We will fix an event $\DDDD$ in the
output and show that it has the same measure in both processes. For this, it
suffices to come up with an easier to handle countable partition of $\DDDD$
into measurable sets and show that the measures coincide on each set of the
partition. This partition roughly brings paths in the existential and the
parallel chase trees into a one-to-one correspondence with respect to the
effect of rule applications.

\begin{proof}[Proof of {\cref{thm:samedist}}]
  Let $\mu^{\otimes}_{\app,D_{\IN}}$ denote the probability measure on
  the path space $( \DD^{\omega}, \DDD^{ \otimes \omega } )$ of the
  Markov process associated with the sequential chase using
  $\app$. Likewise, we let $\smash{\mu^{\otimes}_{\App,D_{\IN}}}$
  denote the probability measure on
  $( \DD^{\omega}, \DDD^{\otimes\omega} )$ from the parallel
  chase. That is,
  $\smash{\mu^{\otimes}_{\app,D_{\IN}} \after
    \liminst_{\app,D_{\IN}}^{-1}}$ is the sub-probability measure of
  $\G_{\app}( D_{\IN} )$, and
  $\smash{\mu^{\otimes}_{\App,D_{\IN}} \after
    \liminst_{\App,D_{\IN}}^{-1}}$ is the sub-probability measure of
  $\G_{\App}( D_{\IN} )$. Let $\phi_1, \dots, \phi_k$ be the rules of
  $\G^{\exists}$.
  
  First note that if $D$ is the label of a leaf in
  $T_{ \app, D_{\IN} }$, then it is also the label of a leaf in
  $T_{ \App, D_{\IN} }$, and vice versa.\footnote{This is not complicated to
  verify, but tedious. A proof of this can be found in \cite[Lemma 6.4.14]{Lindner2021}.}
  We call the (unique) path
  from $D_{\IN}$ to $D$ (in either tree) the $(D_{\IN},D)$-path. If
  $D$ is a leaf label, then the set of facts that are produced on the
  $(D_{\IN},D)$-path in $T_{ \app, D_{\IN} }$ coincides with the set
  of facts that are produced on $(D_{\IN},D)$-path in
  $T_{ \App, D_{\IN} }$.  We say a rule fires on a path if it appears
  in one of the chase steps along the path. If $D$ is a leaf label,
  then for every existential rule $\phi$ of $\G^{\exists}$, the number
  of times it fires on the $(D_{\IN},D)$-path in $T_{\app,D_{\IN}}$ is
  equal to the number of times it fires on the $(D_{\IN},D)$-path in
  $T_{\App,D_{\IN}}$. 

	Consider two paths, one in $T_{ \app, D_{\IN} }$ and one in $T_{ \App,
	D_{\IN} }$, that end in a leaf with the same label. For these two paths,
	we say that the application of $\phi$ in the sequential chase produces the
	same fact as the application of $\phi'$ in the parallel chase, if there 
	exist intermediate instances $D$ (in the sequential chase) and $D'$ (in the
	parallel chase) such that all of the following holds.
	\begin{itemize} 
		\item For all $b$ it holds that $f_{\phi}( \tup a, b ) = f_{\phi'}( \tup
			a', b ) \eqqcolon f_b$ where $( \phi,\tup a ) = \app(D)$ and $( \phi',
			\tup a') \in \App(D)$.
		\item The instance $D \cup \set{ f_b }$ is the successor instance of $D$
			on the path in $T_{\app, D_{\IN}}$.
		\item The instance $D' \cup \set{ f_b }$ is contained in the successor
			instance of $D'$ on the path in $T_{ \App, D_{\IN} }$.
	\end{itemize}

  We say that an instance $D$, the rule application of
  $\phi$ in the sequential chase produces the same fact as a rule application
  of $\phi'$ in the parallel chase, if $f_\phi( \tup a, b ) = f_{\phi'}( \tup
  a', b )$
  
  \bigskip
  
  We now fix an arbitrary measurable set $\DDDD \in \DDD$ and show that
  \[
    \mu^{\otimes}_{\app,D_{\IN}} \big( \liminst_{\app,D_{\IN}}^{-1}( \DDDD
    )\big) = \mu^{\otimes}_{\App,D_{\IN}} \big(
    \liminst_{\App,D_{\IN}}^{-1}( \DDDD )\big)\text.
  \] 
  Let $m,n \in \NN$. For $i = 1,\dots,m$ let $C^{(i)}$ be a $\set{0,1}$-valued
  $n\times k$-matrix
  \[
    C^{(i)} =
    \begin{pmatrix} 
      c^{(i)}_{11} & \dots & c^{(i)}_{1k}\\
      \vdots	& \ddots & \vdots\\
      c^{(i)}_{n1} & \dots & c^{(i)}_{nk}
	\end{pmatrix}\text.
  \]
  We call such a matrix a \emph{correspondence matrix}. We use such matrices to
  fix correspondences in paths of our both chase trees (cf.
  \cref{fig:chasecorrespondence}). The gist of our approach is to partition
  the discussed event $\DDDD$ according to correspondence matrices for the 
  paths in our chase trees.

  Let $\phi^{(1)}, \dots, \phi^{(m)}$ be a sequence of $m$ rules and let
  $\vec\ell_1, \dots, \vec\ell_n$ be a sequence of $n$ firing configurations.
  Then we define a set
  \[
    \DDDD_{ \vec\phi, \vec L, \vec C }= \DDDD_{ \phi^{(1)}, \dots, \phi^{(m)}, \vec\ell_1, \dots, \vec\ell_n,
      C^{(1)}, \dots, C^{(m)}}
  \]
  by letting $ D \in \DDDD_{ \vec\phi, \vec L, \vec C } $ if $D \in \DDDD$ and
  \begin{enumerate}
  \item $D$ is a leaf label on level $m$ in $T_{\app,D_{\IN}}$, such that
    the sequence of rules firing on the $(D_{\IN},D)$-path in 
    $T_{\app,D_{\IN}}$ is exactly $\phi^{(1)},\dots,\phi^{(m)}$; and
  \item $D$ is a leaf label on level $n$ in $T_{\app,D_{\IN}}$, such that
    the sequence of firing configurations on the $(D_{\IN},D)$-path in
	 $T_{\App,D_{\IN}}$ is exactly $\vec\ell_1, \dots, \vec\ell_n$ (where
	 $\vec\ell_j = ( \ell_{j1}, \dots, \ell_{jr} )$); and
 \item it holds that $c_{jr}^{(i)} = 1$ if and only if the application of
	 $\phi^{(i)}$ in the $i$th step in the sequential chase tree
	 $T_{\app,D_{\IN}}$ produces the same fact as one of the $\ell_{jr}$
	 applications of rule $\phi_r$ in the $j$th step in the parallel chase tree
	 $T_{\App,D_{\IN}}$.
  \end{enumerate}
  In essence, the matrix $C^{(i)}$ tells us how the application of $\phi^{(i)}$
  in the sequential chase corresponds to rule applications in the parallel
  chase, subject the firing sequence we fixed.
  For an illustration of the setup of $\DDDD_{\vec \phi, \vec L, \vec C}$, see
  \cref{fig:chasecorrespondence}.
  
  Note that if $\phi^{(i)}$ is an existential rule, then $C^{(i)}$ contains
  exactly one $1$-entry (this is because all existential rules come with
  different head relation symbols and since the same fact cannot be produced by
  the same existential rule under different instantiations).

  If $\phi^{(i)}$ is deterministic, then all but one row of 
  $C^{(i)}$ are all zeroes (if this were not the case, then the fact produced
  by $\phi^{(i)}$ in the sequential chase would get produced in two different
  rounds of the parallel chase, which is impossible by the definition of rule
  applicability). It can happen though that
  $c_{jr}^{(i)} = c_{jr'}^{(i)} = 1$ for some $r \neq r'$. This happens when
  multiple (deterministic) rules produce the same fact as the application of
  $\phi^{(i)}$ in the $i$th step of $T_{\app,D_{\IN}}$.

	\begin{figure}[t]\centering%
		\begin{tikzpicture}[sibling distance=.3cm,level distance=1cm]
			\node[hidden, label={[font=\small,name=Din]above:$D_{\IN}$}] 
				(root) at (0,0){};
			\node[below left=1cm and .8cm of root.center,anchor=center] (al) {};
			\node[below right=1cm and .8cm of root.center,anchor=center] (ar) {};
			\node[below left=3.5cm and 1.5cm of al.center,anchor=center] (bl) {};
			\node[below right=3.5cm and 1.5cm of ar.center,anchor=center] (br) {};

			\draw[thick] (root) to (al.center);
			\draw[thick] (root) to (ar.center);
			\draw[thick,dashed] (al.center) to (bl.center);
			\draw[thick,dashed] (ar.center) to (br.center);

			\node[hidden, label={[font=\small,name=Din']above right:$D_{\IN}$},
				right=6.5cm of root] (root') {};
			\node[below left=1cm and .8cm of root'.center,anchor=center] (al') {};
			\node[below right=1cm and .8cm of root'.center,anchor=center] (ar') {};
			\node[below left=3.5cm and 1.5cm of al'.center,anchor=center] (bl') {};
			\node[below right=3.5cm and 1.5cm of ar'.center,anchor=center] (br') {};

			\draw[thick,double] (root') to (al'.center);
			\draw[thick,double] (root') to (ar'.center);
			\draw[thick,dashed,double] (al'.center) to (bl'.center);
			\draw[thick,dashed,double] (ar'.center) to (br'.center);

			\node[hidden,below left=1cm and .4cm of root.center,anchor=center] (d1) {};
			\node[hidden,below right=1cm and .5cm of d1.center,anchor=center] (d2) {};
			\node[hidden,below left=1cm and .2cm of d2.center,anchor=center] (d3) {};
			\node[hidden,below right=1cm and .6cm of d3.center,anchor=center] (d4) {};

			\node[hidden,below left=.333334cm and .133334cm of root'.center,anchor=center] (d1') {};
			\node[hidden,below left=.333333cm and .133333cm of d1'.center,anchor=center] (d2') {};
			\node[hidden,below left=.333333cm and .133333cm of d2'.center,anchor=center] (d3') {};
			\node[hidden,below right=.5cm and .2cm of d3'.center,anchor=center] (d4') {};
			\node[hidden,below right=.5cm and .2cm of d4'.center,anchor=center] (d5') {};

			\coordinate (e1) at (d1);
			\coordinate (e2) at (d2);
			\coordinate (e3) at (d3);
			\coordinate (e4) at (d4);
			\coordinate (e1') at (d1');
			\coordinate (e2') at (d2');
			\coordinate (e3') at (d3');
			\coordinate (e4') at (d4');
			\coordinate (e5') at (d5');
			
			\begin{scope}[on background layer]
				\fill[b20,opacity=.5] (root.center) -- (e1) -- (e3') -- (e2') -- cycle;
				\fill[b50,opacity=.5] (e1) -- (e2) -- (e4') -- (e3') -- cycle;
				\fill[b20,opacity=.5] (e2) -- (e3) -- (e1') -- (root'.center) -- cycle;
				\fill[b50,opacity=.5] (e3) -- (e4) -- (e2') -- (e1') -- cycle;
			\end{scope}

			\draw[->,thick] (root.center) to (d1);
			\draw[->,thick] (d1.center) to (d2);
			\draw[->,thick] (d2.center) to (d3);
			\draw[->,thick] (d3.center) to (d4);

			\node[inner] at (root) {};
			\node[inner] at (d1.center) {};
			\node[inner] at (d2.center) {};
			\node[inner] at (d3.center) {};
			\node[leaf,label={[font=\small,name=Em]below right:$E_{m}$}] at (d4.center) {};

			\draw[-{Stealth[round,width=8pt,length=6pt]},double,thick] (root'.center) to (d3');
			\draw[-{Stealth[round,width=8pt,length=6pt]},double,thick] (d3'.center) to (d5'); 

			\node[inner] at (root') {};

			\node[inner] at (d3'.center) {};
			\node[leaf,label={[font=\small,name=Dn]below right:$D_{n}$}] at (d5'.center) {};

			\coordinate[below right=-.05cm and 1.4cm of root] (C1);
			\coordinate[below left=.2cm and 1.4cm of root'] (C2);

			\draw[<->,dotted,b75] 
				(C1) to[bend right=3.5]
				node[above,sloped,font=\small,text=b35]
				{correspondence via $\vec C$} (C2);
			\draw[thick,dotted,b50] (Em) to[bend right=10] node[circle,fill=white]{$=$}(Dn);
			\draw[thick,dotted,b50] (Din) to[bend left=15] node[circle,fill=white]{$=$}(Din');

		\coordinate[left=.2cm of al] (phi1) {};
		\coordinate[left=.2cm of bl] (phi2) {};
		\draw[shorten <=-.8cm,shorten >=.4cm,->,dotted,b75] 
			(phi1) to node[midway,above,sloped,pos=.35,font=\small,text=b35]
			{rule sequence $\vec\phi$} (phi2);

		\coordinate[right=.2cm of ar'] (L1) {};
		\coordinate[right=.2cm of br'] (L2) {};
		\draw[shorten <=-.8cm,shorten >=.2cm,->,dotted,b75] 
			(L1) to node[midway,above,sloped,pos=.35,font=\small,text=b35]
			{firing conf. sequence $\vec L$} (L2);

		\end{tikzpicture}
		\caption{Paths obtained from $\DDDD_{\vec\phi,\vec
		L,\vec C}$ in the sequential (left) and the parallel (right) chase tree.}
		\label{fig:chasecorrespondence}
	\end{figure}

	Every instance $D \in \DDDD$ lies in exactly one set 
	$\DDDD_{ \vec\phi, \vec L, \vec C }$. That is, the sets $\DDDD_{ \vec\phi, 
	\vec L, \vec C }$ partition $\DDDD$ into countably many sets. (Moreover, 
	$\DDDD_{ \vec\phi, \vec L, \vec C }$ are measurable. They can be expressed
	using $\app$, sequences of measurable selections $\app_{ij_i}$ (which are
	introduced below), the measurable functions $f_{\phi}$ and the diagonals in
	the fact spaces.)

	\bigskip

	Now let $\vec{ \DDDD } = \vec{ \DDDD }( \vec\phi, \vec L, \vec C ) 
	\coloneqq \liminst_{ \App, D_{\IN} }^{ -1 }\big(
	\DDDD_{ \vec\phi, \vec L, \vec C } \big)$. Then, by definition\footnote{For
	details see \cite[p. 21~\&~Proposition 8.2]{Kallenberg2002}.},
	\begin{equation}
			\mu^{\otimes}_{\App,D_{\IN}}\big( \vec{ \DDDD } \big)
			=
			\begin{aligned}[t]
			\iint\dotsi\int & \mathbbm{1}_{ \vec{ \DDDD } }( D_{\IN}, D_1, \dots, 
				D_{n-1}, D_n, D_n, \dots )\\
								 &	\qquad 
								\StepP( D_{n-1}, \d{D_n} ) \dotsm
								\StepP( D_1, \d{D_2} )
								\StepP( D_{\IN}, \d{D_1} )
		\end{aligned}
		\label{eq:parint}
	\end{equation}
	where $\StepP( D_i, \d{ D_{i+1} } )$ is shorthand for $\d{\big(
	\StepP(D_i,\placeholder) \big)}$.

	\bigskip
	Recall that for any firing configuration $\vec\ell' = (\ell_1', \dots, 
	\ell_k')$ for parallel chasing, we can obtain a sequence of $\sum_i
	\ell'_i$ measurable selections $\app_1, \dots, \app_{ \sum_i \ell'_i
	}$ such that for all $D$ with firing configuration $\vec\ell'$ it holds that
	\[
		\App( D ) = 
		\set[\big]{ \app_j( D ) \with 1 \leq j \leq \textstyle\sum_i \ell'_i }
		\text.
	\]
	(We have already used this in the proof of 
	\cref{lem:StepE}\labelcref{itm:StepE1}, details can be found ibid.)
	For all the firing configurations $\vec\ell_1, \dots, \vec\ell_n$ we fixed
	before, we chose such a sequence of measurable selections $\app_{ij_i}$ such
	that $\app_{ij}$ is the $j_i$th measurable selection belonging to 
	$\vec\ell_i = (\ell_{i1}, \dots, \ell_{ik})$. In the following we let 
	\[
		( \phi_{ij_i}, \tup a_{ij_i} ) \coloneqq \app_{ij_i}( D_i )
	\]
	where $1 \leq j_i \leq \ell_{i1} + \dots + \ell_{ik} \eqqcolon s_i$ for all
	$i = 1, \dots, n$. Moreover, we let $\psi_{ij_i}$ denote the parameterized
	distribution from rule $\phi_{ij_i}$ and $\mu_{ij_i}$ the associated
	underlying measure (cf. \cref{ssec:paramdist}). Note that since we fixed
	$\vec\ell_1, \dots, \vec\ell_n$, the collections $(\phi_{ij_i})_{j_i}$,
	$(\psi_{ij_i})_{j_i}$ and $(\mu_{ij_i})_{j_i}$ do not depend on $D_i$.

	Now for every measurable function $g$ it holds that
	\begin{equation}
		\begin{multlined}
			\int_{ \DD } g( D_{i+1} ) \StepP( D_i, \d{ D_{i+1} } )
			=
			\int_{ \WW } 
				g\big( 
					\Ext_{\vec\ell_i}( 
						D_i, 
						\phi_{i1},\tup a_{i1},b_{i1},
						\dots, 
						\phi_{is_i}, \tup a_{is_i}, b_{is_i}
					)
				\big)\\
				\cdot 
				\psi_{i1}\params{\tup a_{i1}}( b_{i1} ) 
				\cdot\dotsc\cdot
				\psi_{is_i}\params{ \tup a_{is_i} } ( b_{is_i} )
				\:\d{\big( 
					\mu_{i1} \otimes \dots \otimes \mu_{is_i}
				\big)}
		\end{multlined}
		\label{eq:spacetransfi}
   \end{equation}
	where $\WW = \prod_{j = 1}^{ s_i } \WW_{\phi_{ij}}$ and $i \in \set{ 0,
	\dots, n-1 }$ with $D_0 = D_{ \IN }$. The equality in 
	\labelcref{eq:spacetransfi} is obtained as follows. We apply the 
	substitution rule (see \cref{fac:substitution}) for the measurable function
	$\Ext_{\vec\ell_i}$ to transform the domain of integration. With
	applications of the chain rule (see \cref{fac:chainrule}, we extract
	$\psi_{ij_i}$ from the integration measure. 

	Now by Fubini's theorem we can rewrite the right hand side of
	\labelcref{eq:spacetransfi} as an iterated integration with $s_i$ integrals.
	We do so from the inside to the outside for all of the integrals appearing
	in \labelcref{eq:parint}. We obtain that $\mu_{\App,D_{\IN}}^{\otimes}( 
	\vec{\DDDD} )$ is equal to
	\begin{equation}\label{eq:parDphiLC}
		 \begin{multlined}
			 \Gamma_{\vec{\DDDD}} =
			 \int
			\mathbbm{1}_{ \DDDD_{ \vec\phi, \vec L, \vec C } } 
				\Big( 
					\Ext_{\vec\ell_n}\big( 
						\Ext_{\vec\ell_{n-1}}\big( 
							\dotsm 
							\Ext_{\vec\ell_1}\big(
								D_{\IN}, 
								(\phi_{1j_1},\tup a_{1j_1},b_{1j_1})_{j_1}
							\big)
						\dotsm\big),
						(\phi_{nj_n}, \tup a_{nj_n}, b_{nj_n})_{j_n}
					\big)
				\Big)\\
				\cdot\prod_{ i = 1 }^{ n } 
				\prod_{ j_i = 1 }^{ s_i } 
				\psi_{ij_i}\params{ \tup a_{ij_i} }( b_{ij_i} )
				\:\d{\big( \mu_{11} \otimes\dots\otimes \mu_{ns_n}\big)}\text.
		\end{multlined}
	\end{equation}
	With $\vec{\EEEE} = \vec{\EEEE}( \vec\phi, \vec L, \vec C ) \coloneqq 
	\liminst_{\app,D_{\IN}}^{-1}( \DDDD_{\vec\phi, \vec L, \vec C} )$, applying
	the same procedure to the corresponding expression for
	$\mu^{\otimes}_{\app,D_{\IN}}( \vec{\EEEE} )$ yields
	\begin{equation}\label{eq:seqDphiLC}
		\begin{multlined}
			\gamma_{\vec{\EEEE}} = 
			\int \mathbbm{1}_{ \DDDD_{ \vec\phi,\vec L,\vec C } }
		\Big( 
			\ext( \ext( \dotsm \ext( 
				D_{\IN}, \phi^{(1)}, \tup a^{(1)}, b^{(1)} 
			), \dots),
			\phi^{(m)},\tup a^{(m)}, b^{(m)} )
		\Big)\\
		\cdot\prod_{i=1}^{m} \psi^{(i)}\params{\tup a^{(i)}}
		\d{\big(\mu^{(1)}\otimes\dots\otimes\mu^{(m)}\big)}\text.
		\end{multlined}
	\end{equation}
	We have already argued that the rules, parameterized distributions, and
	underlying measures are fixed in these expressions. Yet, we note that the
	parametrizations $\tup a^{(i)}$ and $\tup a_{ij_i}$ in
	\labelcref{eq:seqDphiLC} and \labelcref{eq:parDphiLC} depend on the outcome
	of the previous sample via $\app$ and the measurable selections
	$\app_{ij_i}$ we constructed before. Note that the integrands of
	\labelcref{eq:parDphiLC,eq:seqDphiLC} are functions in $b_{11}, \dots,
	b_{ns_n}$, respectively $b^{(1)}, \dots, b^{(m)}$. Every (suitable) tuple
	$(b_{1s_1}, \dots, b_{ns_n})$ gives rise to a random path
	\[
		D_{\IN} \StepE D_1 \StepE D_2 \StepE \dotsc \StepE D_n
	\]
	in $T_{ \App, D_{\IN} }$, that is almost surely contained in $\vec{\DDDD} = 
	\liminst_{ \App, D_{\IN} }^{-1}( \DDDD_{\vec\phi,\vec L, \vec C})$. 
	Likewise, every (suitable) tuple $b^{(1)}, \dots, b^{(m)}$ in 
	\labelcref{eq:seqDphiLC} describes a random path
	\[
		D_{\IN} \stepE E_1 \stepE E_2 \stepE \dotsc \stepE E_m
	\]
	in $T_{ \app, D_{\IN} }$ that is a.\,s. contained in $\vec{\EEEE} = 
	\liminst_{ \app, D_{\IN}}^{-1}(\DDDD_{\vec\phi,\vec L,\vec C})$. By the
	definition of $\DDDD_{ \vec\phi, \vec L, \vec C }$, we have $D_n = E_m \in
	\DDDD$ where no rule is applicable anymore. Recall that in $\DDDD_{\vec\phi,
	\vec L, \vec C}$, the application of rule $\phi^{(i)}$ in $T_{ \app, D_{\IN}
	}$ has the same resulting fact effect as $\smash{c^{(i)}_{jr}}$ of the
	applications of $\phi_r$ in the $j$th step in $T_{ \App, D_{\IN} }$. The
	confinement to $\smash{\DDDD_{\vec\phi,\vec L,\vec C}}$ establishes a
	one-to-one correspondence between $\vec{\DDDD}$ and $\vec{\EEEE}$, resp.\
	between tuples $(b_{11},\dots,b_{ns_n})$ and $(b^{(1)},\dots,b^{(m)})$ 
	(where the first of these sequences may contain redundant $*$'s as leftovers
	from collapsing deterministic rules). In this case, we call the tuples
	equivalent. It remains to show that if $(b_{11}, \dots, b_{ns_n})$ and
	$(b^{(1)}, \dots, b^{(m)})$ are equivalent, then
	\begin{equation}\label{eq:psis}
		\prod_{ i = 1 }^{ n } \prod_{ j_i = 1 }^{ s_i }
			\psi_{ij_i}\params{\tup a_{ij_i}}(b_{ij_i})
		= \prod_{ i = 1 }^{ m } \psi^{(i)}\params{\tup a^{(i)}}(b^{(i)})\text.
	\end{equation}
	This is the case because of the correspondence fixed by $\vec C$, which can
	be seen as follows. First of all, all factors that belong to non-existential
	rules can be canceled from \labelcref{eq:psis} (even though there might be a
	different number of them on both sides) as they evaluate to $1$ (cf. 
	\cref{rem:determpdist}). The remaining numbers of factors on both sides
	coincide (see our discussion of $\smash{c^{(i)}_{jr}}$ before) and, again,
	are in one-to-one correspondence with each other via $\vec C$. In
	particular, if $\tup a_{ij_i}, b_{ij_i}$ and $\tup a^{(i)},b^{(i)}$ produce
	the same fact, then $\tup a_{ij_i} = \tup a^{(i)}$ and $b_{ij_i} = b^{(i)}$. 

	A similar argument applies to the product measures (of the $\mu^{(i)}$,
	resp.\ the $\mu_{ij_i}$) used for the integration in
   \labelcref{eq:parDphiLC,eq:seqDphiLC}---recalling that the
   measures themselves are already fixed by fixing $\vec\phi$ and
   $\vec L$ and that the measure spaces from the non-existential
	rules are trivial (\cref{rem:determpdist,rem:determmeas}).  
	\bigskip

	We have shown that the expressions \labelcref{eq:parDphiLC,eq:seqDphiLC}
	have the same value (for fixed $\DDDD$), i.\,e.\ for all suitable $\vec\phi$,
	$\vec L$, and $\vec C$ it holds that
	\begin{equation}\label{eq:sameval}
		\mu^{\otimes}_{\app,D_{\IN}}
		\big( \vec{\EEEE}(\vec\phi,\vec L,\vec C) \big)
		= \gamma_{\vec{\EEEE}(\vec\phi,\vec L,\vec C)}
		= \Gamma_{\vec{\DDDD}(\vec\phi,\vec L,\vec C)}
		= \mu^{\otimes}_{\App,D_{\IN}}
		\big( \vec{\DDDD}(\vec\phi,\vec L,\vec C) \big)\text.
	\end{equation}
	As $\DDDD_{\vec\phi, \vec L, \vec C}$ is a countable, measurable partition
	of $\DDDD$, we obtain
	\[
		\mu^{\otimes}_{\app,D_{\IN}}\big( \liminst_{\app,D_{\IN}}^{-1}( \DDDD ) 
		\big)
		= \sum_{ \vec\phi,\vec L,\vec C } 
		\gamma_{ \vec{\EEEE}( \vec\phi,\vec L, \vec C ) }
		\mathrel{\overset{\mathclap{\labelcref{eq:sameval}}}{=}} 
		\sum_{ \vec\phi,\vec L,\vec C } 
		\Gamma_{ \vec{\DDDD}( \vec\phi,\vec L, \vec C ) }
		= \mu^{\otimes}_{\App,D_{\IN}}\big( \liminst_{\App,D_{\IN}}^{-1}( \DDDD ) 
		\big)\text.\qedhere
	\]
\end{proof}

When fed an initial distribution instead of a single instance $D_{\IN}$, we
directly obtain the following corollary.

\begin{corollary}
	For all sub-probabilistic databases $\D = ( \DD_{\IN}, \DDD_{\IN}, P )$, and
	all measurable chase policies $\app$ it holds that $\G_{\app}( \D ) = \G_{
	\App }( \D )$.
\end{corollary}

Because of the results in this subsection, we just write $\G(\D)$ for the
sub-probabilistic database that is obtained by running $\G$ on $\D$. Recall 
that we can also eliminate the auxiliary tuples if we want to (see 
\cref{rem:endproj}).

\subsection{Comparison to the Original Semantics}\label{ssec:origsim}

We continue from the discussions surrounding
\cref{exa:behavior1,exa:behavior2}. The key difference between our semantics
and Bárány et al.'s semantics is in the mechanism that prevents generating
infinitely many facts that only differ in a sampled value. We tie this to the
rules: each probabilistic rule is only permitted to fire once for each setting
of the parameters. Bárány et al.\ tie it to the distributions: each distribution is only
permitted to produce one sample for each parameter setting. This difference in
the semantics (in particular, the decoupling of sampling from distribution
names) resolves the issue sketched in \cref{exa:behavior2}.

Both we and Bárány et al.\ have mechanisms to relax the requirement of only sampling
once: we allow it to repeat rules; Bárány et al.\ introduce symbolic
parameters that can be added to a distribution to allow several
applications. Multiple copies of the same rule have a different behavior under
our semantics as they would exhibit in the original one. We treat multiple
copies as separate instructions, while they collapse to the effect of just a
single such rule in the original semantics. We achieve this by associating
different existential rules to the individual copies during the translation to
the existential program, leading to the semantic difference pointed out in
\cref{exa:behavior1}.

\medskip

With our version of the semantics, we can, however, simulate the semantics of
Bárány et al., \emph{as far as only the distribution over finite outcomes is
concerned}. This means that for a program $\G$ using the original semantics, we
can construct a new program $\H$ that, under our semantics, has the property 
that the sup-probabilistic database it produces on (finite) outcome instances
coincides with the distribution over the finite outcomes of $\G$. As database
instances are finite per definition, this is, after all, the interesting part
of the distribution. The semantics of Bárány et al.\ extends towards a
distribution also over the infinite chase paths, that is, over infinite
outcomes. Recall that in our semantics, all infinite chase paths are merged
into a single error event. Therefore, we necessarily lose the information about
the infinite program executions that are present in the semantics of
\cite{Barany+2017}.

The following example exposes the difference in the semantics and illustrates
the simulation mentioned above.

\begin{figure}[H]
	\mbox{}\hfill%
	\begin{subfigure}[b]{.4\textwidth}\centering%
	\begin{tabular}{ l l }
		\toprule
		$\G \colon$
		& $S\big( \Flip\params[\big]{ \frac12 } \big) \whenever R( 0 )$ \\[1ex]
		& $T\big( \Flip\params[\big]{ \frac12 } \big) \whenever R( 0 )$ \\
		\bottomrule
	\end{tabular}
	\subcaption{The program $\G$.}\label{fig:flipflip}
	\end{subfigure}
	\hfill%
	\begin{subfigure}[b]{.4\textwidth}\centering%
	\begin{tabular}{ l l }
		\toprule
		$\H \colon$
		& $A\big( \Flip\params[\big]{ \frac12 } \big) \whenever R( 0 )$\\[1ex]
		& $S( x ) \whenever A( x )$ \\[1ex]
		& $T( x ) \whenever A( x )$ \\
		\bottomrule
	\end{tabular}
	\caption{The program $\H$.}\label{fig:flipflipsim}
	\end{subfigure}
	\hfill\mbox{}%
	\caption{Exemplary simulation of the original semantics.}
\end{figure}

Under Bárány et al.'s semantics, $\G$ has outcomes
$\set{ R( 0 ), S( 0 ), T( 0 ) }$ and $\set{ R( 0 ), S( 1 ), T( 1 ) }$
with probability $1/2$ each, whereas our semantics yield the four
possible outcomes $\set{ R( 0 ), S( 0 ), T( 0 ) }$,
$\set{ R( 0 ), S( 0 ), T( 1 ) }$, $\set{ R( 0 ), S( 1 ), T( 0 ) }$ and
$\set{ R( 0 ), S( 1 ), T( 1 )} $, each with probability $1/4$. Yet, we
can easily simulate the original semantics of $\G$ by pulling out the
sampling to a separate rule, as in \cref{fig:flipflipsim}. This
program has outcomes $\set{ R( 0 ), A( 0 ), S( 0), T( 0 ) }$ and
$\set{ R( 0 ), A( 1 ), S( 1 ), T( 1 ) }$ with probability $1/2$
each. We can ignore the auxiliary predicate $A$ and restrict the
resulting probabilistic database to the schema $\set{ R, S, T }$
without changing the probabilities.

This simple argument can be generalized to arbitrary programs. We note that the
original semantics also featured the use of \enquote{event expressions} that
could be employed to reuse samples for the same parameter configuration. Such
can be simulated by decomposing the rule, and putting the event expression into
a separate relation. We leave the details to the reader.  Let us remark that it
is similarly easy to simulate our semantics with that of Bárány et al. All our
results would also hold starting from Bárány et al.'s semantics for discrete
distributions, so it is really just a matter of taste which version the reader
prefers.

\begin{remark}[First-Order Equivalence]\label{rem:fo-eq}
	While the different semantics can simulate each other, the choice of Bárány
	et al.'s semantics was in particular made to obtain a decidable sufficient
	criterion for the \enquote{semantic equivalence} of two programs.

	Following \cite{Barany+2017}, two programs $\G$ and $\G'$ are 
	\emph{first-order equivalent}, if the first-order theories defined by their
	collection of rules coincide.\footnote{To be precise, this notion is
	introduced for the full PPDL language (see \cref{sec:ppdl}), not only the
	generative part in \cite{Barany+2017}.} For this, they interpret
	parameterized distributions as function symbols. Bárány et al. show that
	their notion of first-order equivalence is decidable for a simple syntactic
	class of programs, and that first-order equivalence of two programs entails
	that they produce the same output distribution under their semantics
	\cite[Theorem 5.6 and 5.5]{Barany+2017} (this latter property is called
	\emph{semantic equivalence}). In general, semantic equivalence is 
	undecidable, even when there are no parameterized distributions at all
	\cite[Theorem 5.6]{Barany+2017}.

	While the above definition of first-order equivalence feels natural and is
	in line with typical notions of first-order equivalence, the decision to
	interpret every occurrence of a parameterized distribution as another
	occurrence of the same function symbol is susceptible to debate. In
	particular, this treats a parameterized distribution just like a function,
	yet sampling repeatedly from a distribution without changing its parameters
	may well yield different outcomes. If the distribution is 
	continuous, this even happens almost surely.\footnote{With this, it is also
	clear that the first-order equivalence of \cite{Barany+2017} is a weaker
	notion than semantic equivalence.}

	The result from \cite{Barany+2017}, that (their notion of) first-order
	equivalence implies semantic equivalence no longer holds with our semantics.
	One may easily verify that with our semantics, for example, that the program
	$\G_0$ from \cref{exa:behavior1} is first-order equivalent to the program
	obtained by removing one of its rules. Yet, under our semantics, they
	produce different outputs.

	We propose to adapt the notion of first-order equivalence to be tailored to 
	our semantics. Recall that one major change we made was that we allowed 
	rules to be present multiple times. When translating $\GDL$ to $\EDL$ 
	programs though, we introduced a distinction between the copies of a
	probabilistic rule, by providing (per rule) a unique new relation for 
	storing the sample outcomes (cf. \labelcref{eq:new_rules} in
	\cref{sec:assocedl}). Thus, we propose the following change in the notion
	of first-order independence: Call two $\GDL$ programs $\G$ and $\G'$ 
	\emph{first-order} equivalent, if the first-order theories defined by their
	collection of rules coincide, \emph{where every occurrence of a
	parameterized distribution $\psi\params{\vec p}$ is treated as an occurrence
	of a new distinguished function symbol.} With this changed notion, copies 
	of probabilistic rules cannot be (first-order equivalently) rewritten into
	single rules anymore.

	With respect to this new notion of first-order equivalence it seems possible
	to transfer the results of \cite[Section 5.1]{Barany+2017} to our semantics.
\end{remark}

\begin{remark}[Fairness]
In \cite{Barany+2017}, the authors restrict their scope to
  \enquote{fair} chase trees. This means that if for an intermediate
  instance $D$ it holds that $D \models \phi_b( \tup u )$ for some
  valuation $\tup u$, then eventually $D' \models \phi_h( \tup u )$
  for some later intermediate instance, on all paths. Note that by
  construction, any finite path of the chase tree that is ending in a
  leaf satisfies this fairness condition. Contrary to
  \cite{Barany+2017}, we do not impose the fairness condition on
  the infinite parts of the tree, as all infinite paths are collectively
  mapped to our error event $\Error$ anyway. The important point is that the
  (sub-)probability space obtained from the finite paths is
  independent of any fairness issue.
\end{remark}

\subsection{Termination Behavior}\label{ssec:termination}

In \cref{sec:seqchase,sec:parchase}, we have constructed Markov processes for 
given $\GDL$ programs. Every point in the $i$th component of the path space
$(\DD^{\omega},\DDD^{\otimes\omega})$ can be seen as corresponding to a program
configuration after $i$ steps and every path in $(\DD^{\omega},
\DDD^{\otimes\omega})$ as a program run. A run is called \emph{terminating} if
it corresponds to a finite path in the respective chase tree. The program $\cal
G$ is called \emph{terminating} if all its runs terminate. The following result
from the original paper trivially extends to our setting, with the notion of
\emph{weak acyclicity} remaining unchanged (see \cite{Barany+2017}).

\begin{theorem}[{cf. \cite[Theorem 3.10]{Barany+2017}}]
	Let $\mathcal G$ be a \GDL{} program. If $\mathcal G$ is weakly acyclic, 
	then $\mathcal G$ is terminating.
\end{theorem}

That is, whenever there are no circular dependencies involving probabilistic
rules, then all paths in any chase tree are finite. In general, the $\GDL$
program $\G$ terminates on input $D_{\IN}$ if and only if its existential
version is terminating on $D_{\IN}$. Likewise, $\G$ terminates on every input
probabilistic database, if its existential version terminates on all instances.
There exists a lot of research surrounding the termination of existential
Datalog programs, in particular regarding classes of programs (going way beyond
the notion of weakly acyclicity) where termination is guaranteed, or at least
decidable (see, for example, \cite{CuencaGrau+2013,Gogacz+2020}). 

For probabilistic programs, termination is a more subtle notion. A program is said to be \emph{almost surely terminating} if it terminates with probability one. If the program terminates in finitely many steps in expectation, it is called \emph{positively almost-sure terminating} \cite{BournezGarnier2005}. The hardness of these notions is studied in \cite{KaminskiKatoen2015}. Active research in probabilistic programming is concerned with identifying sufficient criteria for determining (positive) almost-sure termination \cite{Chatterjee+2020}.

A more thorough investigation of, in particular, the probabilistic termination
of generative Datalog programs is still open at this point.

\section{The Full Probabilistic Programming Datalog Language}\label{sec:ppdl}

The GDatalog language is extended by constraints to obtain the full
\emph{probabilistic-programming Datalog (PPDL)} language 
\cite[Section 5]{Barany+2017}. Formally, a \emph{PPDL program} is a pair
$( \G, \Phi )$ where $\Phi$ is some constraint specification, for example, a
first order sentence over the database schema of $\G$. The semantics of the
PPDL program is then given by conditioning the output $\G(\D) \coloneqq 
(\DD, \DDD, P)$ of $\G$ on the set of instances $\DD_{\Phi}$ satisfying the
constraint specification $\Phi$ (see \cite[Definitions 5.1 and
5.3]{Barany+2017}). That is, the probability of an event $\DDDD$ is given as
\[
	\frac
	{ P\big( \DDDD \cap \DD_{\Phi} \big) }
	{ P\big( \DD_{ \Phi } \big) }
	\text.
\]

There are multiple pitfalls to be aware of here. Conditioning the 
(sub-)probability space $\G(\D) = ( \DD, \DDD, P)$ in the suggested way 
requires not only that the set of instances $\DD_{\Phi}$ satisfying $\Phi$ is 
measurable, but also that $P( \DD_{\Phi} ) > 0$. If these two requirements are
fulfilled, then the PPDL program returns a well-defined (sub-)probabilistic 
database.

As Bárány et al.\ note, the measurability of $\DD_{\Phi}$ is already an issue in
their setting of discrete distributions unless the program $\G$ is weakly
acyclic \cite[p.~22:18]{Barany+2017}. The good news is that in our framework of
standard PDBs, the measurability of $\DD_{ \Phi }$ is constituted by
\cref{fac:measurableviews} as long as the constraint can be expressed by (for
example) a relational algebra query. 

Thus, we see the requirement of having $P( \DD_{\Phi} ) > 0$ as the more
delicate one. This discussion is bypassed in \cite{Barany+2017} by resorting to
the weakly acyclic case, where the probability space gets discrete in
their setup. For example, it might be reasonable to use constraints involving
equality when designing PPDL programs. Alas, for example equality in $\RR$ 
yields a set of Lebesgue measure $0$ (namely the diagonal) in $\RR^2$. From 
this, one can easily construct examples where $P( \DD_{ \Phi } ) = 0$ even
though $\Phi$ is a very natural and typical kind of constraint. Trying to 
condition on events of measure $0$ usually yields paradoxical results, as in
the Borel-Kolmogorov paradox \cite[p.~50~et~seq.]{Kolmogorov1956}. Still, this
does not rule out sensible definitions of a probabilistic database conditioned
on such a constraint. Work from the area of probabilistic programming however
suggests that resolving the paradox is no trivial task 
\cite{Borgstrom+2013,Jacobs2021}.

\begin{remark}[More Simulations]
	Bárány et al.\ \cite{Barany+2017} showed that their language version of PPDL
	can simulate Markov Logic Networks (MLNs) \cite{RichardsonDomingos2006} and
	Probabilistic Context-Free Grammars (PCFGs) \cite{BoothThompson1973} (see
	\cite{Icard2020}). Because we can simulate their language, we can simulate
	MLNs and PCFGs as well using the methods described in \cite{Barany+2017}.
	Since our language is more powerful, one might wonder whether these
	simulations can be reasonably extended. 

	There exist two extensions of MLNs that come to mind, one with infinitely 
	many variables (\emph{Infinite MLNs} \cite{SinglaDomingos2007}), and one 
	with uncountable variable domains (although then finitely 
	many variables; \emph{Hybrid MLNs} \cite{WangDomingos2008}). Unfortunately, 
	there is no easy correspondence between our language and these models.
	
	\begin{itemize}
		\item Hybrid MLNs describe distributions over a finite and fixed number
			of variables with uncountable domains, contrasting the unbounded size
			of our collections of facts. Moreover, the semantics of Hybrid MLNs
			only takes intervals of $\RR$ into consideration.

		\item Infinite MLNs describe a probability space over countably (possibly
			countably infinitely) many variables all of which have (to our
			understanding) at most countable domains. While the resulting
			probability space is indeed uncountable, this contrasts our setting
			with a probability distribution over finite subsets of uncountable
			spaces. In particular, our continuous distributions seem to add no
			power here and, should a correspondence between infinite MLNs and
			Probabilistic Programming Datalog exist, then it is already exhibited
			by the discrete semantics of \cite{Barany+2017}.
	\end{itemize}

	To our knowledge, there is no sensible generalization of PCFGs that involve
	uncountable spaces. Typically, the underlying domains in PCFGs are finite.
\end{remark}

\section{Conclusion}\label{sec:conclusion}

The original Probabilistic Programming Datalog language of Bárány, ten Cate,
Kimelfeld, Olteanu and Vagena \cite{Barany+2017} is limited to discrete
distributions. Given that continuous distributions appear in a variety of
application scenarios for probabilistic databases
(cf.~\cite{Cheng+2003,Deshpande+2004}), an extension with support for
continuous distributions was noted as an open problem in \cite{Barany+2017}.
In this paper, we developed such an extension.

Our key technical results are as follows:
\begin{enumerate}
	\item The Generative Datalog language of \cite{Barany+2017} can be
		faithfully extended towards the support of continuous distributions,
		adding to its expressive power. Into this extended language, one may also
		incorporate conditioning under events of positive probability just as in
		\cite{Barany+2017}.
	\item We consider the semantics using a sequential and a parallel chase
		procedure, and show that the results of programs coincide under both
		kinds of approaches. In particular, outcomes do not depend on chase 
		policies.
	\item We may also use (sub-)probabilistic databases as inputs, and we obtain
		a well-defined output sub-probabilistic database as long as the
		probability of a finite computation is greater than $0$.
\end{enumerate}

We summarize the technical developments in a high-level view of the (new) 
semantics. The generative part of a PPDL program can sample from probability
distributions in order to generate new attribute values, and it can do so
recursively. The semantics itself is described via chase trees, where different
branches correspond to different samples. This becomes technically challenging
with the introduction of continuous distributions, as nodes of the chase trees
may now have uncountably many children. Exploiting the advanced machinery of
probability and measure theory, we show that such uncountable chase trees are
encodings of a Markov process of database instances. Associated with each 
Markov process is a probability measure on its paths, that is, on the paths of
the chase tree. Each path in the tree that is starting in the root node 
corresponds to the process of building a database instance by adding facts one
by one. Paths that end in leaf nodes then correspond to well-defined database 
instances. In cutting off infinite paths, and projecting the rest back to the
represented database instances, we obtain a sub-probabilistic database
\enquote{generated} by the program, that may afterwards be conditioned on
satisfying a given set of constraints.

In addition to the added expressive power, we embed our semantics of PPDL into
the Standard PDB framework of \cite{GroheLindner2020}. This makes the language
compositional in the sense that it may use (sub-)probabilistic input databases,
and produces (sub-)probabilistic output databases within the framework. We show
the equivalence of sequential chase procedures with a notion of parallel chase,
rendering the semantics quite robust.

\medskip

\subsubsection*{Future Work}

The discrete version of PPDL can simulate every finite probabilistic database,
for example, by using the MLN simulation. Thus, PPDL can be seen as a
\emph{complete} representation system for finite PDBs. From the point of view
of probabilistic databases, the most interesting question is how powerful PPDL
is as a representation system for infinite PDBs. This is even unclear for the
purely discrete version.

Moreover, PPDL raises a lot of natural questions in the overlap of
probabilistic databases and probabilistic programming that remain unanswered.
For the constraint part of a PPDL program, can we determine from the syntax of
the program whether a particular constraint has measure $0$? Are there sensible
languages or fragments that avoid the issue of measure $0$ conditioning? In the
case of measure $0$ conditioning, can we apply techniques from the
probabilistic programming community in order to resolve the emergent problems
in a reasonable way? And can we deal algorithmically with possibly infinite
computation paths without resorting to mechanisms that ensure that all
computations are finite?

That last question touches upon the challenging field of investigating the 
termination behavior of $\GDL$ programs in detail, and without restricting
ourselves to a setting where all computations are finite. In general, whether
the execution of a $\GDL$ program terminates may depend on the random choices
that are made during the computation. That is, some paths in the chase tree may
be terminating, while others are not. This is fine, as long as, for example
the probability of the program terminating is $1$. It is an open research 
question to build a thorough understanding of the probabilistic termination of
$\GDL$ programs.

Finally, it is not known whether (and if so, in which cases) the output
probabilistic database of a generative Datalog program admits a concise
representation allowing for the output to be queried effectively. At the 
moment, a generative Datalog program should be thought of as a representation 
of a probabilistic database itself, in particular with the option to have
infinitely many possible worlds. In case of termination, we can always
approximate the output through Monte Carlo sampling, by letting the program run
multiple times. This can be extended to query answers by querying the samples.
The properties of this option are yet to be explored.

\begin{acks}
	We are grateful to Benny Kimelfeld for bringing our attention to \GDL{} and 
motivating the questions answered therein. Further thanks go to Marcel Hark 
for discussions regarding termination.

This research is partially funded by Deutsche Forschungsgemeinschaft (DFG,
German Research Foundation) under grants GR 1492/16-1 and GRK 2236 (UnRAVeL),
and by the European Research Council (ERC) under Advanced Grant 787914
(\mbox{FRAPPANT}).
\end{acks}

\bibliographystyle{plainurl}

\appendix
\section{Background Results from Measure Theory}\label{app:meas}
This section is intended to extend \cref{ssec:mt} by some well-known results.
They can accordingly be found in the literature \cite{Kallenberg2002,Srivastava1998}.

\subsection{Measurability of Functions and Sets}

The following statement says that collections of measurable functions yield a
function that is measurable with respect to the product $\sigma$-algebra.

\begin{fact}[{\cite[Lemma 1.8, p. 5]{Kallenberg2002}}]\label{fac:families}
	If $(\XX,\XXX)$ and $(\XX_i,\XXX_i)$ are measurable spaces (for $i$ in some 
	index set $I$) and $f_i \from \XX \to \XX_i$ is measurable for all $i\in I$, 
	then $f \from \XX \to \prod_{i\in I}\XX_i \with x \mapsto	(f_i(x))_{i\in I}$ 
	is $\big(\XXX,\bigotimes_{i\in I}\XXX_i\big)$-measurable.
\end{fact}

The next two results are concerned with the measurability of certain kinds of
integration maps.

\begin{fact}[{\cite[Lemma 1.41(i), p. 21]{Kallenberg2002}}]%
	\label{fac:kernelandfunction}%
	Let $\mu$ be a stochastic kernel from $\XX$ to $\YY$ and let $f \from \XX
	\times\YY \to \RR_{\geq 0}$ be measurable. Then
	\begin{equation*}
		\XX\to\RR_{\geq 0} \with x \mapsto \int f(x,\placeholder) 
		\:\d{\mu(x,\placeholder})
	\end{equation*}
	is measurable.
\end{fact}

\begin{fact}[{\cite[Lemma 1.26, p. 14]{Kallenberg2002}}]\label{fac:measint}
	Let $(\XX, \XXX)$ and $(\YY, \YYY)$ be measurable spaces, $\mu$ a
	$\sigma$-finite measure on $\XX$ and $f \from \XX \times \YY \to \RR_{\geq
	0}$ a measurable function. Then
	\begin{enumerate}
		\item the \emph{$y$-section} $f(\placeholder,y) \colon \XX \to \RR_{\geq 
			0}$ of $f$ is $(\XXX, \Borel(\RR_{\geq 0}))$-measurable for all $y 
			\in \YY$, and
		\item the function $y \mapsto \int f(x,y) \:\mu(\d{x})$ is 
			$(\YYY, \Borel[0,1])$-mea\-sur\-able.\qedhere
	\end{enumerate}
\end{fact}

If we have a measurable function between two standard Borel spaces, then the
image of measurable sets needs not to be measurable in general, the standard
example perhaps being projection functions (see \cite[Proposition 4.1.1 and 
Theorem 4.1.5]{Srivastava1998}). Given certain conditions however, measurable 
sets have measurable images under measurable functions:

\begin{fact}[{\cite[Theorem 4.5.4, p. 153]{Srivastava1998}}]\label{fac:measinjectionimages}
	Let $(\XX,\XXX)$ and $(\YY,\YYY)$ be standard Borel, $\XXXX \in \XXX$ and let
	$f \from \XXXX \to \YY$ be an injective, $(\XXX\restriction_{\XXXX},\YYY)$-measurable 
	function. Then $f(\XXXX) \in \YYY$.
\end{fact}

Finally, we come back to the multifunctions of \cref{sssec:mf} and explicitly
state the theorem of Kuratowski and Ryll-Nardzewski on the existence of
measurable selections:

\begin{fact}[Kuratowski and Ryll-Nardzewski \cite{Kuratowski+1965}, see 
	{\cite[Theorem 5.2.1]{Srivastava1998}}]\label{fac:KRN}
	Let $(\XX, \XXX)$ be a measurable space and let $(\YY, 
	\Borel(\YY))$ be standard Borel. Then every closed-valued
	$\XXX$-measurable multifunction $M \from \XX \toto 
	\YY$ has a $(\XXX,\Borel(\YY))$-measurable selection
	$s \from \XX \to \YY$.
\end{fact}

\subsection{Identities for Integration}
If $\mu$ is a measure and $f$ a measurable function, then $f\cdot\mu \coloneqq
\nu$, defined by $\nu(\XXXX) = \int_{\XXXX} f \d\mu$ is a measure. The following
chain and substitution rules are the main tools to establish statements
regarding the equality of transformed measures.

\begin{fact}[{Chain Rule, cf. \cite[Lemma 1.23, p. 12]{Kallenberg2002}}]\label{fac:chainrule}
	Let $(\XX,\XXX,\mu)$ be a measure space and $f \from \XX \to \RR$ and $g
	\from \XX \to \RR_{\geq 0}$ be measurable functions. Let $\nu\coloneqq f\cdot
	\mu$. Then, if either of the following integrals exists (i.\,e. is finite),
	it holds that 
	\begin{equation*}
		\int_{\XX} f \cdot g \:\d\mu = \int_{\XX} g \:\d\nu\text.\qedhere
	\end{equation*}
\end{fact}

\begin{fact}[{Substitution, cf. \cite[Lemma 1.22, p.12]{Kallenberg2002}}]\label{fac:substitution}
	Let $(\XX,\XXX)$ and $(\YY,\YYY)$ be measurable spaces and $\mu$ a measure
	on $(\XX,\XXX)$. Let $f \from \XX \to \YY$ and $g \from \YY \to \RR$ be
	measurable. Then, if either of the following integrals exists (i.\,e. is
	finite), it holds that
	\begin{equation*}
		\int_{\XX} g \after f \:\d\mu = \int_{\YY} g \:\d{(\mu\after f^{-1})}
	\end{equation*}
	where $\mu\after f^{-1}$ is the push-forward measure of $\mu$ along $f$ on
	$(\YY,\YYY)$.
\end{fact}

\subsection{Existence of Markov Processes}\label{app:markov}
\begin{fact}[{Existence of Markov Processes, Kolmogorov, cf. \cite[Theorem~8.4]
	{Kallenberg2002}}]\label{fac:kolmogorov}
	Let $( \XX, \XXX )$ be a standard Borel space, $\mu_0$ a 
	probability measure on $(\XX, \XXX)$ and $(\mu_i)_{i \geq 
	1}$ a family of stochastic kernels $\mu_i \from \XX \times
	\XXX \to [0,1]$ for $i\geq 1$. Then there exists a Markov process 
	$\xi$ (with time scale $\N$ and paths in $\prod_{i=0}^{\infty} 
	\XX$) with initial distribution $\mu_0$ and transition kernels 
	$\mu_i$.
\end{fact}

\end{document}